\documentclass[11 pt]{article}

\usepackage{amsfonts,amssymb}
\usepackage{color}
\usepackage{amsthm, amssymb, latexsym, amsmath}
\usepackage[margin=0.9 in]{geometry}

\newtheorem{theorem}{Theorem}[section]

\newtheorem{lemma}[theorem]{Lemma}

\newtheorem{proposition}[theorem]{Proposition}
\newtheorem{corollary}[theorem]{Corollary}

\newtheorem{remark}{Remark}

\numberwithin{theorem}{section} \numberwithin{equation}{section}

\newcommand{\nc}{\newcommand}

\nc{\be}{\begin{equation}} \nc{\la}{\label} \nc{\ba}{\begin{array}}
\nc{\ea}{\end{array}} \nc{\bs}{\begin{split}} \nc{\es}{\end{split}}

\newcommand{\R}{\mathbb{R}}
\newcommand{\Z}{\mathbb{Z}}

\newcommand{\N}{{\mathbb N}}

\newcommand{\cA}{{\cal{A}}}

\newcommand{\imsgap}{{\gamma}}

\newcommand{\al}{{\alpha}}
\newcommand{\del}{{\delta}}

\newcommand{\s}{{\sigma}}

\nc{\G}{\Gamma} \nc{\g}{\gamma} \nc{\Omt}{\tilde{\Omega}}
\nc{\ta}{\tau} \nc{\w}{\omega} \nc{\io}{\iota} \nc{\h}{\theta}
\nc{\Si}{\Sigma}

\renewcommand{\H}{H}
\newcommand{\E}{E(y)}
\newcommand{\Einfty}{E(\infty)}

\newcommand{\hb}{\hat{b}}
\newcommand{\hc}{\hat{c}}

\nc{\ra}{\rightarrow} \nc{\ran}{\rangle} \nc{\lan}{\langle}

\nc{\bP}{\bar{P}} \nc{\bQ}{\bar{Q}} \nc{\bL}{\bar{L}} \nc{\1}{{\bf
1}}

\nc{\p}{\partial}
\newcommand{\DETAILS}[1]{}

\newcommand{\supp}{\operatorname{supp}}

\newcommand{\Ran }{\operatorname{Ran}}

\newcommand{\ls}{\lesssim}

\newcommand{\Rbkl}{R_{B_k B_l}^{\s,\bot}}
\newcommand{\Rbij}{R_{B_i B_j}^{\s,\bot}}
\newcommand{\Rakl}{R_{A_k A_l}^{\s,\bot}}
\newcommand{\Raij}{R_{A_i A_j}^{\s,\bot}}
\newcommand{\Rbkld}{R_{B_k B_l,\delta}^{\s,\bot}}
\newcommand{\Rbkldo}{R_{B_k B_l,\delta_0}^{\s,\bot}}

\begin{document}
\title{Remainder estimates for the Long Range Behavior of the van der Waals interaction energy}

\author{Ioannis Anapolitanos \thanks{Dept.~of Math.,
Univ. of Stuttgart, Stuttgart, Germany; Supported by the German
Science Foundation under Grant GR 3213/1-1.}}
\bigskip

\bigskip

\bigskip

\bigskip
\bigskip

\maketitle

\begin{abstract}
   The van der Waals-London's law, for a collection of atoms at large separation, states that their interaction energy is pairwise attractive and decays proportionally to one over their distance to the sixth. The first rigorous result in this direction
   was obtained by Lieb and Thirring \cite{LT}, by proving an upper bound which confirms this law.
   Recently the van der Waals-London's law was proven under some assumptions by I.M. Sigal and
   the author \cite{AS}. Following the strategy of \cite{AS} and reworking the approach appropriately,
    we prove estimates on the remainder of the interaction energy.
   Furthermore, using an appropriate test
  function, we prove an upper bound for the interaction energy, which is sharp to leading order.
  For the upper bound, our assumptions are weaker, the remainder estimates stronger and the proof is simpler.
  The upper bound, for the cases it applies, improves considerably the upper bound of Lieb and Thirring.
  However, their bound is much more general. Here we consider only spinless Fermions.
\end{abstract}


\section{Introduction}

  The van der Waals forces are forces between atoms or molecules.
They are much weaker than ionic or covalent bonds. The physicist J.D. van der Waals
discovered them during his effort to formulate an equation of state of gases that
is compatible with experimental measurements (see \cite{vdW1} and \cite{vdW2}). These forces play a
fundamental role in quantum chemistry, physics and material
sciences. Due to them, for instance, water
condenses from vapor. They force gigantic molecules like enzymes,
proteins, and DNA into the shapes required for biological activity.
They explain why diamond, consisting of carbon atoms connected with
covalent bonds only, is a much harder material than graphite, which
consists of layers of carbon atoms that attract each other through
van der Waals forces (see \cite{C}).

         We begin with a mathematical formulation of the problem.
  We consider a system of $M$ interacting atoms with nuclei fixed
  at $y_1,\dots,y_M \in \R^3$, respectively, and described
by the Hamiltonian
\begin{equation}\label{Hy}
H^N(y)=\sum_{i=1}^{N} \bigg(- \Delta_{x_i}-\sum_{j=1}^{M}
\frac{e^2 Z_j}{|x_i-y_j|}\bigg)+\sum_{i<j}^{1,N}
\frac{e^2}{|x_{i}-x_{j}|}+\sum_{i<j}^{1,M} \frac{e^2 Z_{i} Z_{j}}{|y_{i}-y_{j}|}.
\end{equation}
 Here $N$ is the total number of electrons, $x_i,y_i\in \R^3$ denote
the coordinates of the electrons and the nuclei, respectively,
$y=(y_1,\dots, y_M)$,  $-e $ is the electron charge, and $Z_j \in \mathbb{N}$
 is the atomic number of the $j$-th nucleus.
The notation $\sum_{i<j}^{1,N}$ means that $i,j$ are summed from $1$
to $N$ for all values $i<j$. The scaling has been chosen so that the mass
of the electrons is $\frac{1}{2}$, and therefore $-\Delta_{x_i}$
is the operator of kinetic energy of the electron with coordinate $x_i$.
 We consider a system of atoms (an atom has total charge zero)
so we must have $\sum_{j=1}^M Z_j=N$. The Hamiltonian $H^N(y)$
arises from the standard full Hamiltonian of the system by fixing
the positions of the nuclei and neglecting their kinetic energies.
This is an approximation called the Born-Oppenheimer approximation,
and $H^N(y)$ is called the Born-Oppenheimer Hamiltonian. The
Born-Oppenheimer approximation relies on the fact that the nuclei
are much heavier than the electrons. We refer to \cite{AS} and
references therein for a discussion of the Born-Oppenheimer
approximation and the Born-Oppenheimer Hamiltonian \eqref{Hy}. The operator
$H^N(y)$ acts on $\otimes_1^N L^2(\mathbb{R}^3)$. In this article we study spinless electrons.
This has to be taken into account in the mathematical formulation of the problem.
 Let $n \in \mathbb{N}$ and $S_n$ be the
 permunation group of $\{1,\dots,n\}$. For any
 $\pi \in S_n$ we define
  the unitary operator $T_{\pi}$ on $\otimes_1^n L^2(\mathbb{R}^3)$ given by
\begin{equation}\label{def:Tpi}
(T_{\pi} \Psi)(x_1,\dots,x_n)=\Psi(x_{\pi^{-1}(1)},\ldots,
x_{\pi^{-1}(n)}).
\end{equation}
Then
\begin{equation}\label{Qdef}
Q_n:=\frac{1}{n!} \sum_{\pi \in S_n} \text{sgn}(\pi) T_{\pi},
\end{equation}
is the orthogonal projection onto the space $\wedge_1^n L^2(\R^3)$ of
antisymmetric functions with respect to permutations of coordinates.  Furthermore, we define
\begin{equation*}
H^{N,\sigma}(y):= H^N(y) Q_N
\end{equation*}
and let
\begin{equation}\label{def:GSE}
E(y):=\inf \sigma(H^{N,\sigma}(y)|_{\Ran  Q_N}),
\end{equation}
 be the ground state energy of the system. The restriction $H^{N, \sigma}(y)|_{\Ran  Q_N}$ onto
 the range of $Q_N$ makes sense, because $H^N(y)$ commutes with $Q_N$. We note
 that in this paper all operators will be considered as acting on the entire $L^2$ spaces, unless
 an explicit restriction is written.
For all  $i \in \{1,\dots,M\}$ and $n_i \in \mathbb{Z}$ with $n_i \leq Z_i$, we define
\begin{equation}\label{def:Hm}
H_{i,n_i}:=\sum_{j =1}^{Z_i-n_i} \left(- \Delta_{x_j}- \frac{e^2
Z_i}{|x_j|}\right)+ \sum_{j<k}^{1, Z_i-n_i}\frac{e^2}{|x_j-x_k|}.
\end{equation}
In simple words, $H_{i,n_i}$ is the Hamiltonian of an ion with atomic number $Z_i$
and total charge $e n_i$. We can define $H_{i,n_i}^\sigma$ similarly as
$H^{N, \sigma}(y)$, namely
\begin{equation}\label{def:Hinis}
H_{i,n_i}^\sigma=H_{i,n_i} Q_{Z_i-n_i}.
\end{equation}
 Furthermore, we define
\begin{equation}\label{def:Eini}
E_{i,n_i}:=\inf \sigma(H_{i,n_i}^\sigma|_{\Ran Q_{Z_i-n_i}}).
\end{equation}
The following well known lemma is a corollary of the HVZ and Zhilin's theorems,
as we shall explain in Section \ref{importantprop}.
\begin{lemma}\label{lem:negGSE}
 For all $i \in \{1,\dots,M\}$ and all $n_i \in \mathbb{Z}$ with $n_i< Z_i$, we have that
 \begin{equation}
  E_{i,n_i}<0
 \end{equation}
 and 0 lies on the essential spectrum of $H_{i,n_i}^\s|_{\Ran Q_{Z_i-n_i}}$.
 Moreover, if $n_i \geq 0$, then $E_{i,n_i}$ lies on the discrete spectrum of $H_{i,n_i}^\s|_{\Ran Q_{Z_i-n_i}}$.
\end{lemma}
From Lemma \ref{lem:negGSE} and the fact that $H_{i,n_i}^\sigma|_{(\Ran Q_{Z_i-n_i})^\bot}=0$,
it follows that the restriction onto $\Ran Q_{Z_i-n_i}$ does not change the discrete
spectrum and the corresponding eigenspaces. More precisely:
\begin{corollary}\label{cor:restrictionaway}
 For all $i \in \{1,\dots,M\}$ and all $n_i \in \mathbb{Z}$ with $n_i< Z_i$, we have that
 \begin{equation}
E_{i,n_i}=\inf \sigma(H_{i,n_i}^\sigma).
\end{equation}
Moreover the discrete spectrum of $H_{i,n_i}^\sigma$
is the same as that of $H_{i,n_i}^\sigma|_{\Ran Q_{Z_i-n_i}}$ (it might be empty) and the corresponding
eigenspaces are the same.
\end{corollary}
This corollary will allow us to remove the restrictions onto the antisymmetric spaces, which will
be very convenient for the proof.

Before stating the van der Waals-London's law, we
formulate and discuss a property of many-body systems playing an
important role below.
\begin{itemize}
\item[(E)] For all $i,j \in \{1,\dots,M\}$ with $i \neq j$ we assume that:
if $m,n \in \N \cup \{0\}$ and $l \in \N$ so that $m+l \leq Z_i$,
then we have that
\begin{equation}\label{Prop(E)}
E_{i,m}+E_{j,-n}< E_{i,m+l}+E_{j,-n-l}.
\end{equation}
\end{itemize}
\medskip
   The inequality $m+l \leq Z_i$ is imposed, because $H_{i,m+l}$ is not defined otherwise.
   It means physically that a positive ion can not have a charge that is bigger
    than the one of its nucleus. The physical meaning of \eqref{Prop(E)} is that if we have
  a system of a non-positive ion and a non-negative ion that are infinitely far from each other (no interaction), then it costs
   energy to transfer electrons from the non-negative ion to the non-positive one. So far every experimental
 measurement verifies Property (E). For a detailed discussion of this fact we refer to
  \cite{AS} (pages 6-7). However, theoretically, Property (E) is an
 open problem except for the case of a system of hydrogen atoms (see Proposition \ref{prop:ConditionEhyd}
 below). The hydrogen atom is an atom with atomic number $1$.

        A simple induction argument on the number of atoms shows that Property (E) implies
\begin{itemize}
\item [(E')]  
\qquad $\sum_{i=1}^M E_{i,0} < \sum_{i=1}^M  E_{i,n_i}, \quad \forall (n_1, \dots, n_M): \sum_i n_i=0,\ \sum_i |n_i|> 0,\ n_i \leq Z_i$.    
\end{itemize}
The physical meaning of Property (E') is that if we have a system  of atoms that are infinitely far from each other (no interaction),
then it costs energy to make them a system of ions by transferring electrons between them.
 It was proven in \cite{AS} (see page 30) that
Property (E') is a necessary Condition for the van der Waals-London's law to
hold. The proof of the necessity is based on the fact that ions
interact with each other through Coulomb interaction (behaving like $|y_i-y_j|^{-1}$). This is
inconsistent with the inverse sixth power van der Waals-London's law (according to which the
 interaction behaves like $|y_i-y_j|^{-6}$).
 For simplicity, we will assume Property (E') in the discussion below.
For all $i \in \{1,\dots,M\}$ we define
\begin{equation}\label{Eidef}
E_i=E_{i,0},
\end{equation}
where $E_{i,0}$ was defined in \eqref{def:Eini}. In other words,
$E_i$ is the ground state energy of the atom with atomic number $Z_i$.
From Lemma \ref{lem:negGSE} it follows that
\begin{equation}\label{eqn:Eiless0}
 E_i<0, \text{ } \forall i \in \{1,2,\dots,M\}.
\end{equation}
 The interaction energy $W(y)$ of the system
is defined as
\begin{equation}\label{interene}
W(y):=E(y)-\Einfty,
\end{equation}
where $E(y)$ was defined in \eqref{def:GSE} and
\begin{equation}\label{Einftydef}
\Einfty= \sum_{i=1}^M E_i.
\end{equation}
The quantity $\Einfty$ can be roughly understood as the ground energy of
the system when the atoms are infinitely far from each other.

     It was expected after  London (see \cite{Lo}), that $W(y)$ is a sum of pair
interactions which are attractive and decay at infinity
as $-|y_i-y_j|^{-6}$. In other words, one expects that
\begin{equation}\label{vdWlaw}
W(y)=-\sum_{i<j}^{1,M}\frac{e^4
\sigma_{ij}}{|y_i-y_j|^6}+O\left(\sum_{i<j}^{1,M} \frac{1}{|y_i-y_j|^7}\right),
\end{equation}
 provided that
 \begin{equation}\label{Rdef}
 R:=\min\{|y_i-y_j|: 1\leq i <j \leq M\}
 \end{equation}
  is large enough. Here $\sigma_{ij}$ are positive constants depending  on the atomic
numbers $Z_i, Z_j$. An upper bound of the form
 \begin{equation}\label{LTupper}
W(y) \leq -\sum_{i<j}^{1,M}\frac{ e^4 C_{ij}}{|y_i-y_j|^6},
 \end{equation}
where $C_{ij}$ are some positive constants, was proven by Lieb and Thirring \cite{LT} in 1986, using an
 intricate test function. The upper bound was proven without any assumptions and it holds for a system of molecules
 as well. A first rigorous proof of \eqref{vdWlaw}
 was given by the author and Sigal (see Theorem 1.1 in \cite{AS}) assuming Condition (D) below and Property
 (E'). There, a related result was proven for the case of Fermions with
 spin as well. In \cite{AS} there was no information on how
 large $R$ should be, so that the remainder
 is small relative to the leading term. The error terms involved sums
 over sets whose cardinality was growing extremely fast (at the rate $N^N$)
 in the number of electrons. It was claimed,
 however, that, after reworking the approach, the sums of the error terms can be controlled
  appropriately, to obtain error estimates that grow
   much slower than the cardinality of these sums. This is the first goal of this work.
 We estimate the remainder, up to constants depending only on
\begin{equation}\label{def:maxZj}
  Z=\max\{Z_j, j \in \{1,2,\dots,M\}\}.
  \end{equation}
  Our point is that these constants depend only on the kinds
 of atoms involved in the system and not on how many there are.
  The second goal of this work is to provide an upper
  bound for $W(y)$ under weaker assumptions, by using an appropriate test
function. The upper bound is sharp in the leading order, and therefore, when it applies,
it improves considerably the upper bound of Lieb and Thirring in \cite{LT}. We note once more, however,
that the upper bound in \cite{LT} holds without any assumptions and for a system of molecules as well.
 We refer to \cite{MS} for the related retarded van der Waals potential.

Before we state our main theorems, we state another condition which
we will need. We define
\begin{equation}\label{def:His}
H_i=H_{i,0}, H_i^\sigma:=H_{i,0}^\sigma,
\end{equation}
 where $H_{i,0}^\sigma$ was defined in \eqref{def:Hinis}.
 In simple words, $H_i$ is the Hamiltonian
of the atom with atomic number $Z_i$ and with nucleus at $0$.
From \eqref{def:His}, \eqref{Eidef} and Corollary \ref{cor:restrictionaway}
it follows that
\begin{equation}\label{eqn:EiHs}
E_i=\inf \s(H_i^\sigma), \text{ } \forall i \in \{1,\dots,M\}.
\end{equation}
It is well known that $H_i^\sigma|_{Ran Q_{Z_i}}$ has a ground
state for all $i \in \{1,\dots,M\}$ (see Theorem \ref{zysl} below). It follows then, from Corollary
\ref{cor:restrictionaway}, that $H_i^\sigma$ has a ground state for all $i \in \{1,\dots,M\}$.
We are going further to assume
\begin{itemize}
\item[(D)] For all $i \in \{1,\dots,M\}$ the ground state energy of
 $H_i^\sigma$ is non-degenerate.
\end{itemize}
Assuming condition (D), we define $\phi_{i}$ to be the ground state
of $H_i^\sigma$, for all $i \in \{1,\dots,M\}$.  Throughout the text we will
 always assume that $\|\phi_i\|=1$ for all $i \in \{1,\dots,M\}$. Of course
 the ground state is unique up to a constant phase factor $e^{i \theta}$.
 However, this factor can be chosen arbitrarily and the results and the proofs do not depend on this choice.
 We are now going to define the constants $\sigma_{ij}$. To this end,
 we need to introduce some notation that will be useful later on as well.
 For all $i,j \in \{1,2,\dots,M\}$, with $i \neq j$, we define
\begin{equation}\label{def:Hijs}
 H_{ij}^{\sigma}:= H_i^\sigma\otimes I^{Z_j}+ I^{Z_i} \otimes H_j^\sigma,
\end{equation}
where, for $n \in \mathbb{N}$, $I^n$ is the identity acting on $\otimes_1^n L^2(\mathbb{R}^3)$.
Then $H_{ij}^{\sigma}$ has the unique ground state $\phi_i \otimes \phi_j$.
We further define
\begin{equation}
 H_{ij}^{\sigma, \bot}:=H_{ij}^\sigma (1-P_{\phi_{i} \otimes \phi_{j}}),
\end{equation}
where $P_{\phi_{i} \otimes \phi_{j}}$ denotes the orthogonal projection onto the ground state $\phi_{i} \otimes \phi_{j}$.
From \eqref{eqn:EiHs} and \eqref{def:Hijs} it follows that the ground state energy of $H_{ij}^\s$ is $E_i+E_j$. By
\eqref{eqn:Eiless0} we have that $E_i+E_j<0$. Moreover, $E_i+E_j$ belongs to the discrete spectrum of $H_{ij}^\s$ by
\eqref{def:Hijs} and Corollary \ref{cor:restrictionaway} (see \cite{RSI} Theorem VIII.33). It follows, therefore, that there exists a constant
$c>0$ so that for all $i,j \in \{1,\dots,M\}$ with $i<j$ we have
\begin{equation}\label{ineq:Hijsbot}
 H_{ij}^{\sigma,\bot}-E_i-E_j \geq c, \text{ on } \otimes_{1}^{Z_i+Z_j} L^2(\mathbb{R}^3).
\end{equation}
Therefore, the resolvent
\begin{equation}\label{def:rijsbot}
R_{ij}^{\sigma, \bot}:=(H_{ij}^{\sigma,\bot}-E_i-E_j)^{-1},
\end{equation}
 is defined, bounded and positive on the entire $\otimes_{1}^{Z_i+Z_j} L^2(\mathbb{R}^3)$.
 For each vector $v \in \mathbb{R}^3$ we define the function
\begin{equation}\label{def:fij}
f_{ij,v}(z_1,\dots,z_{Z_i+Z_j})=\sum_{i=1}^{Z_i} \sum_{j=Z_i+1}^{Z_i+Z_j} \left(z_i \cdot z_j-3 (z_i \cdot v)(z_j \cdot v) \right)
\end{equation}
and the number
\begin{equation}\label{sigmaijvdef}
\sigma_{ij}(v):= \langle f_{ij,v} \phi_i \otimes \phi_j, R_{ij}^{\sigma,\bot} f_{ij,v}\phi_i \otimes \phi_j \rangle.
\end{equation}
We will prove in Section \ref{sec:sigma}:
\begin{lemma}\label{lem:sigma}
For all $i,j \in \{1,2,\dots,M\}$, with $i<j$, and for all unit vectors $v \in \mathbb{R}^3$,
  $\sigma_{ij}(v)$ is positive and does not depend on the choice of $v$.
\end{lemma}
We can therefore define
\begin{equation}\label{def:sigmaij}
\sigma_{ij}:=\sigma_{ij}(v) \text{ for some unit vector } v \in \mathbb{R}^3.
\end{equation}
\smallskip
  We are now ready to state our main results. As before, $M$ is the
  number of atoms, $N$ is the number of electrons, $R$ is defined
  in \eqref{Rdef} and $Z$ is defined in \eqref{def:maxZj}.
\begin{theorem}\label{thm:vdW-maxspin}
(i) Assume Condition (D) and Property (E).
 Then, there exist positive constants
$C_1, C_2, C_3$, depending only on $Z$, so that if
$R \geq C_1 N^{\frac{4}{3}}$, then
 \begin{equation}\label{verbesserung}
 \bigg|W(y)+\sum_{i<j}^{1,M} \frac{e^4 \sigma_{ij}}{|y_i-y_j|^6}\bigg|
    \leq  C_2 \left( \sum_{i<j}^{1,M} \frac{1}{|y_i-y_j|^7}+\frac{ M^4}{R^9}\big(1+  N^Z e^{-C_3 R}\big)
    \right).
 \end{equation}
 \newline
 (ii) Assume Condition (D) and Property (E'). Then the same conclusion as in part (i) holds with the constants
 $C_1, C_2$ replaced by two constants $C_1', C_2'$ which are also positive and depend only on $Z$.
\end{theorem}
\begin{theorem}\label{thm:upper}
Assume Condition (D). Then, there exist positive constants $C_4, C_5, C_6$,
depending only on $Z$, such that if $R \geq C_4
M^\frac{1}{3}$, then
\begin{equation}\label{bnd:upper}
W(y) \leq- \sum_{i<j}^{1,M} \frac{e^4 \sigma_{ij}}{|y_i-y_j|^6}+C_5
\left(\sum_{i<j}^{1,M} \frac{1}{|y_i-y_j|^7}+
\frac{M^3}{R^9}+ M^5 e^{-C_6 R}\right).
\end{equation}
\end{theorem}
\begin{remark}
The constants $C_1', C_2'$ in part (ii) of Theorem \ref{thm:vdW-maxspin}
are in principle much larger that the constants $C_1, C_2$. Our strategy of proving
part (i) provides better constants and is simpler.
 So assuming that one proves Property (E) for a system of atoms,
then one should follow the strategy of proof of part (i) to obtain better bounds.
This will be made clear in the proof of Proposition \ref{prop:Eimpliesgap} below.
\end{remark}
\begin{remark}
 Similar theorems hold in the case that we do not take the Fermionic
 statistics into account. In this case, Condition (D) follows from the positivity
improving property  of $e^{-\beta H_i},  \beta >0, i \in \{1,\dots,M\}$, where $H_i$ was defined
in \eqref{def:His}, and from Perron-Frobenious theory (see
for example \cite{RSIV} Chapter XIII Section 12). Therefore, in this case, the assumption of
Condition (D)  can be omitted. In particular, the upper bound \eqref{bnd:upper} holds,
in this case, with no assumptions.
\end{remark}
\begin{remark}
By the previous remark, Condition (D) holds for a system of hydrogen atoms ($Z_i=1$ for all $i \in \{1,\dots,M\}$)
independently of statistics, since the hydrogen atom has only one
electron. As we will discuss below (see Section \ref{sec:sigma}), Property (E) holds in this case too.
Therefore, the conclusions of Theorems \ref{thm:vdW-maxspin} and
\ref{thm:upper} hold for a system of hydrogen atoms with no
assumptions.
\end{remark}
\begin{remark}
 It is important that the constants
 $C_1, C_1', C_2, C_2', C_3, C_4, C_5, C_6$ in Theorems \ref{thm:vdW-maxspin} and \ref{thm:upper}
  do not depend on the number of atoms $M$ but
  only on $Z$, or in simple words only on the kinds of atoms involved in the system
  and not on how many there are.
   Note, however, that we have been
   unable to determine how the
constants depend on $Z$. The remainder in Theorem
\ref{thm:upper} is small relative to the leading order provided that $R \geq c
M^{\frac{1}{3}}$, where $c$ again depends on $Z$. In Theorem
\ref{thm:vdW-maxspin}, the assumption $R \geq C_1 N^{\frac{4}{3}}$
ensures that the remainder is small relative to the leading order. Of course if $M$ is large, the assumptions on $R$
 are too strong.
\end{remark}
\begin{remark}
As indicated from the title of the paper, we do not really estimate the force between the atoms but their interaction
energy. The forces $F_j(y):=-\nabla_{y_j} W(y)$, $j \in \{1,\dots,M\}$ have never been rigorously studied,
as far as we know. Here $F_j(y)$ denotes the force that the atom at $y_j$ experiences from
the rest of the system. We conjecture that
\begin{equation*}
F_j(y)=\sum_{i \neq j} \frac{6 \sigma_{ij}(y_i-y_j)}{|y_i-y_j|^8}+O\left(\sum_{i \neq j}\frac{1}{|y_i-y_j|^8}\right), \text{ } \forall j \in \{1,\dots,M\},
\end{equation*}
where $\sigma_{ij}$ are the same constants as in Theorem \ref{thm:vdW-maxspin}. In other words, we expect
that the leading term of the force is given by minus the gradient of the leading term of the interaction energy.
We believe that our methods can be adapted in order to prove such an estimate.
 However, our result does not imply this estimate,
 because the error term in \eqref{verbesserung} could in principle be fast oscillating
  and have a large gradient.
\end{remark}
\begin{remark}
 As it is clear from the statements of Theorems \ref{thm:vdW-maxspin} and \ref{thm:upper},
 the remainder term of the order of one over distance to the seventh is  relevant, but its size relative to the leading order does not
 grow when $M$ grows. The origin of this remainder term
 can be explicitly seen in the proof (see Lemma \ref{lem:Uaagenauer} below). The worse remainder term, however,
 is the one of the order of one over distance to the ninth, because it grows
 relative to the leading term at the rate $M^2$ in Theorem \ref{thm:vdW-maxspin}
and $M$ in Theorem \ref{thm:upper}, when $M$ grows.
\end{remark}
    Theorems \ref{thm:vdW-maxspin} and \ref{thm:upper} describe the interaction energy of the atoms at
a pairwise large separation between them. Note that for small
distances, the interaction energy is repulsive (positive) as follows from the rough estimate
$$H^N(y) \geq -C+\sum_{i<j}^{1,M}\frac{e^2 Z_i Z_j}{|y_i-y_j|},$$ for some
 constant $C$ independent of $y$, implied by the
bound $\frac{e^2 Z_m}{|x_n-y_m|} \leq -\alpha \Delta_{x_n}+\beta,$
valid for any $\alpha>0$ and a corresponding $\beta>0$. For a proof of the last bound we refer to \cite{RSIV} Chapter XIII Section 11.

 Often the
interaction energy for two atoms ($M=2$) is modeled by the
Lennard-Jones potential
$W_{LJ}(y)=\frac{a}{|y_1-y_2|^{12}}-\frac{b}{|y_1-y_2|^6}$, where
the constants $a,b>0$ are determined experimentally. This potential
was originally proposed by J.E. Lennard-Jones in the form $\frac{a}{|y_1-y_2|^{m}}-\frac{b}{|y_1-y_2|^n}$, during his effort
to deduce an appropriate law of dependence of the viscosity of a gas on the temperature  (see \cite{J1}), and to study the equation of state of gases
(see \cite{J2}).

   Our approach for the proof of Theorem \ref{thm:vdW-maxspin} is
based on perturbation theory in the parameter $\frac{1}{R}$, for
which the Feshbach map is used. We follow \cite{AS} closely. Essentially, our new
elements for the proof of Theorem \ref{thm:vdW-maxspin} are in the
proofs of Proposition \ref{prop:Eimpliesgap} and of Lemma \ref{lem:Uentwicklung}.
 The main ideas of these proofs are still
similar to ideas introduced in \cite{AS}, but we substantially
rework the techniques that were introduced there, in order to obtain
stronger estimates of error terms. Parts that are similar to
\cite{AS} will be repeated here, so  that the present work is
self-contained. We shall now sketch the proof of Theorem
\ref{thm:vdW-maxspin} and afterwards the proof of Theorem
\ref{thm:upper}.


\textbf{Brief sketch of the proof of Theorem \ref{thm:vdW-maxspin}.} 
 The main ingredient of the proof is the Feshbach map and the
 Feshbach-Schur method. We refer to \cite{BFS} Section IV for an exposition
 of the method in a general form and with proofs. For purpose
 of simplicity we will state everything in the special form we need.
 Let $\Pi=|\Psi\rangle \langle \Psi|$ be the orthogonal projection onto
a function $\Psi \in \wedge_{1}^N L^2(\R^{3})$, with $\|\Psi\|=1$ (normalized function), and $\Pi^\bot=1-\Pi$.
To simplify the notation we will write $H^\s$ instead of $H^{N,\s}(y)$. We introduce the
notation $\H^{\s,\bot}=\Pi^\bot \H^\s \Pi^\bot$. We recall that $E(y)$ is the ground state energy
of $H^\s|_{\Ran Q_N}$. As we shall see in the proof, $E(y)$ is as well the ground state energy of
$H^\s$, when $R$ is large enough. Therefore, we shall work with $H^\s$ and assume that $R$ is large enough so that $E(y)$ is its ground state energy.
The  Feshbach-Schur method states that if
\begin{itemize}
\item[(a)]  $\Psi \in D(\H^\s)$ (domain of $\H^\s$);
\item[(b)]   The operator $(\H^{\s,\bot}-E(y))$   is invertible;
\end{itemize}
then the Feshbach map
\begin{equation}\label{FP}
F_\Pi(\lambda)=(\Pi H^\s \Pi-V(\lambda))|_{\Ran  \Pi},
\end{equation}
where
\begin{equation}\label{W}
V(\lambda):=\Pi \H^\s \Pi^\bot (\H^{\s,\bot}-\lambda)^{-1} \Pi^\bot
\H^\s \Pi,
\end{equation}
is well defined at $\lambda=E(y)$ and
\begin{equation}\label{FSE}
E(y)= F_\Pi(E(y)).
\end{equation}
Note that since $F_{\Pi}(\lambda)$ is a linear operator on a one-dimensional space, it can be
 identified with a multiplying coefficient. This reduces the
problem of determining  the ground state energy of $H^\s$ to a
scalar nonlinear fixed point problem, because $\Pi$ is a rank one
projection. Besides Equation \eqref{FSE}, we have that the ground
state of $H^\s$ is the normalization of the function
\begin{equation}\label{groundFesh}
\Psi-(\H^{\s,\bot}-E(y))^{-1} \Pi^\bot \H^\s \Psi.
\end{equation}

 Now we outline how we use the Feshbach map in order to prove
Theorem \ref{thm:vdW-maxspin}. For purpose of simplicity of the outline,
some things are stated below heuristically only. They will be made rigorous in the
actual proof. We take $\Psi=\Phi:=\frac{\wedge_{j=1}^M
\phi_{j,y_j}}{\|\wedge_{j=1}^M \phi_{j,y_j}\|}$ where $\phi_{j,y_j}$ is the ground state
of the atom with atomic number $Z_j$ and nucleus at $y_j$. Here $\wedge_{j=1}^M \phi_{j,y_j}$ denotes
the antisymmetric tensor product of the functions $\phi_{j,y_j}$. The wave
function $\Phi$ is an approximation of the ground state of the
system, the error arising from the fact that there are interaction
terms between atoms. The error depends on $\frac{1}{R}$ and becomes
small when $R$ is large. Since $\Phi$ is an approximate ground state
 of the system and Property (D) holds, one intuitively expects that
 when $R$ is large, then
 \begin{equation}\label{stabout}
  (\H^{\s,\bot}-E(y)) \geq d >0,
 \end{equation}
 or, in other words, that when we project out $\Phi$, the resulting operator  $\H^{\s,\bot}$ has a gap above
  the ground state energy $E(y)$ of $\H^\s$.
  We prove such an estimate in Section \ref{Hbotbndseveral} using Property (E') and the IMS localization formula:
  we find an appropriate family $\{J_a, a \in \hat{\mathcal{A}}\}$ of functions, such that
   $\sum_{a \in \hat{\mathcal{A}}} J_a^2=1$ and
    \begin{equation}\label{stabout2}
  (\H^{\s,\bot}-E(y)) \geq \sum_{a \in \mathcal{A}} J_a  (\H^{\s,\bot}-E(y)) J_a-O(\frac{1}{R}).
     \end{equation}
 Due to the stated properties of $J_a$, showing \eqref{stabout} reduces to showing that
 there exists $c>0$ such that
 \begin{equation}\label{stabout1}
 J_a  (\H^{\s,\bot}-E(y)) J_a \geq c J_a^2-O(\frac{1}{R}), \forall a \in \hat{\mathcal{A}}.
 \end{equation}
  The family $\{J_a, a \in \hat{\mathcal{A}}\}$ consists of
functions supported either on a set where  each of the electrons is
close to some nucleus, or on a set where at least one electron is
far from all nuclei. If at least one electron is far from all nuclei
then \eqref{stabout1} is obtained by the HVZ theorem (see Theorem \ref{hvz} below). If
all electrons are close to some nucleus, then this corresponds to a
decomposition of the system into ions/atoms with total charge 0.  If
 the decomposition has only atoms, then
\eqref{stabout1} follows from the fact that $\Pi^\bot$ projects
out their ground states. Property (E') gives \eqref{stabout1},
 if in the decomposition there are ions with nonzero
charges.  From \eqref{stabout} it follows that the Feshbach-Schur method is
applicable and therefore \eqref{FSE} holds.

  In view of \eqref{FSE} we need to estimate $\Pi H^\s \Pi$ and $V(E(y))$.
  Recall that $\Einfty$, defined in \eqref{Einftydef}, is the sum of the ground
 state energies of the atoms. Due to the
 interaction terms between the atoms, the equality $\Pi H^\s \Pi=\Einfty \Pi$ does not hold.
  However, due to Condition (D), which says
that the ground state energy of each atom is non-degenerate, it
turns out that the ground state of each atom has a
 spherically symmetric one-electron density (see Proposition \ref{spherical}). Therefore,
  we can apply Newton's theorem (see \cite{LL} Section 9.7) to show that the error arising
 from the interaction terms  is exponentially decaying in $R$, because the ground states of the
 atoms are exponentially decaying. In other words, we obtain that
 \begin{equation}\label{PHPout}
 \Pi \H^\s \Pi \approx \Einfty \Pi,
 \end{equation}
 where the approximate equality is understood up to an error which
 is  exponentially decaying in $R$.
From \eqref{PHPout}, \eqref{interene}, \eqref{FP} and \eqref{FSE} it follows that
\begin{equation}\label{intenap}
W(y) \approx -V(E(y))|_{\Ran  \Pi},
\end{equation}
so that estimating the interaction energy reduces to estimating
$V(E(y))$. From \eqref{stabout} it follows that $V(E(y))>0$ when $R$ is
large. From this observation and \eqref{intenap} it follows that
the interaction energy is negative. We now sketch how we estimate
$V(E(y))$. If $\Phi$ were the exact ground state of $\H^\s$ then
$\Pi$ would commute with $\H^\s$ and therefore we would have $\Pi
\H^\s \Pi^\bot=0$ and thus $V(E(y))=0$. In this sense, $V(E(y))$
originates from the error of our choice of $\Phi$ as an approximate
ground state of the system. Since the error in our choice of $\Phi$
originates from the fact that we have neglected the interaction terms
between the atoms, it turns out that $\Pi \H^\s \Pi^\bot$ is
proportional to the interaction. If we make a Taylor expansion of the
Coulomb interaction terms between two atoms in powers of their one over their distance,
 it turns out that there
is cancelation in the first two orders, because both atoms are neutral. Thus, the total
interaction is to leading order proportional to $R^{-3}$. Therefore,
$\Pi \H^\s \Pi^\bot \sim R^{-3} +O(R^{-4})$, where $\sim$ has the
unprecise meaning of proportional, which together with \eqref{W} and
\eqref{stabout} implies that
\begin{equation*}
V(E(y)) \sim R^{-6}+O(R^{-7}).
\end{equation*}
The last estimate together with \eqref{intenap} implies the desired
result:
\begin{equation*}
W(y) \sim -R^{-6}+ O(R^{-7}).
\end{equation*}


\textbf{Brief sketch of the proof of Theorem \ref{thm:upper}.} 
We shall sketch the proof of the theorem for the case that we have
two atoms only, and without taking into account the Fermionic
statistics, as this is much simpler than the general case. For
purpose of simplicity, the sketch will not be precise. We decompose
the full Hamiltonian of the system as $H=H_{12}+I_{12}$, where
$H_{12}$ is the sum of the Hamiltonians of the two atoms and
$I_{12}$ has the interaction terms of the atoms. We denote by $\psi$
the ground state of $H_{12}$, so that $H_{12} \psi=\Einfty \psi$.
Let also $P_{\psi}^\bot=1-P_{\psi}$ with $P_{\psi}$ the orthogonal
projection onto $\psi$. The test function we consider is the
normalization of the function $\tilde \psi=\psi- R_{12}^\bot I_{12}
\psi$, where $R_{12}^\bot=(H_{12}P_\psi^\bot-\Einfty)^{-1}$. This
test function can be understood as an approximation of the ground
state of $H$ as given by the Feshbach map. Indeed, up to
antisymmetrization, which in this sketch we ignore, the test
function $\tilde \psi$ originates from the function given in
\eqref{groundFesh} after modifying the resolvent by omitting the
interaction between the atoms.

Since the interaction energy $W(y)$ is the ground state energy of
$H-\Einfty$, we have that
\begin{equation*}
W(y) \leq \frac{1}{\|\tilde \psi\|^2} \langle \tilde \psi,
(H-\Einfty) \tilde \psi \rangle.
\end{equation*}
  Expanding the inner product on the right hand side of the last
  estimate and using the equality $(H-\Einfty)\psi=I_{12} \psi$, we obtain that
\begin{align} \notag
 \langle \tilde \psi, (H-\Einfty) \tilde \psi \rangle=  \langle  \psi, I_{12}  \psi \rangle
 -2 \langle I_{12}  \psi, R_{12}^\bot I_{12}  \psi
 \rangle  + & \langle R_{12}^\bot I_{12}  \psi, I_{12} R_{12}^\bot I_{12}  \psi
 \rangle \\ \label{testexp}  + & \langle R_{12}^\bot I_{12}  \psi, (H_{12}-\Einfty) R_{12}^\bot I_{12}  \psi
 \rangle,
\end{align}
where the last two terms originated from the decomposition
$(H-\Einfty)=(H_{12}-\Einfty)+I_{12}$. From Newton's theorem it
follows that $\langle \psi, I_{12} \psi \rangle \approx 0$, as we
explained in the sketch of the proof of Theorem
\ref{thm:vdW-maxspin}, where $\approx$ means that the error decays
exponentially in the distance $|y_1-y_2|$ of the atoms. Therefore, we also have that $P_{\psi}^\bot
I_{12} \psi \approx I_{12} \psi$, so that in the second line of
Equation \eqref{testexp} we have the simplification
$(H_{12}-\Einfty) R_{12}^\bot I_{12} \psi \approx I_{12} \psi$. As a
consequence, the term in the second line of \eqref{testexp} is equal to $\langle
I_{12} \psi, R_{12}^\bot I_{12}  \psi
 \rangle$, up to an  error that is exponentially decaying in $|y_1-y_2|$.
 With these observations we arrive at
\begin{align}\label{finalexp}
 \langle \tilde \psi, (H-\Einfty) \tilde \psi \rangle \approx
 - \langle I_{12}  \psi, R_{12}^\bot I_{12}  \psi
 \rangle + \langle R_{12}^\bot I_{12}  \psi, I_{12} R_{12}^\bot I_{12}  \psi
 \rangle.
\end{align}
As we explained in the sketch of the proof of Theorem
\ref{thm:vdW-maxspin}, if we make a Taylor expansion of the
interaction terms $I_{12}$ between the atoms in powers of one over their distance, it turns out that
$I_{12}$ is to leading term of the order $|y_1-y_2|^{-3}$. If we
drop the higher order terms of $I_{12}$, the term $- \langle I_{12}
\psi, R_{12}^\bot I_{12} \psi
 \rangle$ gives exactly the term
 $-\frac{\sigma_{12}}{|y_1-y_2|^6}$. The remainder of this term can be easily proven
 to be of the order $|y_1-y_2|^{-7}$. Due to the fact that $I_{12}$
 appears three times in the last term of \eqref{finalexp}, it turns
 out that this term is at most of the order $|y_1-y_2|^{-9}$. Making
 the last statement precise is harder, because one of the three $I_{12}$
 terms is not multiplied with the exponentially decaying function $\psi$.
 For this reason, we will push exponential weights through the
 resolvent $R_{12}^\bot$ using boosted Hamiltonians.
 Finally, observing that $\|\tilde \psi\| \leq 1+O(|y_1-y_2|^{-3})$, the theorem follows.

\begin{remark}
Of course Property (D) is a very restrictive assumption. As far as we know, Property (D) remains an open question
for all atoms with only exception the hydrogen atom. Our proof of Theorem \ref{thm:vdW-maxspin} depends heavily
on it. Regarding the method of the proof of Theorem \ref{thm:upper}, it can be adapted in a situation as general as in \cite{LT}.
Lieb and Thirring considered a system of interacting molecules.
They proved that if their separation is large enough, then there exist orientations
of the molecules so that the upper bound \eqref{LTupper} holds. Their strategy was
to construct a test function and to average over all possible orientations. Using our test function in this situation
but arising from a ground state of each molecule
(as opposed to "the ground state") and following the strategy of \cite{LT},
 one obtains a bound of the form \eqref{LTupper}, for some orientations of the molecules. The main reason is that the term $\langle  \psi, I_{12}  \psi \rangle$ in \eqref{testexp} vanishes
after averaging over all orientations, by Newton's theorem.
  We conjecture that the constants $C_{ij}$ are better than in \cite{LT} in this case as well, because the test functions are approximate ground states.
 We note that in the general situation
of a system of molecules one can not expect attraction for every orientation of them. For example, if the molecules have
 dipole moments (see \cite{Le} Definition 2), then the leading term of the interaction energy is expected to be
  proportional to one over their distance to the three. In this case, the sign of the leading term depends on the orientations
 of the molecules.
\end{remark}
\bigskip


The paper is organized as follows. In Section \ref{Sec:prelim} we
discuss preliminaries of quantum many-body systems.
In Section \ref{sec:sigma} we prove Lemma \ref{lem:sigma} and Property (E) for a system
of hydrogen atoms. In Section \ref{sec:reformulation}  we reformulate Theorem \ref{thm:vdW-maxspin}
 in terms of two propositions and two lemmas,
 which we then prove in Sections \ref{sec:setup} and \ref{Hbotbndseveral}.
 In Section \ref{proof:thmupper} we prove Theorem
\ref{thm:upper}.

\bigskip

\textbf{Notation.} We collect here general notation used in this
paper. In what follows,
\begin{itemize} \item $M$ is always the number of the nuclei, $N$ is always the number of the electrons, $R$ is the one
defined in Equation \eqref{Rdef} and $Z$ is the one defined in Equation \eqref{def:maxZj}. We will
always assume that $R>0$.

 \item  For any Banach space $X$, we
denote $B(X):=\{f:X \rightarrow X: f$ linear and bounded$\}$.

\item For an
operator $A$, the symbols $\s(A)$ and $\s_{\textrm{ess}}(A)$ stand
for the spectrum and the essential spectrum of $A$, respectively.

 \item   $C$ and $c$ will denote positive constants that depend only on
 $Z$. They are independent of $R$
 and the number of atoms $M$, but they might change from one equation to the
 other. Such constants will be used very often and it is important
 to always remember this notation.

 \item  The inequality $A \lesssim B$ means the following: there exists $c,C$ so
 that for all $R \geq c$, we have $A \leq C B$. The assumption $R \geq c$ will be
 different only if explicitly stated. Sometimes it is superfluous
 but this will not affect the proof.

 \item    $O(\delta )$  will stand for functions and operators satisfying
  $\|O( \delta)\| \ls \delta$.

 \item $\|\cdot\|$ will denote either the $L^2 - $
 norm of a function or the $B(L^2) - $
norm of an operator, depending on the context,  and the symbols
$O(\delta)$ are understood in this norm, or in the absolute value in
the case of complex numbers.

  \item We will write $A \doteq B$ and $A \doteq_{M^d} B$,  if there
   exists $c$ so that $A-B = O (e^{-cR})$ and $A-B =O( M^d e^{-cR})$,
  respectively.

 \item   
 $\langle x \rangle=(1+|x|^2)^{\frac{1}{2}}$
and $\Delta=\sum_{j=1}^N \Delta_{x_j}$,
$\nabla=(\nabla_{x_1},\dots,\nabla_{x_N})$ with $\Delta_{x_j}$,
$\nabla_{x_j}$ the Laplacian and gradient acting on
the coordinate $x_j\in\R^3$, respectively. 
 \item   For a normalized function $\phi$ we define $P_\phi:= |\phi\ran\lan\phi|$ the orthogonal projection onto
  $\phi$ and $P_{\phi}^\bot:=1-P_{\phi}$.
 \end{itemize}

\medskip
\noindent {\bf Acknowledgements.} The author is grateful to Israel
Michael Sigal and Marcel Griesemer for numerous inspiring
discussions on van der Waals forces that were very important both
for beginning this project, and for its development, as well as for
improving its presentation. The author thanks, in addition, Israel
Michael Sigal for numerous very useful comments on drafts of this
paper. He is also grateful to Mathieu Lewin for numerous stimulating
discussions on van der Waals forces between molecules and for many useful remarks,
to Jeremy Quastel for a discussion that was inspiring for Theorem
\ref{thm:upper}, and to two anonymous referees for very careful reports
with many useful corrections, remarks and questions that improved significantly
the paper. The author would also like to thank Volker Bach and
 Alessandro Giuliani for interesting and useful discussions on van
der Waals forces, and Rupert Frank for interesting discussions on
interacting quantum many-body systems.
 The interest of the authors of the previous work \cite{AS} in the
subject of van der Waals forces was derived from talks by and
personal conversations with Herbert Spohn about the Casimir-Polder
effect. This project was supported by the German Science Foundation
under Grant No. GR 3213/1-1.

\section{Preliminaries about many-body systems} \label{Sec:prelim} 

\subsection{Decompositions}\label{sec:deco}
Recall that $M$ and $N$ are the numbers of the nuclei and electrons,
respectively. Let $a=\{A_1,\dots,A_M\}$ be a partition of
$\{1,2,\dots,N\}$ into disjoint subsets some of which might be empty.
With the set $A_i$ we associate the nucleus at $y_i$ of atomic number
$Z_i$ by assigning the electron coordinates $x_j, j \in A_i$ to be in
the same atom/ion as the nucleus at $y_i$. This gives a decomposition
of the system. We denote the collection of all such decompositions
by $\mathcal{A}$ and we will call $A_1,\dots,A_M$ clusters of the
decomposition $a$. The set of all $a \in \mathcal{A}$ with
$|A_i|=Z_i$ for all $i \in \{1,\dots,M\}$ will be denoted by
$\mathcal{A}^{at}$. Its elements correspond to decompositions of our
system into atoms.

If $a=\{A_1,\dots,A_M\}$ and $ b=\{B_1,\dots,B_M\}$ are elements of
$\mathcal{A}^{at}$, then there exists a permutation $\pi \in S_N$ such that
\begin{equation}\label{decopermu1111}
B_i=\{\pi^{-1}(j)| j \in A_i\}, \quad \forall i \in \{1,\dots,M\}.
\end{equation}
 In this case we write $b=\pi a$.
Various $b \in \mathcal{A}^{at}$ are related by permutations of the
electron coordinates and could be labeled as $b =\pi a, \pi \in S_N$ with
some redundancy.

  For each  decomposition $ a=\{A_1,\dots,A_M\} \in \mathcal{A}$ we
define the Hamiltonian
\begin{equation}\label{Ha}
H_a=\sum_{i=1}^{M}H_{A_i},
\end{equation}
where
\begin{equation}\label{Ha1}
 H_{A_i}:=\sum_{j \in
A_i}(- \Delta_{x_j}- \frac{e^2Z_i}{|x_j-y_i|})+\sum_{j,k \in
A_i, j<k}\frac{e^2}{|x_j-x_k|},
\end{equation}
so that $H_{A_i}$ is the Hamiltonian of the atom or ion at $y_i$ and
$H_a$ is the sum of the Hamiltonians of the atoms or ions of the
decomposition $a$. The inter-cluster interaction is defined as
\begin{equation*}
I_a:=H-H_a,
\end{equation*}
where
\begin{equation*}
H:=H^N(y),
\end{equation*}
and $H^N(y)$, defined in \eqref{Hy}, is the Hamiltonian of the system.
 In other words, $I_a$ consists of all terms of
interaction between the different atoms/ions in the decomposition
$a$. We have that
\begin{equation}\label{Hadecomp}
 H=H_a+I_a.
 \end{equation}
 For any cluster $A_i$ we define $S_{A_i}$ to be the permutation
 group of $A_i$ (as identity we consider the permutation in which the elements
 are in increasing order). We define
\begin{equation}\label{def:QAj}
Q_{A_i}:=\frac{1}{|A_i|!} \sum_{\pi \in S_{A_i}} \text{sgn}(\pi) T_\pi,
\end{equation}
\begin{equation}\label{def:HAis}
 H_{A_i}^\s:=H_{A_i} Q_{A_i}
 \end{equation}
and
\begin{equation}\label{def:Has}
 H_a^\s= H_a Q_a, \text{ where } Q_a=Q_{A_1} Q_{A_2}\dots  Q_{A_M}.
\end{equation}

\bigskip

\subsection{Some important properties of various Hamiltonians}\label{importantprop}
For each $m \in \mathbb{N} \cup \{0\}$ with $m \leq N$, we define
\begin{equation}
H^m(y)=\sum_{j=1}^{m} \bigg(- \Delta_{x_j}-\sum_{i=1}^{M}
\frac{e^2 Z_i}{|x_j-y_i|}\bigg)+\sum_{i<j}^{1,m}
\frac{e^2}{|x_{i}-x_{j}|}+\sum_{i<j}^{1,M} \frac{e^2 Z_{i} Z_{j}}{|y_{i}-y_{j}|}
\end{equation}
and
\begin{equation}\label{def:H^m}
H^{m,\s}(y)=H^{m}(y) Q_m,
\end{equation}
where $Q_m$ was defined in \eqref{Qdef}. In simple words
$H^m(y)$ arises from $H=H^N(y)$ after removing $N-m$ electrons.

 The general information on the essential spectrum of the
Hamiltonians defined in \eqref{def:Hinis} and \eqref{def:H^m}
is given in the following theorem which is a
special case of the HVZ Theorem (see e.g. \cite{H, vW, Zh, HS, CFKS}).
 \begin{theorem}\label{hvz}
For all $m \in \{0,1,\dots,N-1\}$, we have that $\sigma_{\textrm{ess}}(H^{m+1,\s}(y)|_{\Ran  Q_{m+1}})=[\Sigma_m,\infty),$ where
\begin{equation}\label{def:sigmam}
\Sigma_m:=\inf \sigma(H^{m,\s}(y)|_{Ran Q_m}).
\end{equation}
Moreover, for all $i \in \{1,2,\dots,M\}$ and all
$n_i \in \mathbb{Z}$ with $n_i \leq Z_i$ we have that
\begin{equation}
\sigma_{\textrm{ess}}(H_{i,n_i-1}^\s|_{\Ran  Q_{(Z_i-n_i+1)}})=[E_{i,n_i},\infty),
\end{equation}
where $E_{i,n_i}$ was defined in \eqref{def:Eini}.
\end{theorem}
The HVZ Theorem enables  to identify the bottom of the essential spectrum
as the ground state energy of the same system but with one electron removed.

 The next result shows that the Hamiltonians
 $H_{i,n_i}^\s|_{\Ran Q_{Z_i-n_i}}, i \in \{1,\dots,M\}, 0 \leq n_i <Z_i$,
  as well as the Hamiltonians $H^{m,\s}(y)|_{\Ran Q_m}, m \in \{1,2,\dots,N\}$,
 have a ground state (see e.g. \cite{ Zh, HS,
CFKS}):
\begin{theorem}\label{zysl1}  For all $i \in \{1,\dots,M\}$ and $n_i \in \mathbb{N} \cup \{0\}$,
 with $n_i<Z_i$, the operator $H_{i,n_i}^\s|_{\Ran  Q_{Z_i-n_i}}$ has a ground state.
  Its ground state energy $E_{i,n_i}$ is below
 the bottom of its essential spectrum (which by Theorem \ref{hvz} is $E_{i,n_i+1}$).
 Similarly, for all $m \in \{1,2,\dots,N\},$ the Hamiltonian $H^{m,\s}(y)|_{\Ran Q_m}$
 has a ground state. Its ground state energy is below the  bottom of its essential spectrum.
  \end{theorem}
This theorem is known as Zhislin's Theorem. It shows that atoms
and positive ions are stable in the sense that they have a bound state.
Lemma \ref{lem:negGSE} follows from Theorems \ref{hvz} and \ref{zysl1}, by observing that
 $H_{i,Z_i}=0, \forall i \in \{1,...,M\}$, which implies that $E_{i,Z_i}=0, \forall i \in \{1,...,M\}$.

Let $a \in \mathcal{A}$ and $i \in \{1,\dots,M\}$. Then the Hamiltonian
$H_{A_i}$, defined in \eqref{Ha1}, differs from the Hamiltonian $H_{i,Z_i-|A_i|}$ defined in
\eqref{def:Hm} only in that $H_{A_i}$ is translated by $y_i$ and it acts on the coordinates
in $A_i$. Therefore, from Theorems \ref{hvz}, \ref{zysl1} and
Corollary \ref{cor:restrictionaway} we obtain:
\begin{theorem}\label{zysl}
 Let $a \in \mathcal{A}$ and $i \in \{1,\dots,M\}$.
 The operators $H_{A_i}^\s|_{\Ran  Q_{A_i}}$ and $H_{A_i}^\s$ have the same discrete spectrum
 (which might be empty) and the corresponding eigenspaces are the same.
 Moreover,
 \begin{equation}\label{infsHAi}
 \inf \s(H_{A_i}^\s)=E_{i,Z_i-|A_i|}, \text{ and } \inf \s_{ess}(H_{A_i}^\s)=E_{i,Z_i-|A_i|+1}.
 \end{equation}
Furthermore, if $0< |A_i| \leq Z_i$, then the operator $H_{A_i}^\s$ has a ground state,
and its ground state energy is below the bottom of its essential spectrum.
 In addition, if $|A_i|=Z_i$, then Condition (D)
implies that $H_{A_i}^\s$ has a unique ground state  $\phi_{A_i}$.
  \end{theorem}
From this theorem it follows that the restrictions of the
operators $H_{A_i}^\s|_{\Ran  Q_{A_i}}$ onto $\Ran  Q_{A_i}$  can be removed
 without affecting the proof. Therefore, in the rest of the paper
 we will consider only the operators $H_{A_i}^\s$.

The following theorem, says that for all $a \in \mathcal{A}^{at}$ the ground states of the operators $H_{A_i}^\s$ are well localized.
We define $x_{A_i}=(x_j: j \in A_i)$ to be the collection of electron coordinates
 in $A_i$ with increasing order in $j$, and $x_{A_i}-y_i=(x_j-y_i: j \in A_i)$, where the order is again increasing in $j$.
  \begin{theorem}\label{thm:ct}
  Let $a \in \mathcal{A}^{at}$ and $i \in \{1,\dots,M\}$. With the same notation as in Theorem \ref{zysl} we have
\begin{equation}\label{groundstatedecay}
\|e^{\theta \langle x_{A_i}-y_i \rangle} \partial^{\alpha} \phi_{A_i}\|
\ls 1, \forall \alpha \text{ with } 0\leq |\alpha| \leq 2,
\end{equation}
 for any $\theta<\sqrt{E_{i,1}-E_i}$.
\end{theorem}
This theorem is known as the as the Combes - Thomas bound (see \cite{CT}).

 The following proposition, is going to be very useful.
\begin{proposition}\label{spherical}
Let $a \in \mathcal{A}^{at}$ and $i \in \{1,\dots,M\}$.
 If $\Psi$ is an eigenfunction of the Hamiltonian $H_{A_i}^\s$,
 corresponding to a non-degenerate eigenvalue, then the one-electron density
\begin{equation*}
\rho_{\Psi}(x)=\int |\Psi(x,x_2,\dots,x_{Z_i})|^{2} dx_2\dots
dx_{Z_i}
\end{equation*}
 of $\Psi$ is spherically symmetric.
\end{proposition}
\begin{proof}
The proposition is standard but we will provide a proof for convenience of the reader.
 For any rotation $U$ in $\mathbb{R}^3$ we consider the transformation $T_U$
defined by
$$T_U \Psi(x_1,\dots,x_{Z_i})=\Psi(U^{-1} x_1,\dots, U^{-1} x_{Z_i}).$$
 Since the Coulomb potentials are spherically symmetric, we have that $H_{A_i}^\s$ commutes with $T_U$, i.e.
$H_{A_i}^\s T_U=T_U H_{A_i}^\s.$ Since $\Psi$ is an eigenfunction of $H_{A_i}^\s$, the
last equality gives that $T_U \Psi$ is also an eigenfunction of
$H_{A_i}^\s$ corresponding to the same eigenvalue. Since the  eigenvalue is
non-degenerate we obtain that $T_U \Psi=c(U) \Psi,$ where $c(U)$ is
a complex valued function. Since $T_U$ is unitary we have that
$|c(U)|=1$ for any rotation $U$ and therefore,
$$|\Psi(x_1,\dots,x_{Z_i})|^2=|\Psi(U^{-1} x_1,\dots,U^{-1} x_{Z_i})|^2,$$
for any $U$.
 Using this and the definition  of the electron density, we conclude that the latter is spherically
 symmetric, because the change of variables arising from a rotation
 has Jacobian 1.
\end{proof}

\bigskip
\section{Proof of Lemma \ref{lem:sigma} and of Property (E) for a system hydrogen atoms}\label{sec:sigma}
\begin{proof}[Proof of Lemma \ref{lem:sigma}]
 For any rotation $U$ in
$\mathbb{R}^3$ we define, similarly to the proof of Proposition
\ref{spherical}, a transformation $T_U$ acting on the space
$L^2(\mathbb{R}^{3(Z_i+Z_j)})$ by
\begin{equation}
T_U \psi(z_1,\dots,z_{Z_i+Z_j})=\psi(U^{-1} z_1,\dots, U^{-1}
z_{Z_i+Z_j}).
\end{equation}
 Similarly as in the proof of Proposition \ref{spherical}, we can
 show that
 \begin{equation}\label{TRphiij}
 T_U \phi_i \otimes \phi_j=c(U) \phi_i \otimes \phi_j, \text{ where } c(U) \in \mathbb{C} \text{ with } |c(U)|=1.
 \end{equation}
 Using \eqref{def:fij} and the fact that the rotation $U$ is
 unitary on $\mathbb{R}^3$, we obtain that
\begin{equation}\label{fijrotation}
T_U^{-1}f_{ij,v}=f_{ij,U^{-1}v}.
\end{equation}
       Using \eqref{TRphiij}, \eqref{fijrotation} and the fact that $T_U$
commutes with $H_{kl}^{\s,\bot}$ we obtain that
$\sigma_{ij}(v)=\sigma_{ij}(U^{-1}v)$ implying that $\sigma_{ij}(v)$ is independent of $v$.

 From \eqref{ineq:Hijsbot} and \eqref{def:rijsbot} it follows that $R_{ij}^{\s,\bot}$ is a positive operator.
 Therefore, using \eqref{sigmaijvdef}, we obtain that
     \begin{equation}
     \sigma_{ij}(v)=\|(R_{ij}^{\s,\bot})^{\frac{1}{2}} f_{ij,v} \phi_i \otimes \phi_j\|^2>0.
     \end{equation}
     We note that $(R_{ij}^{\s,\bot})^{\frac{1}{2}} f_{ij,v} \phi_i \otimes \phi_j \neq 0$, otherwise
     we would multiply with $(H_{ij}^{\s,\bot}-E_i-E_j)^{\frac{1}{2}}$ to obtain that
     $f_{ij,v} \phi_i \otimes \phi_j=0$, which is a contradiction.
\end{proof}

 We now prove Property (E) for a system of hydrogen atoms $(Z_i=1$, for all $i \in \{1,\dots,M\}$).
  Property (E') for a system of hydrogen atoms has already been proven in \cite{AS} Appendix A.
\begin{proposition} \label{prop:ConditionEhyd}
 If $Z_i=1$, for all $i \in \{1,2,\dots,M\}$, then Property (E) holds.
\end{proposition}
  \begin{proof}
 It is enough to restrict our attention to the first
 two atoms. More, precisely, it is enough to show that
  \begin{equation}\label{hyd-ineq1}
  E_{1,m}+E_{2,-n}<E_{1,m+l}+E_{2,-n-l}, \text{ } \forall m,n \in \mathbb{N} \cup \{0\}, l \in \mathbb{N}, \text{ with } m+l \leq 1.
  \end{equation}
   From the assumptions on $l,m$ we obtain that $m=0, l=1$. Since
  $E_{1,1}=0$ (because $H_{1,1}=0$), showing \eqref{hyd-ineq1} reduces to showing
  that
  \begin{equation}\label{hyd-ineq}
   E_1+E_{2,-n} < E_{2,-n-1}, \text{ } \forall n \in  \mathbb{N} \cup \{0\}.
  \end{equation}
To prove \eqref{hyd-ineq}, we assume that for some
$n \in  \mathbb{N} \cup \{0\}$ we have $E_1+E_{2,-n} \geq E_{2,-n-1}$. Using \eqref{eqn:Eiless0}
 we obtain that $E_{2,-n-1} < E_{2,-n}$. Therefore, from Theorem \ref{hvz} it
follows that $E_{2,-n-1}< \inf \s_{\text{ess}} (H_{2,-n-1}^\sigma)$. The last inequality implies that $H_{2,-n-1}^\sigma$ has a ground state $\Psi$. Since
\begin{equation*}
H_{2,-n-1}^\s=\left(H_{2,-n}^\s+(-\Delta_{x_{n+2}}-\frac{e^2}{|x_{n+2}|}) +\sum_{i=1}^{n+1} \frac{e^2}{|x_i-x_{n+2}|}\right)Q_{n+2},
\end{equation*}
it follows that
\begin{equation*}
H_{2,-n-1}^\sigma \geq \left(E_{2,-n}+E_1 +\sum_{i=1}^{n} \frac{e^2}{|x_i-x_{m+1}|}\right) Q_{n+2}.
\end{equation*}
Taking the expectation value with respect to the ground state $\Psi$  of $H_{2,-n-1}^\s$ we obtain that
 $$E_{2,-n-1} >  E_{2,-n} +E_1,$$ giving a contradiction.
 \end{proof}

\section{Reformulation of Theorem \ref{thm:vdW-maxspin}}\label{sec:reformulation}
In this Section we choose an orthogonal projection
for the Feshbach map. Then we state two propositions and two lemmas
and we show that they imply Theorem \ref{thm:vdW-maxspin}.
We prove them, however,
in subsequent sections.
 As we mentioned above, $\Pi$ should be a projection on an antisymmetric
 tensor product of the ground states of the atoms. It turns out, however,
 that it is useful to cut the ground states off appropriately,
  so that different terms of the antisymmetrization have disjoint supports.
  Before introducing the cut off we introduce the ground states.
 We recall that $\phi_i$ is the ground state of $H_i^\sigma$ (see \eqref{def:His}),
 which is unique, due to Condition (D), for any $i \in \{1,\dots,M\}$.
     Let $a=\{A_1,\dots, A_M\} \in \mathcal{A}^{at}$ and $i \in \{1,\dots,M\}$. The ground state
    $\phi_{A_i}$ of $H_{A_i}^\s$ (see theorem \ref{zysl})
     is given by
 \begin{equation}\label{phiAj}
\phi_{A_i}(x_{A_i})=\phi_{i}(x_{A_i}-y_i),
 \end{equation}
where, recall,  $x_{A_i}=(x_j: j \in A_i)$ is the collection of electron coordinates
 in $A_i$ with increasing order in $j$, and $x_{A_i}-y_i=(x_j-y_i: j \in A_i)$, where the order is again increasing in $j$.
Moreover,
\begin{equation*}
H_{A_i} \phi_{A_i}=E_i \phi_{A_i}.
\end{equation*}
   It follows that $H_a^\s$, defined in \eqref{def:Has}, has a unique ground state
 $\Phi_a$ and that $H_a\Phi_a=\Einfty \Phi_a$, where $\Einfty$ was defined in \eqref{Einftydef}.
  Moreover, $\Phi_a$ is given by
\begin{equation}\label{Phiaeq}
 \Phi_{a}(x_1,\dots,x_N)= \prod_{i=1}^M \phi_{A_i}(x_{A_i}).
\end{equation}

Now we introduce the cut off of the ground states $\Phi_a, \phi_{A_i}, i \in \{1,\dots,M\}$.
 Let $\chi_1: \mathbb{R}^3
\rightarrow \mathbb{R}$ be a spherically symmetric $C^\infty$ function
supported in the ball $B(0,\frac{1}{6})$ and equal $1$ on the ball
$B(0,\frac{1}{7})$ with $ 0 \leq \chi_1 \leq 1$ (smoothed out
characteristic function of the ball $B(0,\frac{1}{6})$), and define
\begin{equation}\label{chiR}
\chi_R(x):=\chi_1(R^{-1} x).
\end{equation}
 We introduce the cut off ground states
\begin{equation}\label{Psiaeq}
 \Psi_{a}(x_1,\dots,x_N):= \prod_{i=1}^M \psi_{A_i}(x_{A_i}),
\end{equation}
where, as in \eqref{phiAj},
\begin{equation}\label{psiAj}
\psi_{A_i}(x_{A_i}):=\psi_i(x_{A_i}-y_i),\ \quad \mbox{with}\ \quad 
 \psi_i:=\frac{\phi_i \chi_R^{\otimes Z_i}}{\|\phi_i \chi_R^{\otimes Z_i}\|}.
\end{equation}
Recall that $E_i$ is the ground state energy of $H_i^\s$.
  Using  \eqref{phiAj}, \eqref{groundstatedecay},
 $H_{A_i} \phi_{A_i}= E_i \phi_{A_i}$ and \eqref{psiAj},
 one can verify that
\begin{align}\label{phijpsij}
&\|\psi_i-\phi_i\|_{H^2} \doteq 0, \|\psi_{A_i}-\phi_{A_i}\|_{H^2} \doteq 0, \text{ } \forall i \in \{1,\dots,M\},\\
\label{almosteigat}
&\|(H_{A_i} - E_i) \psi_{A_i}\|_{H^1} \doteq 0, \text{ } \forall i \in \{1,\dots,M\},
\end{align}
and that there exists $c_1>0$, depending only on $Z$, so that
\begin{align}
\label{Psiadecay} &\|e^{c_1 \langle x_{A_i}-y_i \rangle}
\partial^{\alpha} \psi_{A_i}\|
\ls 1,\ \quad  \forall \alpha\   \text{ with }  0 \leq |\alpha| \leq 2, \quad \forall i \in \{1,\dots,M\}.
\end{align}
Using \eqref{phijpsij}, together with \eqref{Phiaeq}, \eqref{Psiaeq} and the triangle inequality,
we obtain that
\begin{align}\label{phiapsia}
&\|\Phi_a-\Psi_a\|_{L^2} \doteq_M 0, \text{ } \forall a \in \mathcal{A}^{at}.
\end{align}
From \eqref{Einftydef} and \eqref{Ha} it follows that
\begin{equation}
H_a  - \Einfty=\sum_{i=1}^M (H_{A_i}-E_i), \text{ } \forall a \in \mathcal{A}^{at}.
\end{equation}
Using the last equation together with \eqref{Psiaeq}, \eqref{almosteigat}
and the triangle inequality, we obtain that
\begin{align}\label{almosteig}
&\|(H_a  - \Einfty) \Psi_a\|_{H^1} \doteq_M 0, \text{ } \forall a \in \mathcal{A}^{at}.
\end{align}
From \eqref{Psiaeq}, \eqref{psiAj} and the properties of $\chi_R$ it follows that
\begin{align}\label{pairor}
 & \text{supp}(\Psi_a)  \cap \text{supp}(\Psi_b) = \emptyset, \ \quad \forall a,b \in
\mathcal{A}^{at},\ \quad \text{ with }\ \quad a \neq b.
\end{align}
Property \eqref{pairor} was the main reason for introducing the cut
off of the ground states. As we shall see later, a lot of error terms, that would exist without \eqref{pairor},
vanish.

 We are now ready to choose the projection $\Pi$. Let
\begin{equation}\label{Pi}
\Pi:=\left|\frac{Q_N \Psi_a}{\|Q_N \Psi_a\|} \right\rangle \left\langle \frac{Q_N
\Psi_a}{\|Q_N \Psi_a\|}\right|,
\end{equation}
where $Q_N$ was defined in \eqref{Qdef}.
An elementary computation gives that
\begin{equation}\label{QNpsia}
Q_N \Psi_a=\frac{1}{|\mathcal{A}^{at}|}\sum_{b \in \mathcal{A}^{at}} \text{sgn}(b,a) \Psi_b.
\end{equation}
Here $\text{sgn}(b,a)$ is the sign  of the unique permutation $\pi$ with the properties
$\pi(A_i)=B_i$, for all $i \in \{1,\dots,M\}$ and $\pi|_{A_i}$ is an increasing function for all $i \in \{1,\dots,M\}$.
From \eqref{QNpsia} and \eqref{pairor} it follows that $Q_N \Psi_a \neq 0$. Moreover,
from \eqref{QNpsia} it follows that the right hand side of \eqref{Pi} does not depend
on the decomposition $a$.

 To show that the Feshbach map exists with this choice of $\Pi$, we note that $Q_N \Psi_a \in
\text{Dom}(H^\s)$ since $\Psi_a \in H^2(\R^{3N})$. Hence, the
condition (a) for the existence of the Feshbach map holds. Recall that
$E(y)$ was defined in \eqref{def:GSE}.
We will prove in Section \ref{Hbotbndseveral} the following proposition:
 \begin{proposition}\label{prop:Eimpliesgap}
(i) Suppose that Property (E) holds. Then there exist $C_1,\imsgap_1>0$, depending only
on $Z$, such that if $R \geq C_1 N^{\frac{4}{3}}$, then
\begin{equation}\label{Hbotbnd}
\H^{\s,\bot} \geq E(y)+2\imsgap_1,
\end{equation}
where $\H^{\s,\bot}=\Pi^\bot H^\s \Pi^\bot$ with $\Pi$ defined in \eqref{Pi}.
\newline
(ii) Suppose that Property (E') holds. Then the same conclusion as in part (i) holds with the constants
$C_1, \imsgap_1$ replaced by constants $C_1',\imsgap_1'$ which again are positive and dependent only on
$Z$.
 \end{proposition}
The constants $C_2, C_2'$ in Theorem \ref{thm:vdW-maxspin} are different,
 only because the constants $\imsgap_1, \imsgap_1'$ in Proposition
 \ref{prop:Eimpliesgap} are different:
in the proof we shall use the boundedness of the resolvent
 $(\H^{\s,\bot} - E(y))^{-1}$
but the estimate of its norm is different in
 the cases (i) and (ii).

 From Proposition \ref{prop:Eimpliesgap} and Property (E), respectively (E'), it follows that the Feshbach map \eqref{FP}
 is defined at $\lambda=E(y)$ when $R \geq C_1 N^{\frac{4}{3}}$, repsectively $R \geq C_1' N^{\frac{4}{3}}$.
 In Section \ref{sec:setup} we will prove the following lemma:
 \begin{lemma}\label{lem:Eyless0}
There exists $C$ so that if $R \geq C$, then $E(y)<0$.
 \end{lemma}
The operator $H^\s$ differs from $H^\s|_{\Ran  Q_N}$ only because every element of
$(\Ran  Q_N)^\bot$ is eigenvector of $H^\s$ with eigenvalue zero. Therefore, from
Lemma \ref{lem:Eyless0} and Theorem \ref{zysl1} it follows that $E(y)$ is the ground
 state energy of $H^\s$ and that it belongs to its discrete spectrum.
  Therefore, under the assumptions of Proposition \ref{prop:Eimpliesgap} and Lemma \ref{lem:Eyless0}, we have that \eqref{FSE} holds.
From \eqref{FSE} and \eqref{FP} we obtain that
\begin{equation}\label{FSE1}
E(y)= (\Pi H^\s \Pi+ V(E(y)))|_{\Ran  \Pi}.
\end{equation}
In view of this equation, estimating $E(y)$ reduces to estimating the terms
$\Pi H^\s \Pi$ and $V(E(y)$. To this end we will prove
in Section \ref{sec:setup} the following lemma and proposition:
\begin{lemma}\label{lem:PHP}
With $\Pi$ defined in \eqref{Pi}, we have that
 \begin{equation}\label{PHP}
\Pi \H^\s \Pi \doteq_M \Einfty \Pi.
\end{equation}
\end{lemma}

\begin{proposition}\label{prop:Feshest}
Suppose that \eqref{Hbotbnd} holds. Then, there exist
 $C_3, C_7>0$ depending on $Z$ only, so that if $R \geq C_7 \ln M$, then
\begin{equation}\label{Uest}
 \bigg|V(E(y))|_{\Ran  \Pi}- \sum_{i<j}^{1,M} \frac{ e^4
\sigma_{ij}}{|y_i-y_j|^6}\bigg| \ls
\sum_{i<j}^{1,M}\frac{1}{|y_i-y_j|^7}+\frac{M^4}{R^9}\big(1+  N^Z e^{-C_3 R}\big),
\end{equation}
where the constants $\sigma_{ij}$ were defined in \eqref{def:sigmaij}.
\end{proposition}
Theorem \ref{thm:vdW-maxspin} follows
from Propositions \ref{prop:Eimpliesgap} and \ref{prop:Feshest}, and from \eqref{PHP}, \eqref{FSE1} and \eqref{interene}.
The constant $C_3$ in Proposition \ref{prop:Feshest} is the same as the constant $C_3$
in Theorem \ref{thm:vdW-maxspin}.
\begin{remark}
The assumption on $R$ in Proposition \ref{prop:Feshest} is weaker
than the assumption on $R$ in Theorem \ref{thm:vdW-maxspin}. This is explained by Proposition \ref{prop:Eimpliesgap}. We need to assume more  on $R$ in order to prove \eqref{Hbotbnd}. The remainder in \eqref{Uest} is small relative to the leading order provided that $R\geq CM^{\frac{2}{3}}$ for some $C$. Given this observation we believe that the assumption on $R$ in Proposition \ref{prop:Eimpliesgap} can be improved. We have been, however, unable to do this.
\end{remark}



\section{Proof of Proposition \ref{prop:Feshest} and of Lemmas \ref{lem:PHP} and \ref{lem:Eyless0} } \label{sec:setup}

 The proofs are organized as follows.  In Section \ref{sub:PHPEst}
 we prove Lemmas \ref{lem:PHP} and \ref{lem:Eyless0}. In
 Section \ref{Esharplower} we prove rough estimates for $E(y)$ showing that it is fairly close
 to $E(\infty)$.  In Section \ref{sec:VEest} we prove Proposition \ref{prop:Feshest} assuming a lemma,
 which in turn we prove in Section \ref{sec:Resdelta-est}.

 \subsection{Proof of Lemmas \ref{lem:PHP} and \ref{lem:Eyless0}}\label{sub:PHPEst}
We begin with the proof of Lemma \ref{lem:PHP}.
Using  \eqref{pairor}, \eqref{Pi}, \eqref{QNpsia}  and that $\text{supp}(H \Psi_a) \subset \text{supp}(\Psi_a)$ for all $a \in \mathcal{A}^{at}$,
we obtain that
\begin{equation*}
\Pi \H^\s \Pi=\lan \Psi_a, H \Psi_a \ran \Pi,
\end{equation*}
because all the cross terms originating from the antisymmetrization
vanish. Therefore, using \eqref{almosteig} and \eqref{Hadecomp}, we
obtain that
\begin{equation}\label{PHP2}
\Pi \H^\s \Pi \doteq_M \Einfty \Pi + \lan \Psi_a, I_a \Psi_a \ran \Pi.
\end{equation}
\begin{lemma}\label{lem:newton}
For all $a \in \mathcal{A}^{at}$, we have
\begin{equation}\label{psiaIbpsib}
\langle \Psi_{a}, I_a \Psi_{a} \rangle = 0.
\end{equation}
\end{lemma}
\begin{proof}
 Let $a=\{A_1,\dots,A_M\}$. We have that
\begin{equation}\label{Iabreakup}
I_{a}=\sum_{i<j}^{1,M} I_{A_i A_j},\ \quad I_{A_i A_j}:=\sum_{l \in
A_i, m \in A_j}I_{ij}^{lm},
\end{equation}
where
\begin{equation}\label{bijkl}
I_{ij}^{lm}=-\frac{e^2}{|y_j-x_l|}-\frac{e^2}{|y_i-x_m|}+\frac{e^2}{|y_i-y_j|}+\frac{e^2}{|x_l-x_m|}.
\end{equation}
We pass to the variables
\begin{equation}\label{zlm} z_{lr}=x_l-y_r,\
\forall  l \in A_r, \forall r \in \{1,\dots,M\},
\end{equation} and write $z_{A_r}=(z_{lr}: l \in A_r)$ and
$dz_{A_r}=\prod_{l \in A_r} dz_{lr}$.  Using \eqref{Psiaeq} and
\eqref{psiAj}, we obtain that
\begin{equation}\label{intdijkl}
\langle \Psi_a,I_{ij}^{lm} \Psi_a\rangle=\int dz_{A_j} dz_{A_i}
\tilde{I}_{ij}^{lm} |\psi_{i}( z_{A_i})|^2 |\psi_{j}(z_{A_j})|^2,
\end{equation}
where
\begin{equation}\label{dijkl}
\tilde{I}_{ij}^{lm}=-\frac{e^2}{|y_{ij}+z_{li}|}-\frac{e^2
}{|y_{ij}-z_{mj}|}+\frac{e^2}{|y_{ij}|}+\frac{e^2}{|y_{ij}+z_{li}-z_{mj}|},
\end{equation}
    with $y_{ij}=y_i-y_j$.
We will show that
\begin{equation}\label{onesymmetryenough'}
\chi_{2R}^{\otimes Z_j} \int dz_{A_i} \tilde{I}_{ij}^{lm}
|\psi_{i}( z_{A_i})|^2=0,
\end{equation}
where $\chi_{2R}$ was defined in \eqref{chiR} and
$\chi_{2R}^{\otimes Z_j}$ acts on the variables $z_{A_j}$.
     Since by Proposition \ref{spherical}  the one-electron density of the function $\psi_{i}$ is
spherically symmetric, we have by Newton's Theorem (see for example
\cite{LL} Section 9.7) and the support properties of $\psi_i$ due to the cut off
(see \eqref{psiAj}) that
\begin{equation}\label{genexperror2}
\int |\psi_i(z_{A_i})|^2 \frac{1}{|y_{ij}+z_{li}|} dz_{A_i}=
\frac{1}{|y_{ij}|} \int_{|z_{li}| \leq |y_{ij}|} |\psi_i(z_{A_i})|^2
dz_{A_i}= \frac{1}{|y_{ij}|}.
\end{equation}
 In the same way we obtain that
\begin{equation}\label{genexperror3}
\int |\psi_i (z_{A_i})|^2 \frac{1}{|y_{ij}+z_{li}-z_{mj}|} dz_{A_i}
= \frac{1}{|y_{ij}-z_{mj}|} \text{ on  supp} \chi_{2R}^{\otimes
Z_j}.
\end{equation}
 From \eqref{dijkl}, \eqref{genexperror2} and
\eqref{genexperror3} we obtain \eqref{onesymmetryenough'}.
 From \eqref{intdijkl} and \eqref{onesymmetryenough'} we obtain
that $\langle \Psi_a,I_{ij}^{lm} \Psi_a\rangle = 0,$ for all $i,j
\in \{1,\dots,M\}$ with $i \neq j$ and $l \in A_i, m \in A_j$, which
together with \eqref{Iabreakup} implies \eqref{psiaIbpsib}.
\end{proof}

From \eqref{PHP2} and \eqref{psiaIbpsib} we obtain \eqref{PHP}. This
concludes the proof of Lemma \ref{lem:PHP}.

\begin{remark}
Using \eqref{onesymmetryenough'}, going back to the variables $x_j$
and using \eqref{Iabreakup}, we can easily obtain that
\begin{equation}\label{onesymmetryenough}
\chi_{2R}^{A_j} \int dx_{A_i} I_{A_i A_j} |\psi_{A_i}(
x_{A_i})|^2=0, \quad \forall i \neq j,
\end{equation}
where
\begin{equation}\label{charaj}
 \chi_{2R}^{A_j}(x_{A_j}):= \prod_{i \in A_j} \chi_{2R} (x_i-y_j),
\end{equation}
and $\chi_{2R}$ was defined in \eqref{chiR}. This will be useful later
in the proof. The physical meaning of \eqref{onesymmetryenough} is
that the potential created by a spherically symmetric charge
distribution with total charge zero, is zero outside of its support.
\end{remark}
We shall now prove Lemma \ref{lem:Eyless0}.
\begin{proof}[Proof of Lemma \ref{lem:Eyless0}]
Since $\Pi \H^\s \Pi|_{\Ran \Pi}$ is the expectation
of the Hamiltonian against an antisymmetric function we obtain that
$E(y) \leq \Pi \H^\s \Pi|_{\Ran \Pi} \doteq_M \Einfty$, where
the last step follows from \eqref{PHP}. Therefore,
there exist $c$ and $C$ so that
\begin{equation}\label{est:Ey}
 E(y) \leq \Einfty + C M e^{-cR}.
\end{equation}
From \eqref{eqn:Eiless0} and \eqref{Einftydef} it follows that
$\Einfty \leq M \max\{E_j: j \in \{1,\dots,M\}\}<0$. The last inequality together with \eqref{est:Ey} imply
Lemma \ref{lem:Eyless0}.
\end{proof}

\subsection{Rough bounds on $E(y)$}\label{Esharplower}
Before we estimate $V(E(y))$ we first prove some rough bounds for $E(y)$
which are going to be useful. Our goal is to prove
\begin{lemma}\label{prop:ineqE}
Assume that \eqref{Hbotbnd} holds. Then there exists $c$ so that
\begin{equation}\label{ineqE}
-\sum_{1=i<j}^M \frac{1}{|y_i-y_j|^6} \lesssim \E-\Einfty \lesssim M
e^{-cR}.
\end{equation}
\end{lemma}
\begin{proof}
The right-side inequality is just a restatement of \eqref{est:Ey}.
Therefore, it remains to prove the left-side inequality.
By \eqref{Hbotbnd}, the Feshbach map in \eqref{FP} is well defined at $\lambda=E(y)$, and
 \eqref{FSE} holds. By \eqref{FSE}, \eqref{FP} and
\eqref{PHP} we have that
\begin{equation}\label{FPEdot}
E(y)  \doteq_M  \Einfty-V(E(y))|_{\Ran  \Pi}.
\end{equation}
Now we estimate $V(E(y))$. By \eqref{Hbotbnd}, the definition of $V(E(y))$
in \eqref{W},  and $\Pi^\bot H^\s \Pi=(\Pi H^\s \Pi^\bot)^*$, we
obtain
 \begin{equation}\label{Uexpab}
 \| V(E(y)) \| \ls  \|\Pi^\bot H \Pi\|^2,
 \end{equation}
 where we could replace $H^\s$ with $H$ because $\Pi$ is a projection onto
 an antisymmetric function.
 To estimate $\Pi^\bot H \Pi$ we use that
 \eqref{PHP} implies that
\begin{equation*}
\Pi^\bot H \Pi \doteq_M H \Pi- \Einfty \Pi.
\end{equation*}
From the last estimate and \eqref{Pi} it follows that
\begin{equation}\label{PibotHpi}
\|\Pi^\bot H \Pi\| \doteq_M \|(\H - \Einfty) \frac{Q_N \Psi_a}{\|Q_N
\Psi_a\|}\|.
\end{equation}
Using that $Q_N$ commutes with $H$, that $\text{supp}(H \Psi_a) \subset \text{supp} (\Psi_a)$ for all $a \in \mathcal{A}^{at}$,
\eqref{pairor} and \eqref{QNpsia}, one can verify that
\begin{equation*}
 \|(\H - \Einfty) \frac{Q_N \Psi_a}{\|Q_N \Psi_a\|}\|=\|(H-\Einfty) \Psi_a\|.
\end{equation*}
The last equation together with \eqref{PibotHpi},
\eqref{almosteig} and \eqref{Hadecomp} give that
\begin{equation*}
\|\Pi^\bot H \Pi\| \doteq_M \|I_a \Psi_a\|,
\end{equation*}
where, due to symmetry, the right hand side is $a$ independent. The
last approximate equality together with \eqref{Uexpab} gives that there exists
$c>0$ so that
\begin{equation}\label{Usharporder}
\|V(E(y))\| \ls \|I_a \Psi_a\|^2+  M^2 e^{-cR}.
\end{equation}
 To complete the proof we use Lemma \ref{lem:Iapsia} below which
together with \eqref{Usharporder}
 and \eqref{FPEdot} gives the left inequality in \eqref{ineqE}. This
concludes the proof of \eqref{ineqE}.
  \end{proof}

\begin{lemma}\label{lem:Iapsia}
We have that
\begin{equation}\label{IapsiaL2est}
\|I_a \Psi_a\|^2 \lesssim \sum_{i<j}^{1, M} \frac{1}{|y_i-y_j|^6}.
\end{equation}
\end{lemma}

\begin{proof}

By \eqref{Iabreakup} we have that
\begin{equation*}
\|I_a \Psi_a\|^2 =\sum_{i<j}^{1,M} \sum_{k<l}^{1,M} \lan I_{A_i A_j}
\Psi_a, I_{A_k A_l} \Psi_a \ran.
\end{equation*}
It follows then from  \eqref{Psiaeq}, \eqref{psiAj} and
\eqref{onesymmetryenough} that the cross terms of the double sum
vanish, and therefore we obtain
\begin{equation}\label{Ibpsib}
\|I_a \Psi_a\|^2 =\sum_{i<j}^{1,M} \lan I_{A_i A_j} \psi_{A_i}
\psi_{A_j}, I_{A_i A_j} \psi_{A_i} \psi_{A_j} \ran.
\end{equation}
Making the change of variables \eqref{zlm} and using \eqref{psiAj},
we obtain that
\begin{equation}\label{changeIAiAj}
 \|I_{A_i A_j} \psi_{A_i} \psi_{A_j} \| = \|\tilde{I}_{A_i A_j} \psi_i \psi_j\|,
\end{equation}
where
\begin{equation}\label{def:tildeIAiAj}
\tilde{I}_{A_i A_j}: = \sum_{l \in A_i, m \in A_j}
\tilde{I}_{ij}^{lm}
\end{equation}
and $\tilde{I}_{ij}^{lm}$ were defined in \eqref{dijkl}. The tensor
product of the functions $\psi_i, \psi_j$ has been omitted. We will
prove now that
 for all $l \in A_i, m \in
A_j$ we have
\begin{equation}\label{prepreIapsia2est}
\tilde{I}_{ij}^{lm} \psi_i  \psi_j =\frac{e^2}{|y_i-y_j|^3}
f_{ij, \widehat{y_{ij}}}^{lm} \psi_i  \psi_j+O( \frac{1}{|y_i-y_j|^4}),
\end{equation}
where
\begin{equation}\label{def:fijlm}
f^{lm}_{ij, \widehat{y_{ij}}} (z_{li}, z_{mj}):= z_{li} \cdot z_{mj}-
3(z_{li} \cdot \widehat{y_{ij}})(z_{mj} \cdot
  \widehat{y_{ij}}),
\end{equation}
$y_{ij}=y_i-y_j$, $\widehat{y_{ij}}=\frac{y_{ij}}{|y_{ij}|}$ and
$z_{li}, z_{mj}$ were defined in \eqref{zlm}. The estimate
\eqref{prepreIapsia2est} is a little stronger than what we need now,
but it is going to be useful later. We will Taylor expand all the
terms on the right hand side of \eqref{dijkl} in powers of $\frac{1}{|y_{ij}|}$. Using
that
\begin{equation*}
\frac{1}{|y \pm z|}=\frac{1}{|y|} \mp \frac{z \cdot
\hat{y}}{|y|^2}+\frac{3(\hat{y} \cdot z)^2-|z|^2}{2 |y|^3}+O( \frac{
|z|^3}{|y|^4}),
\end{equation*}
provided that $|z| \leq \frac{|y|}{3}$, we see that the
contributions on the right hand side of \eqref{dijkl} cancel in the
first and second order and we compute
\begin{equation}\label{I1sevmu}
\tilde{I}_{ij}^{lm}=\frac{e^2}{|y_{ij}|^3}f_{ij, \widehat{y_{ij}}}^{lm}+O(\frac{
|z_{li}|^3+|z_{mj}|^3}{|y_{ij}|^4}),\
 \text{ on the set } \{|z_{li}|, |z_{mj}|, |z_{li}-z_{mj}| \leq \frac{|y_{ij}|}{3}\}
\end{equation}
and in particular the last estimate  holds on $\supp \psi_i \psi_j$.
The change of variables \eqref{zlm}, and \eqref{Psiadecay} imply
that there exists $c$ so that
\begin{equation}\label{Psidecay}
\|e^{c \langle z_{A_i} \rangle}
\partial^{\alpha} \psi_{i}\|
\ls 1,\ \quad  \forall \alpha\   \text{ with }  0 \leq |\alpha| \leq
2, \quad \forall i  \in \{1,\dots,M\}.
\end{equation}
 Therefore,
 if we multiply both sides of \eqref{I1sevmu} by $\psi_i
 \psi_j$, we can control the term $ |z_{li}|^3+|z_{mj}|^3 $ in the
remainder uniformly in $|y_{ij}|$ due to \eqref{Psidecay}. As a
consequence, we arrive at \eqref{prepreIapsia2est}. From
\eqref{def:tildeIAiAj} and \eqref{prepreIapsia2est} it follows that
   \begin{equation}\label{exp:IAiAj}
 \tilde{I}_{A_i A_j} \psi_i  \psi_j= \frac{e^2}{|y_i-y_j|^3} f_{ij, \widehat{y_{ij}}} \psi_i  \psi_j+O(  \frac{1}{|y_i-y_j|^4}),
\end{equation}
where $f_{ij, \widehat{y_{ij}}}$ was defined in \eqref{def:fij}. Therefore, we arrive at
\begin{equation}\label{preIapsiaL2est}
\|\tilde{I}_{A_i A_j} \psi_i \psi_j \| \lesssim \frac{1}{|y_i-y_j|^3},
\end{equation}
where the inequality $\| f_{ij,  \widehat{y_{ij}}} \psi_i  \psi_j\| \ls 1$ follows from
\eqref{Psidecay}. From \eqref{Ibpsib}, \eqref{changeIAiAj}
and \eqref{preIapsiaL2est} we obtain \eqref{IapsiaL2est}.
\end{proof}


\subsection{ Proof of Proposition \ref{prop:Feshest}}\label{sec:VEest}
For this proof it is more
 convenient to write $\Pi$ as a product of two projections. To this
 end we choose $P$ to be the orthogonal projection on  $\text{span}
\{\Psi_a: a \in \mathcal{A}^{at}\}$. Using
 \eqref{pairor} it follows that
\begin{equation}\label{P}
   P = \sum_{a \in \mathcal{A}^{at}} P_{\Psi_a}.
\end{equation}
 It is straightforward to see that $P$ is symmetric with respect to
 the electron coordinates in the sense that
\begin{equation}\label{Psym}
P T_\pi=T_\pi P, \quad \text{ } \forall \pi \in S_N,
\end{equation}
where $T_{\pi}$ was defined in \eqref{def:Tpi}.
 Using \eqref{Psym} and \eqref{Qdef} it follows that
\begin{equation}\label{pqqp}
Q_N P=P Q_N,
\end{equation}
 and that $Q_N P=PQ_N $ is also an orthogonal projection. It is, moreover, easy to
 check that
\begin{equation}\label{Piqp}
\Pi=Q_N P=P Q_N.
\end{equation}
Therefore, by \eqref{W} we obtain that
\begin{equation}\label{WU}
V(\lambda)=Q_N P \H P^\bot (\H^{\s,\bot}-\lambda)^{-1} P^\bot \H P Q_N.
\end{equation}
 We will use \eqref{WU} to estimate $V(E(y))$. Before doing so, we will introduce some
 useful notation. For a decomposition $b=\{B_1, B_2,\dots, B_M\} \in \mathcal{A}^{at}$ we define
\begin{align} \notag
H_{B_k B_l}:=H_{B_k}+ H_{B_l},\\
\notag H_{B_k B_l}^\s= Q_{B_k}  Q_{B_l} H_{B_k B_l},
\end{align}
 where $H_{B_k}, H_{B_l}, Q_{B_k}, Q_{B_l}$ where defined in
\eqref{Ha1} and \eqref{def:QAj}.
  We define further
\begin{align}
 \label{Hbkldef} H_{B_k B_l}^{\s,\bot}&:=P_{\psi_{B_k} \psi_{B_l}}^\bot H_{B_k
 B_l}^\s  P_{\psi_{B_k} \psi_{B_l}}^\bot,\\\label{Rbkldef}
  \Rbkl&:=(H_{B_k B_l}^{\s,\bot}-E_k-E_l)^{-1},
\end{align}
where recall that $\psi_{B_k}, \psi_{B_l}$ where defined in \eqref{psiAj}.
Using that $\psi_{B_k} \psi_{B_l}$ is a cut off-ground state of $H_{B_k B_l}$
and \eqref{Psiadecay}, one can verify that there exists
$C$ so that for $R \geq C$ the resolvent $\Rbkl$ is defined and $\|\Rbkl\| \ls 1$.
  We begin with
  \begin{lemma}\label{lem:Uentwicklung}
If \eqref{Hbotbnd} holds, then there exists $C_3, C_7>0$ depending only
on $Z$ so that it $R \geq C_7 \ln M$, then we have that
\begin{equation}\label{Uentwicklung}
 \bigg| V(E(y))|_{\Ran  \Pi}- \sum_{k<l}^{1,M} \lan I_{A_k A_l} \psi_{A_k}
\psi_{A_l}, \Rakl I_{A_k A_l} \psi_{A_k} \psi_{A_l}\ran \bigg| \ls
\frac{M^4}{R^9}+ \frac{M^4}{R^9} N^Z e^{-C_3 R}.
\end{equation}
   \end{lemma}
\begin{proof}
In view of \eqref{WU}, we will first estimate $P^\bot H P$. We have
that
\begin{equation}\label{PbotHP}
P^\bot H P= \sum_{a \in \mathcal{A}^{at}} |\varphi_a \rangle \langle \Psi_a|, 
\end{equation}
where $\varphi_a:= P^\bot H \Psi_a$. From \eqref{almosteig},
\eqref{pairor}, \eqref{P} and the fact that $\text{supp}(H_a \Psi_a) \subset \text{supp}(\Psi_a)$
 we obtain that $ P^\bot H_a \Psi_a \doteq_M 0$, which together with the definition of $\varphi_a$ and
 \eqref{Hadecomp}  imply
\begin{equation}\label{chiaest1}
\varphi_a \doteq_M P^\bot I_a \Psi_a.
\end{equation}
 By \eqref{pairor}, \eqref{P} and \eqref{psiaIbpsib}  we have
\begin{equation}\label{chiaest2}
P^\bot I_a \Psi_a= I_a \Psi_a.
  \end{equation}
The last equation together with \eqref{chiaest1} gives that
\begin{equation}\label{chiaest}
\varphi_a \doteq_M I_a \Psi_a.
\end{equation}
 Now we use the following
inequality:
 If $(\phi_n)_{n=1}^m$ are pairwise orthogonal, and
$(\psi_n)_{n=1}^m$ are also pairwise orthogonal, then
\begin{equation}\label{pairorthnorm}
\|\sum_{n=1}^m |\phi_n \rangle \langle \psi_n| \| \leq \max_{n \in
\{1,\dots,m\}} \| |\phi_n \rangle \langle \psi_n| \|.
\end{equation}
The inequality \eqref{pairorthnorm} follows by
$\|B\|=\sup_{\|\phi\|,\|\psi\|=1} \lan \phi, B \psi \ran$, with
$B=\sum_{n=1}^m |\phi_n \rangle \langle \psi_n|$, and from the
Cauchy-Schwarz and Parseval's inequalities (in fact in \eqref{pairorthnorm} equality
holds but we do not need this).

  Using  \eqref{PbotHP}, \eqref{chiaest} and the fact
  that $\varphi_a -I_a \Psi_a$, $a \in \cA^{at}$ have disjoint supports, (following from \eqref{pairor} and the
  fact that $\text{supp} \varphi_a \subset \text{supp} \Psi_a$), we can
  apply  \eqref{pairorthnorm}, to obtain that
  \begin{equation}\label{PbotHPdot}
P^\bot H P \doteq_M \sum_{a \in \mathcal{A}^{at}} I_a P_{\Psi_a}.
  \end{equation}
The last estimate together with \eqref{Hbotbnd} and \eqref{WU} give
that
\begin{equation}\label{Uedot}
V(\lambda) \doteq_M Q_N \sum_{a,b \in \mathcal{A}^{at}} P_{\Psi_a} I_a
(H^{\s,\bot}-\lambda)^{-1} I_b P_{\Psi_b} Q_N, \forall \lambda \leq
E(y)+\imsgap_1.
\end{equation}
Using \eqref{pairor},  \eqref{IapsiaL2est} and \eqref{pairorthnorm}
we obtain that
\begin{equation}\label{sumpsiaia1}
\|\sum_{a \in \mathcal{A}^{at}} P_{\Psi_a}  I_a \| \ls
\frac{M}{R^3}.
\end{equation}
In addition, \eqref{Hbotbnd} and \eqref{est:Ey} imply that there
exist $C_7>0$, depending only on $Z$, so that
\begin{equation}\label{Hbotest}
\Einfty \leq \E+\imsgap_1 \text{ and } (H^{\s,\bot}-\Einfty) \geq \gamma_1, \quad  \forall R \geq C_7 \ln M.
\end{equation}
 Using \eqref{Hbotbnd}, \eqref{ineqE},
 \eqref{Hbotest}, \eqref{Uedot} (applied at $\lambda=E(y)$ and at $\lambda=\Einfty$), \eqref{sumpsiaia1}
and the second resolvent formula we obtain that
\begin{equation}\label{UeUeinfty}
V(E(y))=V(\Einfty)+O(\frac{M^4}{R^{12}}), \quad \forall R \geq C_7 \ln M,
\end{equation}
where $C_7$ is the same as in \eqref{Hbotest}.
In the rest of the proof we will estimate $V(\Einfty)$.
  Using \eqref{Uedot}, \eqref{Hbotest} and \eqref{Iabreakup} we obtain that
\begin{equation}\label{UEdec}
V(\Einfty) \doteq_M \sum_{k<l}^{1,M} V_{kl}, \quad  \forall R \geq C_7 \ln M,
\end{equation}
where
\begin{equation}\label{Ukldef}
V_{kl}:= \sum_{a,b \in \mathcal{A}^{at}} Q_N P_{\Psi_a} I_a
(H^{\s,\bot}-\Einfty)^{-1} I_{B_k B_l} P_{\Psi_b} Q_N.
\end{equation}
We now fix $k,l$.
Recall that the resolvent $\Rakl$ was defined
in \eqref{Rbkldef}. We will prove that there exist $C_3>0$, depending only on $Z$,
 so that
\begin{equation}\label{Uklest}
\| V_{kl} - \Pi \lan I_{A_k A_l} \psi_{A_k} \psi_{A_l},
\Rakl I_{A_k A_l} \psi_{A_k} \psi_{A_l}\ran \| \ls \frac{M^2}{R^9}+
 \frac{M^2}{R^9} N^Z e^{-C_3 R}, \quad \forall R \geq C_7 \ln M,
\end{equation}
where note that the second term on the left hand side does not
depend on the decomposition $a \in \mathcal{A}^{at}$.
 From the second resolvent formula we obtain that
\begin{equation}\label{Ukldec}
V_{kl} = V_{kl,1}+ V_{kl,2},
\end{equation}
where
\begin{equation}\label{Ukl1def}
V_{kl,1}=\sum_{a,b \in \mathcal{A}^{at}} Q_N P_{\Psi_a} I_a \Rbkl
I_{B_k B_l} P_{\Psi_b} Q_N,
\end{equation}
and
\begin{align}\label{Ukl2def}
 V_{kl,2}=  \sum_{a,b \in \mathcal{A}^{at}} Q_N P_{\Psi_a} I_a
(H^{\s,\bot}-\Einfty)^{-1} D_{kl} \Rbkl I_{B_k B_l} P_{\Psi_b} Q_N,
\end{align}
with
\begin{equation}\label{Akldef}
D_{kl}:=-P^\bot H^\s P^\bot  +\Einfty+H_{B_k B_l}^{\s,\bot}-E_k
-E_l.
\end{equation}
   We will now estimate $ V_{kl,1} $ and afterwards $ V_{kl,2} $.

\paragraph{Estimate of $ V_{kl,1}$}
We will prove that
\begin{equation}\label{Ukl1est}
V_{kl,1} \doteq \Pi \lan I_{A_k A_l} \psi_{A_k} \psi_{A_l}, \Rakl
I_{A_k A_l} \psi_{A_k} \psi_{A_l}\ran.
\end{equation}
  Since $\Rbkl I_{B_k B_l}$ acts only on the variables in $B_k \cup B_l$,
  due to the cut off that has been applied to the ground states (see \eqref{Psiaeq}, \eqref{psiAj}), the
  summands in \eqref{Ukl1def} vanish, unless $A_j=B_j$, for all $j \neq
  k,l$. Hence, it is convenient to define in $\mathcal{A}^{at}$ the equivalence
relation $a \sim b \iff A_j=B_j, \forall j \neq k,l$. We denote the
set of equivalence classes by $X$. Since, as was pointed above, the
summands where $a,b$ are not in the same equivalence class vanish,
we obtain that
\begin{equation}\label{VUkl1}
V_{kl,1} =Q_N U_{kl,1} Q_N,
\end{equation}
where
\begin{equation}\label{zero}
U_{kl,1}=\sum_{D \in X} \sum_{a,b \in D} P_{\Psi_a} I_a \Rbkl I_{B_k
B_l} P_{\Psi_b}.
\end{equation}
If in \eqref{zero} we insert the decomposition $I_a=\sum_{i<j}
I_{A_i A_j}$, and use that $a \sim b$, it follows  from
\eqref{Psiaeq} and \eqref{onesymmetryenough} that all the terms
$I_{A_i A_j}$ have zero contribution in \eqref{zero} unless $\{i,j\}
= \{k,l\}$, or in other words that
\begin{equation}\label{Uklsim2}
U_{kl,1}=\sum_{D \in X} \sum_{a,b \in D} P_{\Psi_a} I_{A_k A_l }
\Rbkl I_{B_k B_l} P_{\Psi_b}.
\end{equation}
Splitting into the terms $a=b$ and $a \neq b$ we obtain that
\begin{equation}\label{Ukl1predec}
U_{kl,1}= \sum_{b \in \mathcal{A}^{at}} |\Psi_b \ran\lan I_{B_k B_l}
 \Psi_b, \Rbkl I_{B_k B_l} \Psi_b \ran \lan \Psi_b|+ R_{kl},
\end{equation}
where
\begin{equation}\label{Rkldefinition}
R_{kl}= \sum_{D \in X}  \left( \sum_{a,b \in D, a \neq b} P_{\Psi_a}
I_{A_k A_l } \Rbkl I_{B_k B_l} P_{\Psi_b} \right).
\end{equation}
We observe now that the inner products of the first term on the
right hand side of Equation \eqref{Ukl1predec} are independent of
$b$. This, together with \eqref{Psiaeq} and \eqref{P} implies that
\begin{equation}\label{Ukl1dec}
U_{kl,1}= P \lan I_{A_k A_l} \psi_{A_k} \psi_{A_l}, \Rakl I_{A_k
A_l} \psi_{A_k} \psi_{A_l} \ran+ R_{kl}.
\end{equation}
We will now prove that
\begin{equation}\label{Rklabsch}
 R_{kl} \doteq 0.
\end{equation}
To this end it is convenient to introduce the following projections.
  For $D \in X$, where recall $X$ is the set of equivalence classes,
we define $P_D=\sum_{b \in D} |\Psi_b \rangle \langle \Psi_b|$. From
\eqref{pairor} it follows that $P_D \Psi_b= \Psi_b$ for $b \in
  D$.  Using the last equation and \eqref{Rkldefinition}, we arrive at
\begin{equation}\label{Rklpc}
R_{kl}= \sum_{D \in X} P_D \left( \sum_{a,b \in D, a \neq b}
P_{\Psi_a} I_{A_k A_l } \Rbkl I_{B_k B_l} P_{\Psi_b} \right) P_D.
\end{equation}
If $A_D$ is a family of bounded operators, then
\begin{equation}\label{equivimp}
\| \sum_{D \in X} P_D A_D P_D \| \leq \max_{D \in X} \|A_D\|.
\end{equation}
The inequality follows from $\|B\|=\sup_{\|\phi\|,\|\psi\|=1} \lan
\phi, B \psi \ran$, the Cauchy-Schwarz inequality and by the fact
that $P_D$ are orthogonal projections with $P_{D_1} P_{D_2}=0$ if
$D_1 \neq D_2$. Since in \eqref{Rklpc} the sum over $a,b$ consists
of terms with $a \sim b$,  \eqref{equivimp} and
\eqref{Psiaeq} give
\begin{equation}\label{Rklabsch1}
\|R_{kl}\| \leq \max_{D \in X} \| \sum_{a,b \in D, a \neq b}
R_{kl,ab} \|,
\end{equation}
where
\begin{equation}\label{Rklabdef}
 R_{kl,ab}:= P_{\psi_{A_k} \psi_{A_l}}  I_{A_k A_l }  \Rbkl
I_{B_k B_l} P_{\psi_{B_k} \psi_{B_l}}.
\end{equation}
  Our goal is now to show that for $a \neq b$
\begin{equation}\label{Rklabsch2}
R_{kl,ab} \doteq 0.
\end{equation}

The functions $I_{B_k B_l} \psi_{B_k} \psi_{B_l}$ and $I_{A_k A_l}
\psi_{A_k} \psi_{A_l}$ have disjoint supports. However, the operator
$\Rbkl$  which separates them is non-local. As a result the function
$\Rbkl I_{B_k B_l} P_{\psi_{B_k} \psi_{B_l}}$ has generically
infinite support. We show, however, that its overlap with $I_{A_k
A_l} \psi_{A_k} \psi_{A_l}$ is exponentially small. To prove this,
we quantify the decay of the functions by introducing the
exponential weights $e^{\delta \varphi_{kl}}$ below, and push the
weights through $\Rbkl$ by using boosted Hamiltonians. We now
proceed with details. Recall that for a decomposition $b$, $x_{B_i}$
is the collection of the electron coordinates in $B_i$. Let
$\varphi_i(x_{B_i})\equiv \varphi_R(\langle x_{B_i} \rangle )$ be a
$C^2$ monotonically increasing function of $\langle x_{B_i}
\rangle$, where, recall, $\langle x_{B_i} \rangle=(1+|
x_{B_i}|^2)^{\frac{1}{2}}$, with $|\varphi_R'| \leq 1$ and  with
uniformly (in both $x_{B_i},R$) bounded second derivative and
satisfying,
\begin{align}\label{phi}
\varphi_i(x_{B_i})=\left\{ \begin{array}{ccc}   \langle x_{B_i}
\rangle, & \text{ if } \langle x_{B_i} \rangle \leq \frac{R}{2}-\frac{1}{e^2}  \\
\frac{R}{2}, & \text{ if } \langle x_{B_i} \rangle \geq
\frac{R}{2}+\frac{1}{e^2},
\end{array} \right.
\end{align}
Here $-e$ is the electron charge. The term $\frac{1}{e^2}$
originates from a natural scaling: the ground state $\phi_{B_i}$ of $H_{B_i}^\s$ is given by $\phi_{B_i}(x_{B_i})=e^{3 Z_i}\phi_{B_i,1}(e^2 x_{B_i})$, where $\phi_{B_i,1}$ is the ground state of $H_{B_i}^\s$ with $e=1$. Note that by construction of $\varphi_i$ we have that
\begin{equation}\label{varphigrad}
\|\nabla \varphi_i\|_{L^\infty} \leq 1.
\end{equation}
Let
\begin{equation}\label{varphiji}
\varphi_{kl}(x_{B_k}, x_{B_l}):= \varphi_k(x_{B_k}-y_k)+
\varphi_l(x_{B_l}-y_l),
\end{equation}
where the notation $x_{B_i}-y_i$ has the same meaning as in \eqref{phiAj}.
 For any $\delta >0$, we let
\begin{equation}\label{Hdelta}
\H_{B_k B_l, \delta}^{\s,\bot}:= e^{\delta \varphi_{kl}} \H_{B_k
B_l}^{\s,\bot} e^{-\delta \varphi_{kl}},
\end{equation}
where $\H_{B_k B_l}^{\s,\bot}$ was defined in \eqref{Hbkldef} (here $e$ is the usual one no longer minus the electron charge).
  In Section \ref{sec:Resdelta-est} we prove the following lemma:
\begin{lemma}\label{lem:Resdelta-est}
For all $k,l \in \{1,\dots,M\}$ with $k \neq l$, there exist $c,c_2>0$, depending only on $Z$, such that
if $R \geq c$ and $0 \leq \delta \leq c_2$, then
  $E_k+ E_l$ is in the resolvent set of $\H_{B_k B_l,  \delta}^{\s,\bot}$ and
 \begin{equation}\label{Resdelta-est}
\Rbkld:=(\H_{B_k B_l,  \delta}^{\s,\bot}-E_k-E_l)^{-1} =O(1).
\end{equation}
\end{lemma}
We assume this lemma for now. We choose now
\begin{equation}\label{def:delta0}
\delta_0=\min\{\frac{c_1}{2},c_2\},
\end{equation}
 where $c_1, c_2$ are the same as in \eqref{Psiadecay} and Lemma \ref{lem:Resdelta-est},
 respectively.
We now estimate $R_{kl,ab}$. The equality
\begin{equation}\label{edeltaeminusdelta}
 \Rbkl = e^{-\delta_0 \varphi_{kl} } \Rbkldo e^{\delta_0 \varphi_{kl}}
\end{equation}
 implies that
\begin{equation}\label{Uabdot}
\|R_{kl,ab}\| = |\langle e^{-\delta_0 \varphi_{kl}} I_{A_k A_l }
 \psi_{A_k} \psi_{A_l} ,\Rbkldo
 e^{\delta_0 \varphi_{kl}} I_{B_k B_l} \psi_{B_k} \psi_{B_l} \rangle|.
\end{equation}
 Similarly to \eqref{preIapsiaL2est} one can prove, due to \eqref{Psiadecay}, that we
 have
\begin{equation}\label{sigmaabest2}
\|e^{\delta_0 \varphi_{kl}} I_{B_k B_l } \psi_{B_k} \psi_{B_l}\|
 \lesssim \frac{1}{|y_k-y_l|^3}.
\end{equation}
 In addition, by the construction of $\varphi_{kl}$
 we have that $\varphi_{kl}=R$ on the support of $\psi_{A_k} \psi_{A_l}$, because $a \sim b$ and $a  \neq b$, and therefore
\begin{equation}\label{sigmaabest3}
 e^{-\delta_0 \varphi_{kl}} I_{A_k A_l} \psi_{A_k} \psi_{A_l} \doteq 0.
\end{equation}
Using \eqref{Resdelta-est}, \eqref{Uabdot}, \eqref{sigmaabest2} and
\eqref{sigmaabest3} we arrive at \eqref{Rklabsch2}. Observe now that
the cardinality of each equivalence class $D$ is $Z_k+ Z_l \choose
Z_k$,
 and in particular it depends only on $Z$. Therefore, \eqref{Rklabsch1} and
\eqref{Rklabsch2} imply \eqref{Rklabsch}. The estimate
\eqref{Rklabsch} together with \eqref{Ukl1dec}, \eqref{VUkl1} and
\eqref{Piqp} implies \eqref{Ukl1est}.

\paragraph{Estimate of $V_{kl,2}$}
We will now prove that there exists $c>0$ so that
\begin{equation}\label{Ukl2est}
\|V_{kl,2}\| \ls  \frac{M^2}{R^9}+ \frac{M^2}{R^9} N^Z e^{-cR}, \quad \forall R \geq C_7 \ln M,
\end{equation}
where $C_7$ is the same constant as in \eqref{Hbotest}.
Observe that, since the projections $P^\bot, Q_N$ commute with
$(H^{\s,\bot}-E(y))$, by \eqref{chiaest2} $V_{kl,2}$ remains the same
if we replace $P^\bot H^\s P^\bot$ appearing in the definition of
$D_{kl}$ with $H P^\bot$. Similarly, we can replace $H_{B_k
B_l}^{\s,\bot}$ with $ P_{\psi_{B_k} \psi_{B_l}}^\bot H_{B_k B_l}$.
Since, moreover, $\Einfty=\sum_{j=1}^M E_j$ we obtain that
\begin{equation}\label{VUkl2}
V_{kl,2}=Q_N U_{kl,2} Q_N,
\end{equation}
where
\begin{equation}\label{Ukl2eq}
U_{kl,2}=\sum_{a,b \in \mathcal{A}^{at}} P_{\Psi_a} I_a
(H^{\s,\bot}-\Einfty)^{-1} G_{kl} \Rbkl I_{B_k B_l} P_{\Psi_b},
\end{equation}
and
\begin{equation}\label{Gkldef}
G_{kl}:= -H P^\bot+P_{\psi_{B_k} \psi_{B_l}}^\bot H_{B_k B_l}+
\sum_{j \neq k,l}E_j.
\end{equation}
Therefore, using  \eqref{sumpsiaia1} and \eqref{Hbotest}, we obtain
that
\begin{equation}\label{Ukl2dec}
\|U_{kl,2}\| \ls \frac{M}{R^3} \|U_{kl,3}\|, \quad \forall R \geq C_7 \ln M,
\end{equation}
where
\begin{equation}\label{Ukl3def}
U_{kl,3}:=\sum_{b \in \mathcal{A}^{at}} G_{kl} \Rbkl I_{B_k B_l}
P_{\Psi_b}.
\end{equation}
To estimate $U_{kl,3}$ we use that $P_{\psi_{B_k} \psi_{B_l}}^\bot
H_{B_k B_l}=H_{B_k B_l}-P_{\psi_{B_k} \psi_{B_l}} H_{B_k B_l}$ and
that $-H P^\bot=HP-H=HP -H_b - I_b \chi_b - (1-\chi_b) I_b$, where
\begin{equation}\label{chib}
\chi_b=\chi_R^b:= \otimes_{j=1}^M \chi_{2R}^{B_j},
\end{equation}
and $\chi_{2R}^{B_j}$ was defined in \eqref{charaj},
 to obtain that
\begin{equation}\label{Gklest}
G_{kl}= (-H_b+H_{B_{k} B_l}+\sum_{j \neq k,l} E_i) +H P- I_b \chi_b
- (1-\chi_b) I_b-P_{\psi_{B_k} \psi_{B_l}} H_{B_k B_l}.
\end{equation}
Since $\Rbkl I_{B_k B_l}$ acts only on the coordinates in $B_k \cup
B_l$, if we decompose $H_b$ to the sum of the Hamiltonias of the
atoms $\sum_{j=1}^M H_{B_j}$ (see \eqref{Ha}), all $H_{B_j},j \neq k,l$ act directly
to their cut off ground states. Therefore, we obtain, using \eqref{almosteigat}, that
\begin{equation}\label{Hbsimplification}
H_b \Rbkl I_{B_k B_l} P_{\Psi_b} \doteq_M (H_{B_{k} B_l}+\sum_{j
\neq k,l} E_i) \Rbkl I_{B_k B_l} P_{\Psi_b}.
\end{equation}
 Next we use  \eqref{Ukl3def}, \eqref{Gklest} and \eqref{Hbsimplification} to arrive at
\begin{equation}\label{Ukl3dec}
U_{kl,3} \doteq_M U_{kl,4}- U_{kl,5}- U_{kl,6}-U_{kl,7},
\end{equation}
where
\begin{align}\label{Ukl4def}
U_{kl,4}&=\sum_{b \in \mathcal{A}^{at}} HP \Rbkl I_{B_k B_l}
P_{\Psi_b},\\ \label{Ukl5def} U_{kl,5}&=\sum_{b \in
\mathcal{A}^{at}} \chi_b I_b \Rbkl I_{B_k B_l} P_{\Psi_b},
\\\label{Ukl6def} U_{kl,6}&=\sum_{b \in \mathcal{A}^{at}} (1-\chi_b)
I_b \Rbkl I_{B_k B_l} P_{\Psi_b}, \\ \label{Ukl7def}
U_{kl,7}&=\sum_{b \in \mathcal{A}^{at}} P_{\psi_{B_k} \psi_{B_l}}
H_{B_k B_l} \Rbkl I_{B_k B_l} P_{\Psi_b}.
\end{align}
We begin with estimating $U_{kl,4}$. Using $P=PP$ and $\|HP\| \ls
M$, we obtain that
\begin{align}\label{Ukl4est1}
\|U_{kl,4}\| \ls M \|\sum_{b \in \mathcal{A}^{at}} P \Rbkl I_{B_k
B_l} P_{\Psi_b}\|.
\end{align}
Furthermore, observe that since $\Rbkl$ acts on the variables in
$B_k \cup B_l$ only, we have that
\begin{equation}\label{PPc}
P \Rbkl I_{B_k B_l} P_{\Psi_b}=P_D \Rbkl I_{B_k B_l} P_{\Psi_b} P_D,
\forall b \in \mathcal{A}^{at},
\end{equation}
where $D$ is the equivalence class of $b$, and we have also used
that $P_D \Psi_b= \Psi_b$ to insert $P_D$ on the right. If on the
right hand side of \eqref{Ukl4est1} we write $\sum_{b \in
\mathcal{A}^{at}}=\sum_{D \in X} \sum_{b \in D}$, where recall that
$X$ is the set of equivalence classes, and use \eqref{PPc} and
\eqref{equivimp} we obtain that
\begin{equation}\label{Ukl4est2}
\|U_{kl,4}\| \ls  M  \max_{D \in X} \| P_D \left(\sum_{b \in D}
\Rbkl I_{B_k B_l}  P_{\Psi_b} \right) P_D\|.
\end{equation}
Similarly to \eqref{chiaest2} we can show that
\begin{equation}\label{PbotIBkl}
I_{B_k B_l} P_{\Psi_b}= P_{\psi_{B_k} \psi_{B_l}}^\bot I_{B_k B_l}
P_{\Psi_b}.
\end{equation}
It is easy to show that $[P_{\psi_{B_k} \psi_{B_l}}^\bot, \Rbkl]
=0$, which together with \eqref{PbotIBkl} gives $P_{\Psi_b} \Rbkl
I_{B_k B_l} \Psi_b = 0$. The last equation together with
\eqref{Ukl4est2} implies that
\begin{equation}
\|U_{kl,4}\| \ls M   \max_{D \in X} \| P_D \left(\sum_{b \in D}
(P_D- P_{\Psi_b}) \Rbkl I_{B_k B_l}  P_{\Psi_b} \right) P_D\|.
\end{equation}
Therefore, using  that $(P_D- P_{\Psi_b}) e^{-\delta_0 \phi_{kl}}
\doteq 0$, where $\delta_0$ was defined in \eqref{def:delta0}, proceeding as in the proof of \eqref{Rklabsch2}, and
using that the cardinality of $D$ depends only on $Z_k$ and $Z_l$ we
obtain that
\begin{equation}\label{Ukl4est}
\|U_{kl,4}\| \doteq_M 0.
\end{equation}

We next estimate $U_{kl,5}$. Using \eqref{pairorthnorm}, which can
be applied because $\chi_b \chi_a=0$ when $a \neq b$, and because
\eqref{pairor} holds, we arrive at
\begin{equation}\label{Ukl5est1}
\|U_{kl,5} \| \leq  \max_{b \in \mathcal{A}^{at}} \| \chi_b I_b
\Rbkl I_{B_k B_l}  P_{\Psi_b} \|.
\end{equation}
Therefore, proceeding as in the proof of \eqref{Uabdot}, we obtain
that
\begin{equation}\label{Ukl5est2}
\|U_{kl,5} \| \leq  \max_{b \in \mathcal{A}^{at}} \| e^{-\delta_0
\varphi_{kl}}\chi_b I_b \Rbkldo e^{\delta_0 \varphi_{kl}} I_{B_k B_l}
P_{\Psi_b} \|.
\end{equation}
Using the last inequality, together with \eqref{Psiaeq},
\eqref{sigmaabest2}, \eqref{Resdelta-est} and the inequality
\begin{equation}\label{est:edeltaIb}
 \|e^{-\delta_0
\varphi_{kl}} I_b \prod_{j \neq k,l} \psi_{B_j}\| \ls
\frac{M}{R^3},
\end{equation}
which can be proven similarly to
\eqref{IapsiaL2est}, we obtain that
\begin{equation}\label{Ukl5est}
\|U_{kl,5} \| \ls  \frac{M}{R^6}.
\end{equation}

We now estimate $U_{kl,6}$. In particular, we will show that
there exists $c$, so that
\begin{equation}\label{Ukl6est}
\|U_{kl,6}\| \ls \frac{M}{R^6} N^Z e^{-cR}.
\end{equation}
To this end we need to define a new equivalence
relation $\sim'$ on $\mathcal{A}^{at}$. Let $a,b \in \mathcal{A}^{at}$. We say that $a \sim' b$
if $A_k=B_k$ and $A_l=B_l$. We denote an equivalance class originating from this equivalence
relation by $D'$ and the set of the equivalance classes by $X'$. With the help of elementary
combinatorics, one can verify that
\begin{equation}\label{est:X'card}
 |X'| \leq N^{2Z}.
\end{equation}
For $D' \in X'$ we define,
$P_{D'}:=\sum_{a \in D'} P_{\Psi_a}$. From \eqref{pairor} it follows that
$P_{\Psi_b} P_{D'}=P_{\Psi_b}$, for all $b \in D', D' \in X'$. Therefore, from \eqref{Ukl6def}
we obtain that
\begin{equation}\label{Ukl6rewrite}
 U_{kl,6}= \sum_{D' \in X'} A_{D'} P_{D'},
\end{equation}
where
\begin{equation}\label{def:AD'}
 A_{D'}= \sum_{b \in D'}  (1-\chi_b) I_b
\Rbkl I_{B_k B_l} P_{\Psi_b}.
\end{equation}
Using \eqref{Ukl6rewrite} and that $P_{D_1'} P_{D_2'}=0$ if $D_1' \neq D_2'$, one can verify,
by arguing similarly as in verifying \eqref{equivimp}, that
\begin{equation}\label{Ukl6est1}
 \|U_{kl,6}\| \leq |X'|^{\frac{1}{2}} \max_{D' \in X'} \| A_{D'} \|.
\end{equation}
We shall now estimate $A_{D'}$. Using \eqref{Psiaeq}, we can  write
\begin{equation}\label{ADrewrite1}
 A_{D'}= \sum_{b \in D'}  \bigg|(1-\chi_b) I_b \prod_{j=1, j \neq k,l}^M \psi_{B_j}
\Rbkl I_{B_k B_l} \psi_{B_k} \psi_{B_l} \bigg\rangle \bigg\langle \Psi_b \bigg|,
\end{equation}
because the functions $\psi_{B_j}$ commute with the resolvent $\Rbkl$
for $j \neq k,l$. From \eqref{charaj} and \eqref{psiAj} it follows
that
\begin{equation}\label{eqn:chirbrestriction}
\chi_{2R}^{B_j} \psi_{B_j}=\psi_{B_j}, \text{ } \forall j=1,\dots,M, \forall b \in \mathcal{A}^{at}.
\end{equation}
Due to the last equation and \eqref{chib}, we can replace $\chi_b$ on
 the right hand side of \eqref{ADrewrite1}
with $\chi_{2R}^{B_k}\chi_{2R}^{B_l}$. This gives that
\begin{equation}\label{ADrewrite}
 A_{D'}=\sum_{b \in D'} | F_b \rangle \langle \Psi_b|,
\end{equation}
where
\begin{equation*}
F_b:=(1-\chi_{2R}^{B_k}\chi_{2R}^{B_l}) I_b \prod_{j=1, j \neq k,l}^M \psi_{B_j}
\Rbkl I_{B_k B_l} \psi_{B_k} \psi_{B_l}.
\end{equation*}
We will show that
\begin{equation}\label{suppFaFb}
\text{supp}(F_a) \cap \text{supp}(F_b)=\emptyset, \text{ } \forall a, b \in D', \text{ with } a \neq b, \forall D' \in X'.
\end{equation}
  Indeed, let $D' \in X'$ and $a,b \in D'$. Then $A_k=B_k, A_l=B_l$ so that the resolvents
$\Rakl, \Rbkl$ act on the same coordinates. Since $\psi_{A_k}=\psi_{B_k}$ and $\psi_{A_l}=\psi_{B_l}$,
we obtain from \eqref{pairor} that
$$\text{supp}\left(\prod_{j=1, j \neq k,l}^M \psi_{A_j}\right) \cap \text{supp}\left(\prod_{j=1, j \neq k,l}^M \psi_{B_j}\right)=\emptyset,$$
which in turn implies \eqref{suppFaFb}.
 Due to \eqref{suppFaFb} and \eqref{pairor}, we can apply \eqref{pairorthnorm} to \eqref{ADrewrite} to obtain that
\begin{equation*}
 \|A_{D'} \| \leq \max_{b \in D'}  \|F_b\|.
\end{equation*}
We can estimate $F_b$ in a similar way as the right hand side of \eqref{Ukl5est1} (see \eqref{Ukl5est})
to obtain that there exists $c$ so that
\begin{equation}\label{ADest}
 \|A_{D'} \| \ls \frac{M}{R^6} e^{-cR}, \text{ for all } D \in X'
\end{equation}
We note that the extra factor $e^{-cR}$, which does not exist on the right hand side
of \eqref{Ukl5est}, was gained by the fact that
$\|(1-\chi_{2R}^{B_k}\chi_{2R}^{B_l})e^{-\frac{\delta_0}{2} \varphi_{kl}}\|_{L^\infty} \doteq 0$.
From \eqref{est:X'card}, \eqref{Ukl6est1} and \eqref{ADest} we obtain \eqref{Ukl6est}.

 Now we estimate $U_{kl,7}$. We observe that each of the summands
 in \eqref{Ukl7def} remains invariant if multiplied with $P_{\Psi_b}$ on
 the left, because $\Rbkl I_{B_k B_l}$ acts only on the coordinates in $B_k \cup B_l$.
 Therefore,
 \begin{equation}
U_{kl,7}=\sum_{b \in \mathcal{A}^{at}} P_{\Psi_b}  P_{\psi_{B_k}
\psi_{B_l}} H_{B_k B_l} \Rbkl I_{B_k B_l} P_{\Psi_b},
 \end{equation}
 which together with \eqref{pairorthnorm} implies that
 \begin{equation*}
\|U_{kl,7}\| \leq \max_{b \in \mathcal{A}^{at}} \|P_{\psi_{B_k}
\psi_{B_l}} H_{B_k B_l} \Rbkl I_{B_k B_l} P_{\Psi_b}\|.
 \end{equation*}
Therefore, since $P_{\psi_{B_k} \psi_{B_l}}^\bot$ commutes with
$\Rbkl$, we obtain, by \eqref{PbotIBkl} that
\begin{equation*}
\|U_{kl,7}\| \leq \max_{b \in \mathcal{A}^{at}} \|P_{\psi_{B_k}
\psi_{B_l}} H_{B_k B_l} P_{\psi_{B_k} \psi_{B_l}}^\bot \Rbkl I_{B_k
B_l} P_{\Psi_b}\|,
 \end{equation*}
 which together with \eqref{almosteigat} gives that
\begin{equation}\label{Ukl7est}
\|U_{kl,7}\| \doteq 0.
\end{equation}
 By
\eqref{Ukl3dec}, \eqref{Ukl4est}, \eqref{Ukl5est}, \eqref{Ukl6est}
and \eqref{Ukl7est} we obtain that there exists $c$ so that
\begin{equation}\label{Ukl3est}
\|U_{kl,3}\|  \ls \frac{M}{R^6}+ \frac{M}{R^6} N^Z e^{-cR}.
\end{equation}
From \eqref{Ukl3est}, \eqref{Ukl2dec} and \eqref{VUkl2}
we obtain \eqref{Ukl2est}. Moreover, \eqref{Uklest} follows from \eqref{Ukldec}, \eqref{Ukl1est} and
\eqref{Ukl2est}. Lemma \ref{lem:Uentwicklung} follows
 from \eqref{UeUeinfty}, \eqref{UEdec} and  \eqref{Uklest}.
\end{proof}

In Lemma \ref{lem:Uentwicklung} we estimated $V(E(y))|_{\Ran \Pi}$. The following lemma will
help us to obtain more precise information on $V(E(y))|_{\Ran \Pi}$.

\begin{lemma}\label{lem:Uaagenauer}
  For all $i,j \in \{1, 2,\dots, M\}$ with $i<j$, we have that
\begin{equation}\label{Uaagenauer}
r_{ij}:=\lan I_{A_i A_j} \psi_{A_i} \psi_{A_j}, \Raij I_{A_i A_j} \psi_{A_i}
\psi_{A_j}\ran=\frac{ e^4
\sigma_{ij}}{|y_i-y_j|^6}+O(\frac{1}{|y_i-y_j|^7}),
\end{equation}
where $\sigma_{ij}$ was defined in \eqref{def:sigmaij}.
\end{lemma}
\begin{proof}
 By the change of variables \eqref{zlm} and from
\eqref{exp:IAiAj} and \eqref{Rbkldef}, it follows that
\begin{equation}\label{IAkAlskl}
r_{ij} = \frac{e^4}{|y_i-y_j|^6} \langle f_{ij, \widehat{y_{ij}}} \psi_{i} \otimes \psi_{j},
(P_{\psi_{i} \otimes \psi_{j}}^\bot H_{ij}^{\s} P_{\psi_{i} \otimes \psi_{j}}^\bot - E_i-E_j)^{-1} f_{ij,\widehat{y_{ij}}} \psi_{i} \otimes \psi_{j}
\rangle + O(\frac{1}{|y_i-y_j|^7}),
\end{equation}
where recall that $\widehat{y_{ij}}=\frac{y_i-y_j}{|y_i-y_j|}$, $\psi_i, \psi_j$
where defined in \eqref{psiAj}, and $H_{ij}^\s$ was defined in \eqref{def:Hijs}.
We will show now that the leading term on the right hand side of
\eqref{IAkAlskl} fulfills the estimate
\begin{equation}\label{fastsigma}
\langle f_{ij, \widehat{y_{ij}}} \psi_{i} \otimes \psi_{j},
(P_{\psi_{i} \otimes \psi_{j}}^\bot H_{ij}^{\s} P_{\psi_{i} \otimes \psi_{j}}^\bot - E_i-E_j)^{-1} f_{ij,\widehat{y_{ij}}} \psi_{i} \otimes \psi_{j}
\rangle \doteq \sigma_{ij}(\widehat{y_{ij}}),
\end{equation}
where $\sigma_{ij}(\widehat{y_{ij}}) $ was defined in
\eqref{sigmaijvdef}. Indeed, the right hand side and the left hand side
of \eqref{fastsigma} differ only because the cut off ground states
$\psi_i, \psi_j$ appear on the left instead of the exact ground states
$\phi_i, \phi_j$ that appear in the definition of $\sigma_{ij}(\widehat{y_{ij}}) $.
Therefore, using \eqref{phijpsij}, we obtain \eqref{fastsigma}.
 From \eqref{IAkAlskl}, \eqref{fastsigma} and \eqref{def:sigmaij} we obtain
\eqref{Uaagenauer}.
\end{proof}

 Lemmas \ref{lem:Uentwicklung} and \ref{lem:Uaagenauer}
 imply Proposition \ref{prop:Feshest}. Since in the proof of Lemma \ref{lem:Uentwicklung} we assumed Lemma
\ref{lem:Resdelta-est}, it remains to prove the latter and this is what we do next.

\subsection{Proof of Lemma \ref{lem:Resdelta-est}} \label{sec:Resdelta-est}
 For any $\delta>0$ and any operator $K$ acting on the coordinates in $B_k \cup B_l $, we let
\begin{equation}\label{Adelta}
K_{\delta}:= e^{\delta \varphi_{kl}} K e^{-\delta \varphi_{kl}},
\text{ and } K_\delta^\bot:=(K^\bot)_{\delta}.
\end{equation}
If the operator has indices e.g. $K_{mn}$ we define
$K_{mn,\delta}:=(K_{mn})_{\delta}$.
 We also define $\Delta_{B_k B_l}=\sum_{m \in B_k \cup B_l}
\Delta_{x_m}$, $\Delta_{B_k B_l,\delta}=(\Delta_{B_k B_l})_\delta$,
and similarly $\nabla_{B_k B_l}$ to be the gradient in $\mathbb{R}^{3(Z_k +
Z_l)}$. We begin with two auxiliary lemmas.
\begin{lemma}\label{HdeltaPdeltaP}
There exists $c$ so that for all $\delta \leq c$ we have
\begin{equation}\label{HdeltaminusH-bnd1}
\|(1-\Delta_{B_k B_l})^{-\frac{1}{2}} (H_{B_k B_l, \delta}^\s-H_{B_k
B_l}^\s)(1-\Delta_{B_k B_l})^{-\frac{1}{2}}\| \ls  \delta.
\end{equation}
\end{lemma}
\begin{proof}
Note that to prove \eqref{HdeltaminusH-bnd1} we can disregard $\s$
in the Hamiltonian because the antisymmetrizing projections $Q_{B_k},
Q_{B_l}$ commute with the Laplacians. It is therefore enough to
prove that there exists $c$ so that  for all $\delta \leq c$ we have
\begin{equation}
|\langle \Phi,(H_{B_k B_l, \delta}-H_{B_k B_l}) \Psi \rangle| \ls
\delta \|\Phi \|_{H^1} \| \Psi \|_{H^1}, \forall \Phi, \Psi \in
H^1(\mathbb{R}^{3(Z_k+Z_l)}).
\end{equation}
Observe that \begin{equation}\label{HdeltaminusH1} H_{B_k B_l,
\delta}-H_{B_k B_l}=-\Delta_{B_k B_l,\delta}+\Delta_{B_k B_l}.
\end{equation}
   Furthermore, an elementary computation
   gives that
\begin{equation*}
\langle \Phi,(-\Delta_{B_k B_l,\delta}+\Delta_{B_k B_l}) \Psi
\rangle $$$$= \delta \left[\langle (\nabla_{B_k B_l} \varphi_{kl})
\Phi, \nabla_{B_k B_l} \Psi \rangle - \langle \nabla_{B_k B_l} \Phi,
(\nabla_{B_k B_l} \varphi_{kl}) \Psi \rangle - \delta \langle \Phi,
|\nabla_{B_k B_l} \varphi_{kl}|^2 \Psi \rangle \right].
\end{equation*}
The  last equation together with \eqref{varphigrad} yields the
desired result.
\end{proof}

\begin{lemma}\label{PdeltaminusPlemma}
    For $k,l \in \{1,\dots,M\}$ with $k \neq l$ we define $P_{kl}:= P_{\psi_{B_k} \psi_{B_l}}$.
     Then there exists
$c$ so that if $\delta \leq c$, then
\begin{equation}\label{1minusdeltaPdeltaminusP}
\| P_{kl,\delta}-P_{kl}\| \lesssim \delta,  \quad \|H_{B_k B_l}
(P_{kl,\delta}-P_{kl})\| \lesssim \delta, \quad \|H_{B_k B_l,
\delta} (P_{kl,\delta}-P_{kl})\| \lesssim \delta.
\end{equation}
\end{lemma}
\begin{proof}
Let $g(\delta):=P_{kl,\delta}-P_{kl}.$ Using \eqref{Psiadecay} and the dominated convergence theorem, one can show that
$g \in C^1([0,\frac{c_1}{2}]; B(L^2(\mathbb{R}^{3(Z_k+Z_l)})))$, that
\begin{equation}
g'(\delta)=\varphi_{kl} e^{\delta \varphi_{kl}} P_{kl} e^{-\delta
\varphi_{kl}}- e^{\delta \varphi_{kl}} P_{ij} \varphi_{kl}
e^{-\delta \varphi_{kl}}, \quad \forall \delta \in [0,\frac{c_1}{2}]
\end{equation}
and  that there exists $c$ so that $\|g'(\delta)\| \leq c$, for all $\delta \leq \frac{c_1}{2}$.
 Since $g(0)=0$, by  the fundamental theorem of calculus, we obtain,
 for $\delta \leq \frac{c_1}{2}$, that
\begin{equation}
\|g(\delta)\| \lesssim \delta.
\end{equation}
This implies the first inequality in
\eqref{1minusdeltaPdeltaminusP}. We now prove the second inequality
in \eqref{1minusdeltaPdeltaminusP}. We write $E_{kl}:=E_k+E_l$.
     Applying the Leibnitz rule and using \eqref{almosteigat}
      we obtain that there exists $c$ so that
\begin{align}\notag
\| H_{B_k B_l} (P_{kl, \delta}-P_{kl})\| \leq \| E_{kl}& (P_{kl,
\delta}-P_{kl}) \| \\ \label{leqdelta1} + \|  \nabla_{B_k B_l}
(e^{\delta \varphi_{kl}}) & \cdot \nabla_{B_k B_l} P_{kl} e^{-\delta
\varphi_{kl}}\| + \| (\Delta_{B_k B_l} e^{\delta \varphi_{kl}})
P_{kl} e^{-\delta \varphi_{kl}} \|+O(e^{-cR}).
\end{align}
Therefore, using \eqref{Psiadecay} and the first inequality in
\eqref{1minusdeltaPdeltaminusP} we arrive at the second inequality
in \eqref{1minusdeltaPdeltaminusP} for $\delta \leq \frac{c_1}{2}$.
 The last inequality in \eqref{1minusdeltaPdeltaminusP} follows from
the second and the $H_{B_k B_l}$ boundedness of $H_{B_k B_l,
\delta}$, which in turn follows from \eqref{HdeltaminusH1}.
\end{proof}
We will now prove Lemma \ref{lem:Resdelta-est}.
 We observe that there exists $C,c$ so that
 \begin{equation}\label{est:positivity}
 \H_{B_k
B_l}^{\s,\bot}-E_{kl} \geq c, \forall R \geq C,
\end{equation}
 where the assumption on $R$
is necessary to obtain the gap $c$ due to the cut off of the ground states.
Using \eqref{est:positivity} together with the decomposition
 $\H_{B_k B_l, \delta}^{\s,\bot}-E_{kl}=\H_{B_k
B_l}^{\s,\bot}-E_{kl}+ \H_{B_k B_l, \delta}^{\s,\bot}-\H_{B_k
B_l}^{\s,\bot}$,
 we obtain that
\begin{equation}\label{HdeltaminusE}
\H_{B_k B_l, \delta}^{\s,\bot}-E_{kl} = (\H_{B_k
B_l}^{\s,\bot}-E_{kl})^{\frac{1}{2}}(I+ K) (\H_{B_k
B_l}^{\s,\bot}-E_{kl})^{\frac{1}{2}}, \forall R \geq C,
\end{equation}
where
\begin{equation}\label{Kdef}
K:=(\H_{B_k B_l}^{\s,\bot}-E_{kl})^{-\frac{1}{2}} (\H_{B_k B_l,
\delta}^{\s,\bot}-\H_{B_k B_l}^{\s,\bot}) (\H_{B_k
B_l}^{\s,\bot}-E_{kl})^{-\frac{1}{2}}.
\end{equation}
 We show that if $\delta$ is small enough, and $R \geq C$ then $I+ K$ is invertible, and we estimate its inverse. First, since $\H_{B_k
B_l}^{\s,\bot}:=P_{kl}^\bot H_{B_k B_l}^\s P_{ij}^\bot$, $\H_{B_k
B_l, \delta}^{\s,\bot}:=(\H_{B_k B_l}^{\s,\bot})_{\delta}$ and
$(P_{kl,\delta}^\bot-P_{kl}^\bot)=P_{kl}-P_{kl,\delta}$, we obtain
that
\begin{equation}\label{Hdeltadifdecompo}
 (H_{B_k B_l, \delta}^{\s,\bot}-H_{B_k B_l}^{\s,\bot})
$$$$ =  -P_{kl,\delta}^\bot H_{B_k B_l, \delta}^\s
(P_{kl,\delta}-P_{kl}) + P_{kl,\delta}^\bot  (H_{B_k B_l,
\delta}^\s-H_{B_k B_l}^\s) P_{kl}^\bot - (P_{kl,\delta}-P_{kl})
H_{B_k B_l}^\s P_{kl}^\bot.
\end{equation}
       Using \eqref{HdeltaminusH-bnd1}, $\| Q_{B_k} Q_{B_l} (1-\Delta_{B_k B_l})^{\frac{1}{2}}
(H_{B_k B_l}^{\s,\bot}-E_{kl})^{-\frac{1}{2}} \| \lesssim 1$ and
\eqref{Psiadecay} we obtain that for $\delta$ small enough
\begin{equation}\label{hilfreich}
\|(\H_{B_k B_l}^{\s,\bot}-E_{kl})^{-\frac{1}{2}} P_{kl,\delta}^\bot
(H_{B_k B_l, \delta}^\s-H_{B_k B_l}^\s) P_{ij}^\bot (\H_{B_k
B_l}^{\s,\bot}-E_{kl})^{-\frac{1}{2}} \| \ls \delta, \quad \forall R \geq C.
\end{equation}
Since, moreover, on the right hand side of \eqref{Hdeltadifdecompo}
the first and third terms
      are estimated by the third and second inequality in \eqref{1minusdeltaPdeltaminusP},
  respectively, we obtain, using \eqref{est:positivity}, \eqref{Kdef}, \eqref{Hdeltadifdecompo} and
  \eqref{hilfreich}, that for $\delta$ small enough and $R \geq C$
\begin{equation}\label{UabEest3}
\|K\| \lesssim \del.
\end{equation}
We use the last estimate and take $\delta$ small enough to obtain
that $\|K\| \leq \frac{1}{2}.$ This shows that $I+K$ is invertible
and its inverse is bounded by $2$. This together with
\eqref{est:positivity} and \eqref{HdeltaminusE} gives \eqref{Resdelta-est}, for $\delta$ small enough
and $R \geq C$.



\section{Proof of Theorem \ref{thm:upper}}\label{proof:thmupper}

  A reasonable generalization of the test function as described in
  the sketch of the proof for two atoms in the introduction, is the normalized antisymmetrization of
  the function
\begin{equation}\label{testdef}
\tilde \Psi_b=\Psi_b-\sum_{k<l} \chi_{B_k B_l} \Rbkl I_{B_k B_l}
\Psi_b,
\end{equation}
where $b \in \mathcal{A}^{at}$ and recall that $\Rbkl$ was defined
in \eqref{Rbkldef}. The cut off functions $\chi_{B_k
B_l}=\chi_{2R}^{B_k} \chi_{2R}^{B_l}$, with $\chi_{2R}^{B_m}$
defined in \eqref{charaj}, together with the cut off introduced to
construct the functions $\Psi_b$ (see \eqref{Psiaeq} and
\eqref{psiAj}) ensure that
\begin{equation}\label{tildepsiab}
\text{supp}(\tilde \Psi_a) \cap \text{supp}(\tilde \Psi_b)=\emptyset, \text{ } \forall a, b \in \mathcal{A}^{at} \text{ with } a \neq b.
\end{equation}
The dilation $2R$ of the characteristic function was chosen to
ensure that, by \eqref{eqn:chirbrestriction},
\begin{equation}\label{chipsi}
\chi_{B_k B_l}=1 \text{ on } \supp \Psi_b.
\end{equation}
Similarly to \eqref{QNpsia} one can show that
\begin{equation}\label{tildeQNpsia}
Q_N \tilde \Psi_a=\frac{1}{|\mathcal{A}^{at}|}\sum_{b \in \mathcal{A}^{at}} \text{sgn}(b,a) \tilde \Psi_b.
\end{equation}
Since the interaction energy of the system is the ground state
energy of $Q_N(H-\Einfty)Q_N$, we have that
\begin{equation*}
W(y) \leq \frac{1}{\|Q_N \tilde \Psi_b\|^2} \langle Q_N \tilde \Psi_b ,
(H-\Einfty) Q_N \tilde \Psi_b\rangle.
\end{equation*}
Due to \eqref{tildepsiab} and \eqref{tildeQNpsia} and the fact that
$\text{supp}((H-\Einfty) \Psi_a) \subset \text{supp}(\Psi_a), \forall a \in \mathcal{A}^{at}$,
it turns out that we can drop the anti-symmetrization projection
$Q_N$ to obtain that
\begin{equation}\label{Wupp}
W(y) \leq \frac{1}{\|\tilde \Psi_b\|^2} \langle  \tilde \Psi_b ,
(H-\Einfty) \tilde \Psi_b\rangle.
\end{equation}
By \eqref{Hadecomp}, \eqref{almosteig} and \eqref{chipsi} we obtain that $\chi_{B_k
B_l}(H-\Einfty)\Psi_b = (H-\Einfty)\Psi_b \doteq_M I_b \Psi_b$ for
all $k<l$ and therefore, using \eqref{testdef} and
\eqref{psiaIbpsib}, we obtain that
\begin{equation}\label{proddec}
\langle  \tilde \Psi_b , (H-\Einfty) \tilde \Psi_b\rangle \doteq_M
-2 \sum_{k<l} Re \langle I_b \Psi_b, \Rbkl I_{B_k B_l} \Psi_b
\rangle+D_1+D_2,
\end{equation}
where
\begin{equation}\label{D1def}
D_1=\sum_{i<j} \sum_{k<l} \langle \chi_{B_k B_l} \Rbkl I_{B_k B_l}
\Psi_b, (H_b-\Einfty)\chi_{B_i B_j} \Rbij I_{B_i B_j} \Psi_b
\rangle,
\end{equation}
and
\begin{equation}\label{D2def}
D_2=\sum_{i<j} \sum_{k<l} \langle  \Rbkl I_{B_k B_l} \Psi_b,
\chi_{B_k B_l} I_b \chi_{B_i B_j} \Rbij I_{B_i B_j} \Psi_b \rangle.
\end{equation}
The decomposition $D_1+D_2$ is based simply on the decomposition
$H-\Einfty=(H_b-\Einfty)+I_b$. Now we will estimate the different
terms on the right hand side of \eqref{proddec}.

 For the first term
we write $I_b=\sum_{i<j} I_{B_i B_j}$ and use
\eqref{onesymmetryenough}, which implies that the contribution of
$I_{B_i B_j}$ in the first term is zero unless $\{i,j\}=\{k,l\}$. In
other words, we obtain that
\begin{equation}\label{IbIBkBl}
 \langle I_b \Psi_b, \Rbkl I_{B_k B_l} \Psi_b \rangle= \langle I_{B_k B_l} \psi_{B_k} \psi_{B_l} , \Rbkl I_{B_k B_l} \psi_{B_k} \psi_{B_l}
 \rangle,
\end{equation}
where we have also used \eqref{Psiaeq} and that $\Rbkl I_{B_k B_l}$
acts only on the coordinates in $B_k \cup B_l$. We now estimate
$D_1$. Proceeding as in the proof of \eqref{Rklabsch2}, we obtain
that
\begin{equation*}
\|(1-\chi_{B_k B_l}) \Rbkl I_{B_k B_l} \psi_{B_k} \psi_{B_l}\|_{H^1}
\doteq 0, \text{ for all } k,l \in \{1,\dots,M\} \text{ with } k \neq l.
\end{equation*}
Therefore, the characteristic functions in the definition
 of $D_1$ can be dropped, by paying an error that is exponentially small in $R$.
In other words
\begin{equation}
D_1 \doteq_{M^4} \sum_{i<j} \sum_{k<l} \langle  \Rbkl I_{B_k B_l}
\Psi_b, (H_b-\Einfty) \Rbij I_{B_i B_j} \Psi_b \rangle,
\end{equation}
where the factor $M^4$ in the exponentially small error arises due
to the double sum over the pairs of the atoms. Therefore, using
\eqref{Hbsimplification} and that $(H_{B_i B_j}-E_i-E_j)\Rbij I_{B_i
B_j} \Psi_b \doteq I_{B_i B_j} \Psi_b$, we obtain that
\begin{equation}\label{D1est}
D_1 \doteq_{M^5} \sum_{i<j} \sum_{k<l} \langle  \Rbkl I_{B_k B_l}
\Psi_b, I_{B_i B_j} \Psi_b \rangle$$$$ = \sum_{k<l} \langle I_{B_k
B_l} \psi_{B_k} \psi_{B_l},  \Rbkl I_{B_k B_l} \psi_{B_k} \psi_{B_l}
\rangle,
\end{equation}
where in the last step we used that $\sum_{i<j} I_{B_i B_j}=I_b$ and
\eqref{IbIBkBl}. To estimate $D_2$ we write $I_b=\sum_{m<n} I_{B_m
B_n}$. Then, due to \eqref{onesymmetryenough}, the contribution of
the term $I_{B_m B_n}$ in $D_2$ is zero unless $m,n \in
C_{i,j,k,l}:=\{i,j\} \cup \{k,l\}$. Therefore,
\begin{equation}\label{D2dec}
D_2=\sum_{i<j} \sum_{k<l} \sum_{m,n \in C_{i,j,k,l}} D_{ij,kl,mn}
\end{equation}
where
\begin{equation}\label{Dijklmn}
D_{ij,kl,mn} := \langle  \Rbkl I_{B_k B_l} \Psi_b, \chi_{B_k B_l} I_{B_m B_n} \chi_{B_i B_j}
\Rbij I_{B_i B_j} \Psi_b \rangle.
\end{equation}
We shall now show the following lemma:
\begin{lemma}\label{lem:Dijklmnfast0}
Let $i,j,k,l \in \{1,\dots,M\}$, with $i<j$, $k<l$ and $m,n \in C_{i,j,k,l}$.
Assume that one of the following happens:

(i) $\{i,j\} \cap \{k,l\}=\emptyset$.

(ii) $\{i,j\} \cap \{k,l\}$ has one element which belongs in $\{m,n\}$.
\newline
Then
\begin{equation}
D_{ij,kl,mn} \doteq 0.
\end{equation}
\end{lemma}
\begin{proof}
The argument is similar in both cases and we will illustrate it by
looking into the second case: we will show that
\begin{equation}\label{Dijklmnfast0}
D_{ij,il,il} \doteq 0.
\end{equation}
Indeed, we first observe that from \eqref{onesymmetryenough}
it follows that
$I_{B_i B_j} \Psi_b=P_{\psi_{B_j}}^\bot I_{B_i B_j} \Psi_b$.
Furthermore, by \eqref{almosteigat} it follows that $[P_{\psi_{B_j}}^\bot, H_{B_j}] \doteq 0$,
which implies that $[P_{\psi_{B_j}}^\bot, \Rbij] \doteq 0$, because $P_{\psi_{B_j}}^\bot$
commutes with all other operators of the definition of $\Rbij$.
Moreover, $[P_{\psi_{B_j}}^\bot, \chi_{B_i B_j}]=0$ because of \eqref{chipsi}, and
$P_{\psi_{B_j}}^\bot$ commutes with all other operators in the definition of $D_{ij,il,il}$,
 because they do not act on the
coordinates in $B_j$. Therefore, since  $P_{\psi_{B_j}}^\bot \Psi_b=0$, we obtain \eqref{Dijklmnfast0}.
This concludes the proof of Lemma \ref{lem:Dijklmnfast0}.
\end{proof}
From \eqref{D2dec} and Lemma \ref{lem:Dijklmnfast0}, we obtain that
there exists $c$ so that
\begin{equation}\label{D2est1}
|D_2| \ls M^3 \max_{ij,kl,mn} |D_{ij,kl,mn}|+M^4e^{-cR},
\end{equation}
because most of the $D_{ij,kl,mn}$ terms are exponentially small.
 From \eqref{changeIAiAj} and
\eqref{preIapsiaL2est} it follows that $\|I_{B_k B_l} \Psi_b\| \ls
\frac{1}{R^3}$. Proceeding as in estimating the right hand side of
\eqref{Ukl5est1} we obtain that $\|I_{B_m B_n} \chi_{B_i B_j} \Rbij
I_{B_i B_j} \Psi_b\| \ls \frac{1}{R^6}$. Inserting these estimates
in \eqref{Dijklmn} and using \eqref{D2est1}, we obtain that
\begin{equation}\label{D2est}
|D_2| \ls \frac{ M^3}{R^9}+M^4 e^{-cR}.
\end{equation}

From \eqref{proddec}, \eqref{IbIBkBl}, \eqref{D1est} and
\eqref{D2est} it follows that there exists $c$ so that
\begin{equation}\label{upperscalar}
\langle  \tilde \Psi_b , (H-\Einfty) \tilde
\Psi_b\rangle=-\sum_{k<l} \langle I_{B_k B_l} \psi_{B_k} \psi_{B_l}
, \Rbkl I_{B_k B_l} \psi_{B_k} \psi_{B_l}\rangle+
O\left(\frac{M^3}{R^9}\right)+O(M^5 e^{-cR}).
\end{equation}
By \eqref{testdef} we obtain that
\begin{equation*}
\|\tilde \Psi_b-\Psi_b\|^2= \sum_{i<j} \sum_{k<l} \langle \chi_{B_i B_j} \Rbij I_{B_i B_j}
\Psi_b,\chi_{B_k B_l} \Rbkl I_{B_k B_l}
\Psi_b  \rangle
$$$$\doteq_{M^4} \sum_{k<l} \langle \chi_{B_k B_l} \Rbkl I_{B_k B_l}
\Psi_b,\chi_{B_k B_l} \Rbkl I_{B_k B_l}
\Psi_b  \rangle,
\end{equation*}
where the last step can be verified with the same argument as in the proof of Lemma
\ref{lem:Dijklmnfast0}. The last estimate together with $\|I_{B_k B_l} \Psi_b\| \ls \frac{1}{R^3}$
gives that there exists $c$ so that
\begin{equation}\label{est:psibtilde}
\|\tilde \Psi_b-\Psi_b\|^2 \ls \frac{M^2}{R^6}+M^4 e^{-cR}.
\end{equation}
 In addition, from \eqref{chipsi}, $[P_{\Psi_b}^\bot, \Rbkl]=0$ and the equality $I_{B_k B_l} \Psi_b=P_{\Psi_b}^\bot I_{B_k B_l} \Psi_b$,
  which can be proven similarly to \eqref{chiaest2}, it follows that $\Psi_b$ is orthogonal to $\tilde \Psi_b- \Psi_b$, so that, by \eqref{est:psibtilde},
\begin{equation}\label{tildepsinorm}
1 \leq \|\tilde \Psi_b\|^2 \leq 1+ O\bigg(\frac{M^2}{R^6}\bigg)+O(M^4 e^{-cR}).
\end{equation}
 From \eqref{Wupp}, \eqref{upperscalar}, \eqref{Uaagenauer} and
\eqref{tildepsinorm}, we obtain Theorem \ref{thm:upper}. Note that
 the higher order terms could be dropped by imposing that $R \geq C_4
M^{\frac{1}{3}}$ for an appropriate $C_4>0$ depending only on $Z$.

 \section{ Proof of Proposition \ref{prop:Eimpliesgap} }\label{Hbotbndseveral}
 To simplify the exposition in the rest of the proof, we set the
elementary charge $e$ to be 1. This does not affect the proof, as
$e$ could have been removed from the beginning with a rescaling argument.
We will first prove the first part Proposition \ref{prop:Eimpliesgap} and then
we will discuss how to modify the proof in order to prove its second part.
\subsection{Proof of part (i) of Proposition \ref{prop:Eimpliesgap}}
 Recall that $\mathcal{A}$
is the set of all decompositions of $\{1,\dots,N\}$ into $M$ clusters. We will cover,
for $R \geq 1$, the configuration space $\mathbb{R}^{3 N}$ by the
domains,
\begin{equation}\label{Omegadef}
\Omega_{a}^{ \beta }=\{(x_1,\dots,x_N):|x_i-y_j| \leq \beta
R^{\frac{3}{4}}, \forall j \in \{1,\dots,M\} , \forall i \in A_j \},
\end{equation}
with $\beta>0$, $a=\{A_1,\dots,A_M\} \in \mathcal{A}$, where the
electrons in each cluster are close to the corresponding nucleus,
and
\begin{equation}
\Omega_{\{i\}}^{ \beta}=\{(x_1,\dots,x_N):|x_i-y_j| \geq \beta
R^{\frac{3}{4}}, \forall j \in \{1,\dots,M\} \},
\end{equation}
for $i \in \{1,\dots,N\}$, where the $i-$th electron is far away from all the
nuclei. It will be made clear in the proof that the choice of the
power $\frac{3}{4}$ in the definitions above optimizes the
assumption on $R$ in the proposition. Let
$\hat{\mathcal{A}}=\mathcal{A} \cup \{\{1\},\dots,\{N\}\}$. With the
covering above we associate a partition of unity $(J_{\hb})_{\hb \in
\hat{\mathcal{A}}}$ having the properties (cf. \cite{Sig}):
\begin{equation}\label{partunprop1}
 \supp J_{a} \subset \Omega_{a}^{\frac{1}{6}}, \quad \supp J_{\{i\}} \subset
\Omega_{\{i\}}^{ \frac{1}{12}}, \quad \forall a \in \mathcal{A},
\quad \forall i \in \{1,\dots,N\},
\end{equation}
\begin{equation}\label{partunprop3}
0 \leq J_{\hb} \leq 1, \quad \forall \hb \in \hat{\mathcal{A}},
\end{equation}
\begin{equation}\label{partunprop2}
\sum_{\hb \in \hat{\mathcal{A}}}J_{\hb}^2=1,
\end{equation}
\begin{equation}\label{symmetryJa}
T_\pi J_a=J_a, \forall \pi \in S(a), \forall a \in \mathcal{A},
\end{equation}
where $S(a) \subset S_N$ is the subgroup of permutations that
keeps the clusters of $a$ invariant, and $T_\pi$ was defined in \eqref{def:Tpi}.
 We denote the characteristic
function of a set $K$ by $\chi_K$. We consider the functions
$F_{a}=g*\chi_{\Omega_{a}^{\frac{7}{48}}}, a \in \mathcal{A}$ and
$F_{\{j\}}=g*\chi_{\Omega_{\{j\}}^{ \frac{5}{48}}}, j \in
\{1,\dots,N\}$, where $g:=\otimes_{j=1}^N g_R$,
$g_R(x):=R^{-\frac{9}{4} } g_1(R^{-\frac{3}{4}}x)$ and $g_1:\R^{3}
\rightarrow \R$ is a $C_c^{\infty}$ spherically symmetric function
supported in the ball $B(0,{\frac{1}{48}})$ with $g_1 \geq 0$ and
$\int_{\R^{3}} g_1=1$. Then $F_{\hb} \in C^{\infty}$ and $F_{\hb}
\geq 0$, for all $\hb \in \hat{\cA}$. Furthermore, using the
triangle inequality and the fact that $g$ is supported in
$B(0,{\frac{R^\frac{3}{4}}{48}})$ we obtain that
$F_{\hb}|_{\Omega_{\hb}^{\frac{1}{8}}}=1$ for all $\hb \in \hat{A}$.
 The last equality together with the fact that $\cup_{\hb \in \hat{\mathcal{A}}}
\Omega_{\hb}^{\frac{1}{8}}=\R^{3N}$ gives that $\sum_{\hb \in
\hat{\mathcal{A}}} F_{\hb}^2 \geq 1$. We now define
\begin{equation*}
J_{\hc}=\frac{F_{\hc}}{\sqrt{\sum_{\hb \in \hat{\mathcal{A}}}
F_{\hb}^2}},\ \quad \hc \in \hat{\mathcal{A}}.
\end{equation*}
 All the stated properties of the family $J_{\hc}$ follow easily by construction.
We will now show that there exists $D_1>0$ so that
\begin{equation}\label{nablaleqN}
\sum_{\hc \in \hat{\mathcal{A}}} |\nabla J_{\hc}|^2 \leq \frac{D_1^2
N^2}{R^{\frac{3}{2}}},
\end{equation}
  where by $D_n$ with $n \in \N$, we mean a positive constant
which depends only on $Z$ but, unlike $c, C$, does not change
from one equation to the other.
 Indeed, a direct calculation gives,
using the inequality $\sum_{\hb \in \hat{\mathcal{A}}} F_{\hb}^2
\geq 1$, that
\begin{equation}\label{angular}
\sum_{\hc \in \hat{\mathcal{A}}} |\nabla J_{\hc}|^2 \leq \sum_{\hc
\in \hat{\mathcal{A}}} |\nabla F_{\hc}|^2.
\end{equation}
Moreover, by construction of $F_{\hc}$, there exists
$D_1$ such that
\begin{equation}\label{Fcschranke}
\|\nabla F_{\hc}\|_{L^\infty} \leq \frac{ D_1 \sqrt{N}}{2
R^\frac{3}{4}}, \text{ } \forall \hc \in \hat{\mathcal{A}},
\end{equation}
where the factor $\sqrt{N}$ arises because of the change of the
dimension of the domain of definition $\R^{3N}$, and the factor
$\frac{1}{R^{\frac{3}{4}}}$ arises from the rescaling of the
function $g_R$ involved in the construction. Here $D_1=2 \|\nabla
g_1\|_{L^1}$. Furthermore, observe that $\text{supp}F_{a_1} \cap \text{supp}F_{a_2}=0, \forall
a_1, a_2 \in \mathcal{A}$ with $a_1 \neq a_2$. As a consequence, the
sum $\sum_{\hc \in \hat{\mathcal{A}}} |\nabla F_{\hc}|^2$ consists
locally of at most $N+1$ terms, because $\hat{\cA}$ has $N$ more
elements than $\cA$. The last observation together with
\eqref{angular} and \eqref{Fcschranke} gives \eqref{nablaleqN}.

    Now we use the IMS localization formula (see for example \cite{CFKS})
\begin{equation}\label{IMS}
\H=\sum_{\hb \in \hat{\mathcal{A}}} \left( J_{\hb} \H
J_{\hb}-|\nabla J_{\hb}|^2 \right).
\end{equation}
  From \eqref{IMS} and \eqref{nablaleqN} it follows that
\begin{equation*}
Q_N P^\bot \H P^\bot Q_N \geq \sum_{\hb \in \hat{\mathcal{A}}} Q_N P^\bot
J_{\hb} H J_{\hb} P^\bot Q_N-\frac{ D_1^2 N^2}{R^{\frac{3}{2}}} Q_N
P^\bot.
\end{equation*}
Therefore, using that $Q_N P^\bot \leq 1$ and the equality
$\H^{\s,\bot}=Q_N P^\bot \H P^\bot Q_N$, we obtain that
\begin{equation}\label{Hbotgenest1}
\H^{\s,\bot} \geq \sum_{\hb \in \hat{\mathcal{A}}} Q_N P^\bot J_{\hb}
H J_{\hb} P^\bot Q_N-\frac{ D_1^2 N^2}{R^{\frac{3}{2}}}.
\end{equation}
We will now estimate the different terms $Q_N P^\bot J_{\hb} H J_{\hb}
P^\bot Q_N$. To this end we need some notation and some definitions.
 Recall that $E_j$  denotes the ground state energy $H_j^\s$ (defined in \eqref{def:His}),
 and let $ E_j^{'}$ denote its  first excited state energy.
   We define
\begin{equation}\label{L1def}
D_2=\min\{E_{j}^{'}-E_{j}:j =1,\dots,M\},
\end{equation}
where it is obvious that $D_2$ depends only on $Z$, and
\begin{equation}\label{L3def'}
\gamma= \Sigma_{N-1}-E(y),
\end{equation}
where $\Sigma_{N-1}$ was defined in \eqref{def:sigmam}.
From Theorems \ref{hvz} and \ref{zysl}  it follows that
$\gamma>0$. Arguing similarly as in the proof of Lemma \ref{lem:Eyless0}, we can show that there exists $C$ so that
\begin{equation}\label{GSEless0}
 \inf \sigma(H^{N-1,\s}(y))<0, \text{ for all } R \geq C.
\end{equation}
The only difference in the argument is in the construction of the test function: it is tensor product of $M-1$ cut off ground states of atoms and 1 cut off ground state of a positive ion instead of $M$ cut off ground states of atoms. The interaction vanishes again, because only one ion is positive and the potential created by the atoms is zero (see \eqref{onesymmetryenough}). Therefore, if $R \geq C$, then it follows
from \eqref{def:sigmam} and \eqref{GSEless0} that
\begin{equation}
\Sigma_{N-1}=  \inf \sigma(H^{N-1,\s}(y))<0,
\end{equation}
because removing the restriction onto $\Ran Q_{N-1}$ can change
the spectrum only at zero. From the last estimate and \eqref{L3def'}
it follows that
 \begin{equation}\label{firstexcited}
E(y)+\gamma<0
 \end{equation}
 and that
 \begin{equation}\label{L3def}
  \gamma= \inf \sigma(H^{N-1,\s}(y))-E(y).
 \end{equation}
  Before stating the next lemma we remind and introduce some more notation.
Let $H_{\{j\}}=H-I_{\{j\}}$,  where $I_{\{j\}}=-\sum_{i=1}^M
\frac{ Z_i}{|x_j-y_i|}+\sum_{i\neq j}\frac{1}{|x_i-x_j|}$ is
 the interaction between the $j-$th electron and the rest of the system.
 Hence, $H_{\{j\}}$ is the Hamiltonian of the system with the $j$-th electron decoupled from the rest of the system.
Let $Q_{\{j\}}:=Q_{\{1,2,...,N\}/\{j\}}$ with the latter defined as in \eqref{def:QAj}.
In other words $Q_{\{j\}}$ is the projection onto the functions
 that are antisymmetric with respect to all electron coordinates except the $j$-th.
Recall that for $b \in \cA^{at}$ the function $\chi_b$, defined in \eqref{chib},
is a smooth characteristic function of a set where
the electrons in $B_j$ are near to the nucleus $y_j$.
Recall also that by
$D_n, n \in \N$ we denote positive constants, which depend  on $Z$
 only and do not change from one equation to the other. Our next goal is to prove
\begin{lemma}\label{lem:PJHJPest}
With the notation defined above, we have that
\begin{equation}\label{PJHJPestHVZ}
Q_N P^\bot J_{\{j\}} H J_{\{j\}} P^\bot Q_N \geq \bigg(E(y)+\gamma-\frac{12
N}{R^{\frac{3}{4}}}\bigg) Q_N P^\bot J_{\{j\}}^2
 P^\bot Q_N, \quad \forall j \in
\{1,\dots,N\},
\end{equation}
and, with $D_2$ defined in \eqref{L1def}, there exist  $c, D_3$ such
that
\begin{align}\notag
Q_N P^\bot J_a H J_a P^\bot Q_N \geq \bigg(\Einfty+D_2 &- \frac{
D_3 N^2}{R^{\frac{3}{2}}}\bigg) Q_N P^\bot J_a^2 P^\bot Q_N  \\
\label{PJHJPestatomic} &-Q_N \chi_a O( M e^{-cR^{\frac{3}{4}}}) \chi_a
Q_N, \quad \forall a \in \cA^{at}.
\end{align}
Moreover, Property (E) implies that there exists  $D_4$, such that
\begin{equation}\label{PJHJPestnonat}
Q_N P^\bot J_{b} H J_{b} P^\bot Q_N \geq \bigg(\Einfty+ D_4- \frac{ D_3
N^2}{R^{\frac{3}{2}}}\bigg) Q_N P^\bot J_{b}^2 P^\bot Q_N, \quad \forall R \geq \frac{4N}{D_4},\quad  \forall
b \in \cA /\cA^{at}.
\end{equation}
\end{lemma}
\begin{proof}
 \textbf{Proof of \eqref{PJHJPestHVZ}:} We decompose the left hand side into two terms using that
 $H=H_{\{j\}}+I_{\{j\}}$.
 From \eqref{L3def} it follows that $ Q_{\{j\}} H_{\{j\}} \geq E(y)+\gamma$. Moreover,
 since by \eqref{partunprop1}
 the j-th electron coordinate is at least
$\frac{R^{\frac{3}{4}}}{12}$ far from all nuclei on $\supp
J_{\{j\}}$, we obtain that $I_{\{j\}} \geq
-\frac{12N}{R^{\frac{3}{4}}}$, on $\supp J_{\{j\}}$. The last two
estimates together with the fact that $Q_{\{j\}} Q_N= Q_N$ and that
$Q_{\{j\}}$ commutes with $P^\bot, J_{\{j\}},$ imply
\eqref{PJHJPestHVZ}.
\newline
\textbf{Proof of \eqref{PJHJPestatomic}:}  We decompose the left
hand side into two terms using that $H=H_a+I_a$.
 We will first estimate the term $Q_N P^\bot J_a H_a J_a P^\bot Q_N$.  By
\eqref{Omegadef}, \eqref{partunprop1} and the support properties of
$\Psi_a$ (see \eqref{Psiaeq}, \eqref{psiAj}) we obtain
that
\begin{equation}\label{Janeqb}
 J_c \Psi_a = \delta_{ac} J_c \Psi_a, \quad \forall c \in \mathcal{A}^{at}.
\end{equation}
Due to \eqref{P} and \eqref{Janeqb} we have
\begin{equation}\label{IMS3'}
Q_N P^\bot J_a H_a J_a P^\bot Q_N=Q_N P_{\Psi_a}^\bot J_a H_a J_a
P_{\Psi_a}^\bot Q_N.
\end{equation}
We recall that $H_a^\s, Q_a$, were defined in \eqref{def:Has}. Since $Q_a$
commutes with $J_a, H_a, P_{\Psi_a},$ (with $J_a$ because of
\eqref{symmetryJa}), equation \eqref{IMS3'} implies that
\begin{equation}\label{IMS3}
Q_N P^\bot J_a H_a J_a P^\bot Q_N=Q_N P_{\Psi_a}^\bot J_a H_a^\s J_a
P_{\Psi_a}^\bot Q_N.
\end{equation}
Using \eqref{L1def} and the decomposition $1= P_{\Phi_a}+P_{\Phi_a}^\bot$, where recall that $\Phi_a$ is the (exact) ground state of $H_a^\s$ (see \eqref{Phiaeq}), we obtain that
\begin{equation}\label{est:Hasbelow}
H_a^\s \geq (\Einfty+D_2)-D_2 P_{\Phi_a}.
\end{equation}
Note that we have dropped the operator $Q_a$ on the
right hand side of \eqref{est:Hasbelow} using that $\Einfty+D_2<0$
and that $Q_a \leq 1$. From \eqref{est:Hasbelow} it follows that
\begin{equation}\label{IMS2}
P_{\Psi_a}^\bot J_a H_a^\s J_a P_{\Psi_a}^\bot \geq (\Einfty+D_2)
P_{\Psi_a}^\bot J_a^2 P_{\Psi_a}^\bot-D_2 P_{\Psi_a}^\bot J_a
P_{\Phi_a} J_a P_{\Psi_a}^\bot.
\end{equation}
To estimate the second term on the right hand side of \eqref{IMS2},
we use that
\begin{equation}\label{komutatorsum}
[J_a, P_{\Psi_a}]=\sum_{j=1}^M P_{\psi_{A_1}}\dots P_{\psi_{A_{j-1}}}
[J_a, P_{\psi_{A_j}}] P_{\psi_{A_{j+1}}}\dots P_{\psi_{A_M}}.
\end{equation}
Due to the exponential decay of $\psi_{A_j}$ one expects, since
$J_a=1$ when the electrons are close to the nucleus, that the right
hand side of \eqref{komutatorsum} is small. Indeed,
 by \eqref{Omegadef}-\eqref{partunprop1} and
\eqref{partunprop2}, we have that
$J_a|_{\Omega_a^{\frac{1}{12}}}=1$. This together with
\eqref{Psiadecay}, \eqref{Omegadef} and \eqref{komutatorsum} implies
that there exists $c$ so that
\begin{equation*}
\|[J_a, P_{\Psi_a}]\|  \ls M e^{-c R^{\frac{3}{4}}}.
\end{equation*}
The last estimate together with the fact that $P_{\Psi_a}^\bot
P_{\Phi_a} \doteq_M 0$, which follows from \eqref{phiapsia}, gives
\begin{equation}\label{kom1}
\|P_{\Psi_a}^\bot J_a P_{\Phi_a} J_a P_{\Psi_a}^\bot\|  \ls M e^{-c
R^{\frac{3}{4}}}.
\end{equation}
From  \eqref{Psiaeq}, \eqref{chib} and \eqref{eqn:chirbrestriction} it follows that
\begin{equation}\label{eqn:chibrestriction}
\chi_a|_{\supp \Psi_a}=1.
\end{equation}
Therefore, both sides of
\eqref{IMS2} are invariant when we multiply on the left and right by $\chi_a$.
From this observation and estimate \eqref{kom1} we obtain that
\begin{equation*}
P_{\Psi_a}^\bot J_a H_a^\s J_a P_{\Psi_a}^\bot \geq (\Einfty+ D_2)
P_{\Psi_a}^\bot  J_a^2 P_{\Psi_a}^\bot - \chi_a O(M
e^{-cR^{\frac{3}{4}}}) \chi_a.
\end{equation*}
The last inequality together with \eqref{P} and \eqref{Janeqb} implies
that
\begin{equation}\label{IMS4}
P_{\Psi_a}^\bot J_a H_a^\s J_a P_{\Psi_a}^\bot \geq (\Einfty+ D_2)
P^\bot  J_a^2 P^\bot - \chi_a O(M e^{-cR^{\frac{3}{4}}}) \chi_a.
\end{equation}

 We will now estimate the term  $Q_N P^\bot J_a I_a J_a P^\bot Q_N $.
  Recall that $\tilde{I}_{ij}^{kl}$ is obtained by $I_{ij}^{kl}$ by the change of
variables \eqref{zlm} (see \eqref{bijkl}, \eqref{dijkl}).
We define $\tilde{J}_a$ in a similar way. If $a \in
\mathcal{A}^{at}$, then using \eqref{I1sevmu},
  which holds on the support of $\tilde{J}_a$, and that the electrons are at most
 $\frac{R^{\frac{3}{4}}}{6}$
far away from the corresponding nuclei (see \eqref{partunprop1}), or in other words $|z_{ik}|
\leq \frac{R^{\frac{3}{4}}}{6}$ on the support of $\tilde{J}_a$, we obtain
that there exists $D_3$ such that
\begin{equation}\label{Ibjbzer}
\| \tilde{I}_{ij}^{kl}|_{\supp \tilde{J}_a}\|_{L^\infty} \leq
\frac{D_3 R^{\frac{3}{2}}}{R^3}.
\end{equation}
 Using \eqref{Iabreakup} we obtain that
 $I_a=\sum_{i<j}^{1<M} \sum_{k \in A_i, l \in A_j} I_{ij}^{kl}$,
 which together with \eqref{Ibjbzer} implies that
\begin{equation}\label{Iaschranke}
 I_a|_{\supp J_a}  \ge - \frac{D_3 N^2}{R^{\frac{3}{2}}}, \forall a \in \mathcal{A}^{at}.
\end{equation}
From \eqref{IMS3}, \eqref{IMS4} and \eqref{Iaschranke} we obtain
\eqref{PJHJPestatomic}.
\newline
 \textbf{Proof of \eqref{PJHJPestnonat}:}
If a decomposition is not in $\mathcal{A}^{at}$,
then the intercluster interaction has more attractive terms
and it can not be bounded in the same way as in the case of an atomic decomposition.
With the Property (E) we will counterbalance this issue. We first
show that each decomposition $b \in \mathcal{A}$
is associated to an $a \in \mathcal{A}^{at}$ in the following way:
there is a finite sequence $c_0,\dots, c_l$ of decompositions so that
$c_0=a, c_l=b$ and $c_{m+1}$ is created by $c_m$ by moving an electron
from a non-negative ion of $c_m$ to a non-positive ion of $c_m$.
Going from $c_m$ to $c_{m+1}$ creates attractive terms in the intercluster
interaction. But Property (E) gives a gap which counterbalances the attraction.
We now state everything precisely. For a decomposition with an index
e.g. $c_n \in \mathcal{A}$, we write
$c_n=\{C_{n,1},...,C_{n,M}\}$. We  need the following lemma
\begin{lemma}\label{lem:bca}
Suppose that $b \in \cA/\cA^{at}$. Then there exists $a \in
\cA^{at}, l \in \mathbb{N}$ and a finite sequence $c_0,\dots,c_l$ of elements in $\cA$
such that:
\newline
(i) $c_0=a$ and $c_l=b$.
\newline
(ii) For each $m=0,...,l-1$, there exists $i,j \in \{1,\dots,M\}$ with $i \neq j$
and $k \in \{1,\dots,N\}$
so that $|C_{m,i}| \leq Z_i$, $|C_{m,j}| \geq Z_j$, $k \in C_{m,i}$,
and so that $C_{m+1,n}=C_{m,n}$ for $n \neq i,j$,  $C_{m+1,i}=C_{m,i}/\{k\}$, $C_{m+1,j}=C_{m,j}\cup\{k\}$.
\end{lemma}
 The lemma can be proven by a simple induction on the number of
atoms.  As we said, by Property (E) the sum of ground states of ions
of $c_{m+1}$, is bigger that the one of $c_m$. In principle the gap could depend
on $m$ but we shall show that it depends only on $Z$. Recall that $H_b^\s$
 was defined in \eqref{def:Has} for any $b \in \mathcal{A}$.
\begin{lemma}\label{lem:gapde}
Let $c_0,\dots,c_l$ be as in Lemma \ref{lem:bca}.
 Property (E) implies that there exists $D_4$ so that
\begin{equation}\label{Hcgap}
 \inf \s(H_{c_{m+1}}^\s) \geq \inf \s(H_{c_m}^\s)+2 D_4, \quad
 \forall m=0,\dots,l-1.
\end{equation}
 Moreover, we have that
\begin{equation}\label{Icgap}
 \inf I_{c_{m+1}}|_{\supp J_{c_{m+1}}} \geq \inf I_{c_{m}}|_{
\supp J_{c_{m}}}-\frac{4N}{R}, \quad \forall m=0,\dots,l-1.
\end{equation}
\end{lemma}
\begin{proof}
In the proof we shall use the notation of Lemma \ref{lem:bca}.
From \eqref{Ha}, \eqref{def:HAis}, \eqref{def:Has} and \eqref{infsHAi} and Lemma
\ref{lem:bca} it follows that
\begin{equation}\label{eqn:stepgap}
\inf \s(H_{c_{m+1}}^\s)-\inf \s(H_{c_m}^\s)=E_{i,Z_i+1-|C_{m,i}|}+ E_{j,Z_j-1-|C_{m,j}|}-E_{i,Z_i-|C_{m,i}|}-E_{j,Z_j-|C_{m,j}|}>0,
\end{equation}
where the last inequality follows from Property (E). Therefore,
Property (E) implies \eqref{Hcgap} for some positive constant, which
a priori could depend on $m$. We show that in fact the constant $D_4$
depends only on $Z$. Indeed, by \cite{Lieb}, (see
also \cite{Sig}, \cite{Phan}) we know that the sequence $E_{j,n},$
$n \in \Z, n \leq Z_j$ is constant when $n \leq -Z_j$. Therefore, the set
$B_{Z_j}:=\{E_{j,n}| n \in \Z, n \leq Z_j\}$ consists of at most
$2Z_j+1$ elements. Moreover, all the gaps in Property (E) are
determined by differences of elements in the sets $B_{Z_j}$. Since
these sets are at most $Z$ many, we can define $2 D_4$ to be
the minimum of these gaps, and it depends only on $Z$.

      To prove \eqref{Icgap} we observe, by Lemma \ref{lem:bca}, that the inter-cluster
interactions $I_{c_{m+1}}$, $I_{c_{m}}$ differ only by the interaction terms of the
electron with coordinate $x_k$. Because of the last observation, it is convenient to denote by
$I_{c_n,k}$, where $n=m,m+1$, the part of $I_{c_n}$ which has only
the inter-cluster interaction terms of the electron with coordinate
$x_k$.
     Since by construction of $J_a$ we have $\supp J_a=\supp F_a, \forall a \in \cA$
and $F_a$ is a product function, it turns out that the supports of $J_{c_m}$ and
$J_{c_{m+1}}$ are product sets differing only on the $k$-th element
of the product, which corresponds to the coordinate $x_k$.
Therefore, the difference of the two infimums in \eqref{Icgap}
depends only on the interaction terms $I_{c_n,k}$, $n=m,m+1$. Since
on the support of $J_{c_{m+1}}$  the coordinate $x_k$ is at least $\frac{R}{2}$
far from the nuclei of the other clusters, and the total charge of
these  nuclei is less than $N$, we have that $I_{c_{m+1},k} \geq
-\frac{2N}{R}$, on $\supp J_{c_{m+1}}$. Similarly, since on the
support of $J_{c_m}$ the electon $x_k$ is at least $\frac{R}{2}$ far
from the electrons in the other clusters, we obtain that
$I_{c_{m},k} \leq \frac{2N}{R}$, on $\supp J_{c_{m}}$. Therefore,
\eqref{Icgap} follows.
\end{proof}
 We now continue with the proof of \eqref{PJHJPestnonat}. Let $a,b$ be as
 in Lemma \ref{lem:bca}. Then,
   from \eqref{Hcgap} it follows that
   \begin{equation}\label{est:Habgap}
 \inf \s(H_b^\s)- \inf \s(H_a^\s) \geq 2 l D_4.
    \end{equation}
  From the last estimate and \eqref{Icgap} it follows that
\begin{equation}
 \inf \s(H_b^\s)+ \inf I_b|_{\supp J_b} \geq \inf \s(H_a^\s)+ \inf I_a|_{\supp J_a}+l
(2D_4-\frac{4N}{R}).
\end{equation}
Using the last estimate together with the inequalities $\inf
\s(H_a^\s) \geq \Einfty$ and \eqref{Iaschranke}, we obtain that
\begin{equation}\label{Eprimegap}
J_b H_b^\s J_b+ J_b I_b J_b \geq \bigg(\Einfty + l
\big(2D_4-\frac{4N}{R}\big)-\frac{D_3 N^2}{R^{\frac{3}{2}}}\bigg)
J_b^2.
\end{equation}
This implies \eqref{PJHJPestnonat}, when $R \geq \frac{4N}{D_4}$.
This concludes the proof of Lemma \ref{lem:PJHJPest}.
\end{proof}
 We now continue with the proof of Proposition \ref{prop:Eimpliesgap}.
Observe that we can use Estimate \eqref{est:Ey}, because for its proof we did not need \eqref{Hbotbnd}.
The estimates \eqref{est:Ey}, \eqref{partunprop2}, \eqref{PJHJPestHVZ}, \eqref{PJHJPestatomic} and \eqref{PJHJPestnonat} imply that
 for there exists $C,c$ so that for all $R \geq \frac{4N}{D_4}$
 we have that
$$\sum_{\hb \in \hat{\mathcal{A}}} Q_N P^\bot J_{\hb} H J_{\hb}
P^\bot Q_N $$$$\geq \bigg(E(y)+ \min\{\gamma, D_2, D_4\} -\frac{12
N}{R^{\frac{3}{4}}}-\frac{D_3 N^2}{R^{\frac{3}{2}}}-C e^{-cR} \bigg) Q_N P^\bot
Q_N -\sum_{a \in \cA^{at}} Q_N \chi_{a} O(M e^{-cR^{\frac{3}{4}}})
\chi_{a}Q_N.$$
 Since $\chi_a \chi_b=0$ when $a \neq b$, arguing as in
  the proof of the inequality \eqref{equivimp}, we can show that
  the last inequality implies that
  \begin{equation}\label{sumofpart}
 \sum_{\hb \in \hat{\mathcal{A}}} Q_N P^\bot J_{\hb} H
J_{\hb} P^\bot Q_N $$$$ \geq E(y)+ \min\{\gamma,D_2, D_4 \} -\frac{12
N}{R^{\frac{3}{4}}} -\frac{D_3 N^2}{R^{\frac{3}{2}}}- O(M
e^{-cR^{\frac{3}{4}}}), \quad \forall R \geq \frac{4N}{D_4},
\end{equation}
 where we dropped $Q_N P^\bot Q_N$ on the right hand side using \eqref{firstexcited} and that $Q_N P^\bot Q_N \leq 1$.
  The inequalities \eqref{Hbotgenest1}, \eqref{sumofpart} and the Lemma
\ref{lem:gammac} below, imply that there exists $C_1$, depending only on $Z$ so that, for $R \geq C_1
N^\frac{4}{3}$, Estimate \eqref{Hbotbnd} holds. This concludes the proof of Proposition \ref{prop:Eimpliesgap}.
It therefore remains to prove:
\begin{lemma}\label{lem:gammac}
There exist $c,C$ so that if $R \geq CN^{\frac{4}{3}}$, then we have
that $\gamma \geq c$.
\end{lemma}
\begin{proof}
By \eqref{L3def'}  showing the lemma is equivalent to showing, that
there exists $C,c$ so that if $R \geq CN^{\frac{4}{3}}$, then
 \begin{equation}\label{sigmanminus1n}
\Sigma_{N-1} \geq E(y)+c.
 \end{equation}
 Recall Definition \eqref{def:H^m}.
 Since estimating $H^\bot=P^\bot H^{N,\s}(y) P^\bot$ from below was reduced to estimating $H^{N-1,\s}(y)$ from below (namely the Proof of Proposition
 \ref{prop:Eimpliesgap} was reduced to proving Lemma \ref{lem:gammac}), it turns out
 that we need to estimate from below all $H^{N-k,\s}(y)$, where $k \in \{1,\dots,N-1\}$. To this end
 we construct, similarly as before, a partition $J_{a'}, a' \in
\hat{\cA'}$ where $\hat{\cA'}=\cA' \cup \{\{1\},\dots,\{N-k\}\}$, and
$\cA'$ is the set of decompositions of $\{1,\dots,N-k\}$ into $M$
clusters. The parameters for the construction of $J_{a'}$ are the
same (except for $N$ of course). The Hamiltonian $H_{a'}$ and the intercluster
interaction $I_{a'}$ can be defined in a similar manner as in the case of the decompositions
in $\hat{\cA}$ and, similarly as in \eqref{Hadecomp}, we have
\begin{equation}\label{Hadecomp'}
H^{N-k}(y)=H_{a'}+I_{a'}.
\end{equation}
 Using the IMS localization
formula, we can prove, similarly to \eqref{Hbotgenest1} (omitting
$P^\bot$), that
\begin{equation}\label{IMS'}
H^{N-k,\s}(y)  \geq \sum_{a' \in \hat{A'}} Q_{N-k} J_{a'}
H^{N-k}(y) J_{a'} Q_{N-k}-\frac{D_1^2 (N-k)^2}{R^{\frac{3}{2}}} Q_{N-k}.
\end{equation}
 If $a' \in \cA'$ (decomposition to clusters), then using
 \eqref{Hadecomp'} we can show similarly to \eqref{IMS3}
that
\begin{equation}\label{HNminus1dec}
Q_{N-k} J_{a'} H^{N-k}(y) J_{a'} Q_{N-k} =  Q_{N-k} J_{a'}
H_{a'}^\s J_{a'} Q_{N-k}+ Q_{N-k} J_{a'} I_{a'} J_{a'} Q_{N-k},
\end{equation}
where $H_{a'}^\s=H_{a'} Q_{a'}$, and $Q_{a'}$ is defined similarly as in \eqref{def:Has}.
 Observe now that $a'$ has some
clusters $A_{i_1}',\dots, A_{i_l}'$, for some $l \leq k$,  corresponding to
positive ions of total charge at least $k$, namely $|A_{i_m}'| < Z_{i_m}$ for all
$m \in \{1,\dots,l\}$ and $\sum_{m=1}^l (Z_{i_m}-|A_{i_m}'|) \geq k$. Therefore, $a'$ comes
from an $a \in \cA$ (decomposition of $\{1,\dots,N\}$), after removing
$k$ electrons from nonnegative ions of $a$. More precisely, there
exists an $a \in \cA$  with $A_{i_m}' \varsubsetneq A_{i_m}$ and
$|A_{i_m}| \leq Z_{i_m}$ for all $m \in \{1,\dots, l\}$, with $A_n=A_n'$ for all $n \neq i_1,\dots,i_l$ and with
\begin{equation}\label{est:chargek}
\sum_{m=1}^l(|A_{i_m}|-|A_{i_m}'|)=k.
\end{equation}
 By Theorem \ref{zysl} it follows that
 \begin{equation}\label{est:D5gap}
\inf \s (H_{A_{i_m}'}^\s) - \inf \s (H_{A_{i_m}}^\s) \geq 2(|A_{i_m}|-|A_{i_m}'|)D_5,\text{ } \forall  m \in \{1,\dots, l\},
\end{equation}
where
\begin{equation}\label{def:D5}
D_5=\frac{1}{2} \min\{E_{i,n+1}-E_{i,n}| i \in \{1,\dots,M\}, n \in \{0,\dots,Z_i-1\}\},
\end{equation}
depends only on $Z$.
It therefore follows, using \eqref{est:chargek} and \eqref{est:D5gap}, that
\begin{equation}\label{infsHa'}
\inf \s (H_{a'}^\s) \geq \inf \s(H_{a}^\s) + k 2 D_5.
\end{equation}
 Similarly to \eqref{Icgap} it can be proven that
\begin{equation}\label{infIa'}
\inf( I_{a'}|_{\supp J_{a'}}) \geq \inf(I_{a}|_{\supp J_{a}})-k
\frac{4N}{R}.
\end{equation}
Moreover, in the proof of \eqref{Eprimegap}, we have proven
that
\begin{equation}\label{infHaIa}
\inf \s (H_{a}^\s)+\inf( I_{a}|_{\supp J_{a}}) \geq
\Einfty-\frac{D_3 N^2}{R^{\frac{3}{2}}}, \forall a \in \cA, \forall R \geq \frac{2N}{D_4}.
\end{equation}
 Using \eqref{est:Ey} together with \eqref{HNminus1dec}, \eqref{infsHa'}, \eqref{infIa'}
and \eqref{infHaIa} we arrive at
\begin{equation}\label{dec'}
Q_{N-k} J_{a'} H^{N-k}(y) J_{a'} Q_{N-k}  $$$$\geq
\bigg(E(y)+k\big(2D_5-\frac{4N}{R}\big) -\frac{D_3
N^2}{R^{\frac{3}{2}}}-C M e^{-cR}\bigg) Q_{N-k} J_{a'}^2 Q_{N-k},
\forall a' \in \cA',  \forall R \geq \frac{2N}{D_4} .
\end{equation}
 If $a'=\{m\}, m \in \{1,\dots,N-k\}$, then, arguing similarly as in the proof
 of \eqref{PJHJPestHVZ}, we obtain that
 $I_{a'}|_{\supp J_{a'}} \geq- \frac{12(N-k)}{R^{\frac{3}{4}}}$.
Since moreover $Q_{a'} H_{a'} \geq \Sigma_{N-k-1} Q_{a'}$, where
$\Sigma_m$ was defined in \eqref{def:sigmam}, we obtain that
\begin{equation}\label{j'}
Q_{N-k} J_{a'} H^{N-k}(y) J_{a'} Q_{N-k} $$$$ \geq
\bigg(\Sigma_{N-k-1} -\frac{12 (N-k)}{R^{\frac{3}{4}}}\bigg) Q_{N-k}
J_{a'}^2 Q_{N-k}, \forall a' \in \{\{1\},\{2\},\dots,\{N-k\}\},
\end{equation}
because $Q_{a'}$ commutes with $J_{a'}$ and $Q_{a'} Q_{N-k}=Q_{N-k}$.
 Using \eqref{IMS'}, \eqref{dec'}, \eqref{j'} and
$\sum_{a' \in \hat{\cA'}} J_{a'}^2=1$, and taking $R$ large enough,
so that $(2 D_5-\frac{4N}{R}) -\frac{(D_1^2+D_3)
N^2}{R^{\frac{3}{2}}}-Ce^{-cR} \geq D_5$ and $R \geq \frac{2N}{D_4}$,
(here C,c are the same as in \eqref{dec'}) we arrive at
\begin{equation}\label{sigmarec}
\Sigma_{N-k} \geq \min \{E(y)+ k D_5, \Sigma_{N-k-1}-\frac{12
N}{R^{\frac{3}{4}}}-\frac{D_1^2 N^2}{R^{\frac{3}{2}}}\}, \forall
k \in \{1,\dots,N-1\}.
\end{equation}
By \eqref{est:Ey} and \eqref{Einftydef} it follows that
 there exists $C, D_6$ so that $E(y) \leq - N D_6$ for all $R \geq C$.
 Since, moreover, $\Sigma_0 \geq 0$, we obtain that
    $\Sigma_0 \geq E(y)+ N D_6$. Therefore, using
\eqref{sigmarec} for $k=N-m, m=1,\dots,N-1$, and taking $R$ large
enough so that the assumptions on $R$ of \eqref{sigmarec} are fulfilled
and  $\frac{12 N}{R^{\frac{3}{4}}}+\frac{D_1^2
N^2}{R^{\frac{3}{2}}} \leq \min \{D_5, D_6\}$, it follows, by
induction on $m$, that
\begin{equation}
\Sigma_m \geq E(y)+(N-m) \min \{D_5, D_6\}.
\end{equation}
Applying the last estimate for $m=N-1$ we arrive at
\eqref{sigmanminus1n}. This concludes the proof of Lemma \ref{lem:gammac}.
\end{proof}

\subsection{Proof of part (ii) of Proposition \ref{prop:Eimpliesgap}}
Now we assume Property (E') instead of Property (E). The only estimate for the proof
of which we directly used Property (E) is estimate \eqref{Hcgap}, which implies \eqref{est:Habgap}.
With Property (E') alone \eqref{Hcgap} is no longer valid, but we shall change our strategy in order
to obtain a variant of \eqref{est:Habgap}. Once we achieve this, the rest of the proof remains
unchanged. At first we discuss the additional difficulties that we have to face.
 Each time we were moving an electron from a non-negative ion to a non-positive one, we were gaining a gap that was
only $Z$ dependent. This argument can not be used in the case of Property (E'), as we do not gain
a gap in every single step. Moreover, Property (E) was formulated in terms of pairs of atoms
but Property (E') in terms of the entire system. So it could be in principle that the minimum gap
of Property (E') depends on $M$. To deal with this problem we will proceed as follows: we consider a group of ions of total charge zero, with each of them
having charge no less than $-Z$.
As we shall see in the proof, the assumption that each charge is not less than $-Z$ is not restrictive, because an ion with charge less than $-Z$ does not have an isolated ground state energy.
 We will show that
 such a group of ions, can be partitioned into subgroups with the following properties:  each of the subgroups consists
 of less than $Z^2+2\delta_{Z,1}$ ions and its ions have total charge zero.
Here $\delta_{Z,1}$ is the Kronecker $\delta$.
Then each of these subgroups will give us a gap that depends only on $Z$. By adding the gaps we will obtain a variant of  \eqref{est:Habgap}.

 To this end we need the following auxiliary lemma:
\begin{lemma}\label{lem:dirichlet}
Suppose that $k_1,...,k_Z$ are integers. Then there exists
$\al_1,\dots,\al_Z \in\{0,1\}$ not all of them zero, so that $\sum_{i=1}^Z \al_i k_i$
is a multiple of $Z$.
\end{lemma}
\begin{proof}
The lemma is standard, but we shall give the proof for convenience of the reader. We consider
the numbers $K_j=k_1+\dots + k_j, j \in \{1,\dots,Z\}$. Of course if one of these numbers is a multiple of $Z$
then the conclusion holds. If not then $K_1,\dots,K_Z$ are $Z$ numbers but the remainders $ K_1 (\text{mod} Z),\dots,K_Z (\text{mod} Z)$
are at most $Z-1$ numbers, because none of them is zero. Therefore, there exist $j_1, j_2$ with $j_1<j_2$ so that $K_{j_1} (\text{mod} Z)=K_{j_2} (\text{mod} Z)$. It follows that $K_{j_2}-K_{j_1}$ is a multiple of $Z$.
\end{proof}
Now we consider a group of ions of total charge zero, with each of them
having charge no less than $-Z$.
The following Lemma will enable us to split it to subgroups of less
 than $Z^2+2 \delta_{Z,1}$ ions of total charge zero.
\begin{lemma}\label{lem:groupbreaking}
Let $m \in \mathbb{N}$, and $n_1,\dots,n_m \in \mathbb{Z}/\{0\}$ with  $-Z \leq n_j \leq Z$ for all $j \in \{1,\dots,m\}$.
We assume that $\sum_{j=1}^m n_j=0$ and moreover that $\sum_{j \in S} n_j \neq 0$ for all nonempty sets
$S \varsubsetneq \{1,\dots,m\}$. Then $m < Z^2+2 \delta_{Z,1}$.
\end{lemma}
\begin{remark}
As suggested from the notation, $n_j$ will play the role of the charges of ions.
 The bound $Z^2$ for $Z>1$ is definitely not optimal.
We believe that the sharp bound is $m \leq 2Z-1$. Such a bound would be sharp because if $n_1,\dots,n_Z=Z-1$
and $n_{Z+1},\dots,n_{2Z-1}=-Z$, then the assumptions of the lemma are fulfilled. We have been, however,
unable to prove that $2Z-1$ is a bound.
\end{remark}
\begin{proof}[Proof of Lemma \ref{lem:groupbreaking}]
For $Z=1$ and $Z=2$ the proof is trivial, so we will prove the lemma for $Z>2$.
   Assume that there is a finite sequence $n_1,\dots,n_{Z^2}$ satisfying the assumptions of
Lemma \ref{lem:groupbreaking}. Using the assumptions of the lemma, it is easy to show that
at least $Z$ terms of this sequence have to be positive and similarly
at least $Z$ terms have to be negative.
Furthermore, if a number $d$ appears in the sequence $n_j, j \in \{1,\dots,Z^2\}$, then $-d$ can not
appear because of the assumptions of the lemma. As a consequence, at most $Z$ different numbers can appear
in the sequence. Since the sequence $n_j$ has $Z^2$ terms, one number has to appear at least $Z$ times. We may assume,
without loss of generality, that this number is negative and we denote it by $-q$. We consider $q$ positive elements
$n_{j_1},\dots,n_{j_q}$ of the sequence. By Lemma \ref{lem:dirichlet} there are $\al_{j_1},...,\al_{j_q} \in \
\{0,1\}$ not all zero, so that $\sum_{m=1}^q \al_{j_m} n_{j_m}$ is a multiple of $q$. Since all
elements $n_{j_1},\dots,n_{j_q}$ are positive and not bigger than $Z$ we obtain that
$\sum_{m=1}^q \al_{j_m} n_{j_m}=kq$, where $0< k \leq Z$. Then
$\underbrace{-q+...+(-q)}_{k \text{ times }}+\sum_{m=1}^q \al_{j_m} n_{j_m}=0$, contradicting the
assumption $\sum_{j \in S} n_j \neq 0$ for all nonempty sets
$S \varsubsetneq \{1,\dots,Z^2\}$, because the sum consists of at most $2Z<Z^2$ terms. This concludes the proof of Lemma \ref{lem:groupbreaking}.
\end{proof}
For each subset $F \subset\{1,\dots,M\}$ we define
\begin{equation}
D_F= \min\{\sum_{i \in F} E_{i,n_{i}}-\sum_{i \in F} E_{i,0}|, -Z \leq n_i \leq Z_i, n_i \in \mathbb{Z},
 \forall i \in F, \sum_{i \in F}n_i=0, \sum_{i \in F}|n_i| \neq 0\}.
\end{equation}
and
\begin{equation}
D_4'=\frac{1}{2 Z^3} \min_{F \subset \{1,\dots,M\},|F| < Z^2+2 \delta_{Z,1}} D_F.
\end{equation}
From Property (E') it follows that $D_4'>0$. From the restriction $|F| < Z^2+2 \delta_{Z,1}$ it follows that
$D_4'$ depends only on $Z$. Now we are ready to prove a variant of \eqref{est:Habgap},
for decompositions with ions each of them having charge no less than $-Z$.
\begin{lemma}\label{lem:Habgap'}
Let $a, b, l$ be as in Lemma \ref{lem:bca}. Assume that the decomposition $b=\{ B_1,\dots, B_M\}$ has the
property that $|B_i| \leq Z_i+Z$, $\forall i \in \{1,\dots,M\}$. Then
\begin{equation}\label{est:Habgap'}
\inf \s(H_b^\s)- \inf \s(H_a^\s) \geq  2 l D_4'.
\end{equation}
\end{lemma}
\begin{proof}
For $Z=1$ the proof is trivial, so we will show the Lemma in the case
$Z>1$. We write $l=m Z^3+n$, where $m, n \in \mathbb{N} \cup \{0\}$ with $n<Z^3$. Given the restrictions
of the charges of ions of the decomposition $b$, one can verify that after $l$ steps (where by step is meant going
from $c_m$ to $c_{m+1}$) at least
$2m Z^2$ atoms are ionized. Using Lemma \ref{lem:groupbreaking} it follows that we can break these ions to at least
$2m+1$ groups of ions with the following properties: each group has less than $Z^2$ ions and the sum of the charges of the ions of each group is  zero. By definition of
$D_4'$, each of these $2m+1$ groups is giving at least a gap $ 2 Z^3 D_4'$. Therefore we obtain that
$\inf \s(H_b^\s)- \inf \s(H_a^\s) \geq (2 m+1) 2 Z^3 D_4'$, which together with the equation $l=m Z^3+n$ implies
\eqref{est:Habgap'}.
\end{proof}
Of course we need \eqref{est:Habgap'} to hold for all $b \in \mathcal{A}$ without the restrictions
$|B_i| \leq Z_i+Z$. In other words we need to prove
\begin{lemma}\label{lem:last}
Let $a, b, l$ be as in Lemma \ref{lem:bca}. Then \eqref{est:Habgap'} holds.
\end{lemma}
\begin{proof}
We construct the sequence $c_0,\dots, c_l$ as follows. We go from $c_m$ to $c_{m+1}$
so that $c_{m+1}$ fulfills the Properties of $b$ in Lemma \ref{lem:Habgap'} until
it is no longer possible to do this. In other words we do not let any ion
 have charge  less than $-Z$ until we have no choice to do this. Let $c_{m_0}$ the last
decomposition that fulfills the Properties of $b$ in Lemma \ref{lem:Habgap'}. Then by Lemma
\ref{lem:Habgap'} we obtain that
\begin{equation}\label{est:m0}
\inf \s(H_{c_{m_0}}^\s)- \inf \s(H_a^\s) \geq  2 m_0 D_4'.
\end{equation}
We shall now show that
\begin{equation}\label{est:step}
\inf \s(H_{c_{m+1}}^\s)- \inf \s(H_{c_m}^\s) \geq 2 D_5, \forall m \geq m_0,
\end{equation}
where $D_5$ was defined \eqref{def:D5}. Let $m \geq m_0$. With the notation of Lemma \ref{lem:bca},
we have that $|C_{m,j}| \geq Z_j+Z$. In other we go from $c_m$ to $c_{m+1}$
by transferring an electron to an ion with charge already less or equal to $-Z$.
 By \cite{Lieb} this implies that $H_{C_{m+1,j}}^\s$ has no discrete spectrum and therefore, by Theorem \ref{hvz}, we obtain that
\begin{equation}
E_{j,Z_j-|C_{m,j}|}=E_{j, Z_j-1-|C_{m,j}|}.
\end{equation}
The last equality together with \eqref{eqn:stepgap} implies that
\begin{equation}
\inf \s(H_{c_{m+1}}^\s)-\inf \s(H_{c_m}^\s)=E_{i,Z_i+1-|C_{m,i}|}-E_{i,Z_i-|C_{m,i}|}.
\end{equation}
Estimate \eqref{est:step} follows from the definition of $D_5$. One can easily verify
that $D_5 \geq  D_4'$ and therefore  \eqref{est:m0} and \eqref{est:step} imply \eqref{est:Habgap'}.
This concludes the proof of Lemma \ref{lem:last}.
\end{proof}


\begin{thebibliography}{DGMS}
\bibitem[AS]{AS} I. Anapolitanos, I.M. Sigal: Long Range behavior of van der Waals force.
arXiv:1205.4652v2.


\bibitem[BFS]{BFS} V. Bach, J. Fr\"ohlich and I.M. Sigal:
Quantum Electrodynamics of Confined Nonrelativistic Particles. {\it Adv. in Math.} {\bf 137}, 299-395  (1998).


\bibitem[CFKS]{CFKS} H.L. Cycon, R.G. Froese , W. Kirsch and B. Simon :
 Schr\"odinger Operators with application to quantum mechanics and
global geometry. {\it Texts and Monographs in Physics. Springer
study edition. Springer-Verlag Berlin} (1987).

\bibitem[FS]{FS} G. Feinberg and J. Sucher: General theory of the van
der Waals interaction: a model independent approach. {\it Phys. Rev.
A} {\bf 9}, 2395-2415  (1970).

\bibitem[BFGR]{BFGR}  J. Fr\"ohlich, G.M. Graf, J.-M. Richard and M. Seifert: Proof of stability of the hydrogen molecule. {\it Phys. Rev. Lett.}, {\bf 71}, No.9, 30  1332-1334 (1993).

\bibitem[C]{C} SJ. Cha, YG. Choe, UG. Jong, GC. Ri, CJ. Yu: Refined phase coexistence line between graphite and diamond from  density-functional Theory and van der Waals correction. {\it Physica B: Condensed Matter}
    {\bf 434} 185-193 (2014).

\bibitem[CT]{CT} J.M. Combes, L. Thomas: Asymptotic behavior of eigenfunctions for multiparticle Schr\"odinger
operators. {\it Commun. Math. Phys.} {\bf 34}, 251-270 (1973).

\bibitem[H]{H} W. Hunziker: On the spectra of Schr\"odinger Multiparticle Hamiltonians.
{\it Helv. Phys. Acta} {\bf 39}, 451-462 (1966).

\bibitem[HS]{HS} W. Hunziker and I.M. Sigal: The quantum $N-$body
problem. {\it J. Math. Phys.} {\bf 41} No.6, 3448-3510 (2000).

\bibitem[J1]{J1} J.E. Jones: On the Determination of Molecular Fields. I.
From the Variation of the Viscosity of a Gas with Temperature. {\it Proc. R. Soc. Lond. A} {\bf 106}(738), 441-462 (1924).

\bibitem[J2]{J2} J.E. Jones: On the Determination of Molecular Fields. II.
From the Equation of State of a Gas. {\it Proc. R. Soc. Lond. A} {\bf 106}(738), 463-477 (1924).

\bibitem[Le]{Le} M. Lewin: A Mountain Pass for reacting molecules. {\it Ann. Henri Poincare} {\bf 5} 477-521 (2004).

\bibitem[Lieb]{Lieb} E.H. Lieb: Bound on the maximum negative ionization of atoms and molecules.
 {\it Phys. Rev. A} {\bf 29}, 3018-3028 (1984).

\bibitem[LL]{LL} E.H. Lieb and M. Loss:
 Analysis. {\it Graduate Studies in Mathematics} {\bf
 14} AMS, Providence, RI, second edition (2001).

\bibitem[LSST]{LSST} E. H. Lieb, I. M. Sigal, B. Simon, and W. Thirring:
 Asymptotic neutrality of large-Z ions. {\it Commun. Math. Phys.} {\bf 116}, 635-644 (1988).

\bibitem[LT]{LT}  E.H. Lieb and W. Thirring : Universal nature of van der Waals forces for Coulomb
systems. {\it Phys. Rev. A} {\bf 34} No.1, 40-46 (1986).

\bibitem[Lo]{Lo} F. London: The general theory of molecular forces.
Transactions of the Faraday Society {\bf 33}: 8ÃƒÂ¢Ã¢â€šÂ¬Ã¢â‚¬Å“26 (1937).

\bibitem[MS]{MS} T. Miyao, H. Spohn: The retarded van der Waals potential: Revisited. J. Math. Phys. {\bf 50}, 072103 (2009).

 \bibitem[Phan]{Phan} Phan Th\`anh Nam. New bounds on the maximum ionization of atoms. {\it Commun. Math. Phys.}
 {\bf 312}, 427-445 (2012).

\bibitem[RSI]{RSI} M. Reed and B. Simon: Methods of modern Mathematical Physics I: Functional Analysis. {\it Academic Press Inc.} (1980).

\bibitem[RSIV]{RSIV} M. Reed and B. Simon: Methods of modern Mathematical Physics IV:
 Analysis of Operators. {\it Academic Press Inc.} (1980).


\bibitem[Sig]{Sig} I.M. Sigal: Geometric methods in the quantum many-body problem.
 Non-existence of very negative ions. {\it Comm. Math. Phys.} {\bf 85}, 309-324 (1982).

\bibitem[Sig2]{Sig2} I.M. Sigal: How many electrons can a nucleus bind? {\it Ann. Phys.} {\bf 157} No.2, 307-320 (1984).

\bibitem[So]{Sol2}  J.P. Solovej: The ionization conjecture in Hartree-Fock theory. {\it Ann. of Math.} {\bf 158}, 509-576 (2003).

\bibitem[vdW1]{vdW1} J.D. van der Waals: On the continuity of the Gaseous amd Liquid states. Edited and with
an introduction by J.S. Rowlison. {\it Dover Phoenix Editions.} (1988).

\bibitem[vdW2]{vdW2} J.D. van der Waals: On the continuity of the Gaseous amd Liquid states. Nobel lecture (1910).

\bibitem[vW]{vW} C. Van Winter: Theory of Finite systems of Particles. I. The Green function.
{\it Mat.-Fys. Skr. Dankse Vid. Selsk.} { \bf 2} No.8 (1964).

\bibitem[Zh]{Zh} G.M. Zhilin: Discussion of the Spectrum of Schr\"odinger operators for systems of many particles (In Russian).
{\it Trudy Moskovskogo matematiceskogo obscestva.} {\bf 9}, 81-120 (1960).

\end{thebibliography}
\end{document}